\documentclass[11pt,dvipsnames]{article}

\usepackage[utf8]{inputenc}
\usepackage{amsmath,amsthm,amsfonts,amssymb,mathtools,bbm,mathdots,tensor,stmaryrd}

\SetSymbolFont{stmry}{bold}{U}{stmry}{m}{n}

\usepackage[sans]{dsfont}
\usepackage[mathscr]{eucal}
\usepackage{authblk}
\usepackage{enumitem}
\usepackage{slashed,extarrows,bigints,cancel}
\usepackage[new]{old-arrows}
\usepackage{graphicx}
\usepackage{float,caption,multirow,makecell,nicematrix}
\usepackage{bm}
\usepackage[]{xcolor}
\usepackage[]{hyperref}
\hypersetup{colorlinks = true, linktoc=page, urlcolor=MidnightBlue, citecolor=red, linkcolor=ao!90!black}
\definecolor{ao}{rgb}{0.0, 0.5, 0.0}
\usepackage[hyperpageref]{backref}

\usepackage{tikz}
\usepackage{pgfplots}
\usetikzlibrary{decorations.markings, decorations.pathmorphing, arrows,calc,knots,hobby, decorations.pathreplacing, shapes.geometric, calc, arrows, arrows.meta,patterns,intersections, pgfplots.fillbetween,arrows.meta}

\tikzset{->-/.style={decoration={
			markings,
			mark=at position .58 with {\arrow{>[scale=2]}}},postaction={decorate}}}
\usepackage{tikz-cd}

\usepackage[top=20mm, left=29mm, right=29mm]{geometry}
\usepackage{ascmac}
\usepackage[font={footnotesize},labelfont={bf}]{caption}
\usepackage[sort,compress]{cite}

\newenvironment{eqaligned}
{%
\begin{equation}
    \begin{aligned}
    } 
{%
\end{aligned}
\end{equation}
\ignorespacesafterend}

\newenvironment{eqaligned*}
{%
\begin{equation*}
    \begin{aligned}
    } 
{%
\end{aligned}
\end{equation*}
\ignorespacesafterend}

\newenvironment{eqgathered}
{%
\begin{equation}
    \begin{gathered}
    } 
{%
\end{gathered}
\end{equation}
\ignorespacesafterend}

\newenvironment{eqgathered*}
{%
\begin{equation*}
    \begin{gathered}
    } 
{%
\end{gathered}
\end{equation*}
\ignorespacesafterend}

\newcommand{\mfk}[1]{\mathfrak{#1}}

\newcommand{\mbb}[1]{\mathbb{#1}}

\newcommand{\msf}[1]{\mathsf{#1}}
\newcommand{\mscr}[1]{\mathscr{#1}}
\newcommand{\mcal}[1]{\mathcal{#1}}

\newcommand{\wt}[1]{\widetilde{#1}}

\newcommand*\ovl[1]{%
  \vbox{%
    \hrule height 0.8pt%
    \kern 0.15ex%
    \hbox{%
      \kern 0.0em%
      \ifmmode#1\else\ensuremath{#1}\fi%
      \kern 0.05em%
    }%
  }%
}

\newcommand{\A}{\msf{A}}
\newcommand{\B}{\msf{B}}
\newcommand{\bpsi}{\ovl{\psi}}
\newcommand{\bpsii}[1]{\bpsi^{\,#1}}
\newcommand{\R}{\mbb{R}}
\newcommand{\pois}[2]{\big\{#1,#2\big\}_\star}
\newcommand{\poisb}[2]{\big\{#1,#2\big\}_{\star,\msf{b}}}

\newcommand{\wh}[1]{\widehat{#1}}

\numberwithin{equation}{section}

\def\ben{\begin{eqnarray*}}
\def\een{\end{eqnarray*}}

\newcommand{\C}{\mathbb{C}}
\newcommand{\GL}{\mathrm{GL}}
\newcommand{\SL}{\mathrm{SL}}

\newcommand{\bC}{\mathbb{C}}

\newcommand{\bN}{\mathbb{N}}

\newcommand{\bR}{\mathbb{R}}

\newcommand{\bZ}{\mathbb{Z}}

\newcommand{\cA}{\mathcal{A}}

\newcommand{\cM}{\mathcal{M}}

\newcommand{\cO}{\mathcal{O}}

\newcommand{\fgl}{\mathfrak{gl}}

\newcommand{\fsl}{\mathfrak{sl}}

\newcommand{\Gr}{\operatorname{Gr}}

\newcommand{\Hom}{\operatorname{Hom}}
\newcommand{\End}{\operatorname{End}}
\newcommand{\Tr}{\operatorname{Tr}}
\newcommand{\Id}{\operatorname{Id}}
\newcommand{\im}{\operatorname{im}}

\newcommand{\Proj}{\operatorname{Proj}}

\newcommand{\Spec}{\operatorname{Spec}}

\newcommand{\tr}{\operatorname{tr}}

\theoremstyle{plain}

  \newtheorem{theorem}{Theorem}[section]
  \newtheorem{proposition}[]{Proposition}[section]
  \newtheorem{lemma}[]{Lemma}[section]
  \newtheorem{corollary}[]{Corollary}[section]

\theoremstyle{definition}
  \newtheorem{definition}{Definition}[section]

  \newtheorem{example}[]{Example}[section]
  \newtheorem{remark}[]{Remark}[section]

\numberwithin{equation}{section}

\newlength{\bibitemsep}
\newlength{\bibparskip}
\let\oldthebibliography\thebibliography
\renewcommand\thebibliography[1]{%
  \oldthebibliography{#1}%
  \setlength{\parskip}{\bibitemsep}%
  \setlength{\itemsep}{\bibparskip}%
}

\makeatletter
\renewcommand{\paragraph}{%
  \@startsection{paragraph}{4}%
  {\z@}{2.25ex \@plus 1ex \@minus .2ex}{-1em}%
  {\normalfont\normalsize\bfseries}%
}
\makeatother

\allowdisplaybreaks
\numberwithin{equation}{section}
\pgfplotsset{compat=1.18}

\newcommand{\TheSetup}{
\begin{tikzpicture}[scale=.8]

		\path[line width=2pt, rounded corners=4pt,pattern=crosshatch dots gray	, pattern color=gray, opacity=.4] (1,1) -- (-2,-2) -- (6,-2) -- (9,1) -- (-1,1);
		
		\draw[line width=3pt,Green] (5,1) -- (2,-2);
		
		\draw[fill] (4,0) circle (3pt);
		\draw[fill] (3,-1) circle (3pt);

	    \node[rotate=45] at (-.7,-.2) {$\mfk{gl}_N$ BF theory};

	    \draw[line width=1pt,-stealth] (7,2.5) to [out=270,in=90] (5,+1.1);
	    \node[align=center] at (7,3.2) {the line defect \\ along $\mbb{R}_x$};

	    \draw[line width=1pt,-stealth] (1,2.5) to [out=270,in=90] (3.9,.2);
	    \draw[line width=1pt,-stealth] (1,2.5) to [out=270,in=90] (3,-.8);
	    \node[align=center] at (1,2.8) {local operators $\mathcal{O}^i_j[n]$};

        \draw[fill] (6,-1) circle (3pt);

        \draw[line width=1pt,-stealth] (9,1.5) to [out=180,in=90] (6,-.8);
        \node[align=center] at (10.5,1.6) {local operators \\ $\mathcal{O}[n]$};
\end{tikzpicture}
}

\title{\bf\Large\centering 2d BF Theory Coupled to 1d Quantum Mechanics:\\ The Phase Space and Its Quantization}

\author[a]{Seyed Faroogh Moosavian\footnote{\href{mailto:sfmoosavian@gmail.com}{\texttt{sfmoosavian@gmail.com}}}}
\author[b,c,d,e]{Yehao Zhou\footnote{\href{mailto:yehaozhou1994@gmail.com}{\texttt{yehaozhou1994@gmail.com}}}}

\affil[a]{\small Department of Physics, McGill University, Montr\'eal QC H3A 2T8, Canada}
\affil[b]{\small Perimeter Institute for Theoretical Physics, Waterloo ON N2L 2Y5, Canada}
\affil[c]{\small University of Waterloo, Waterloo ON N2L 3G1, Canada}
\affil[d]{Center for Mathematics and Interdisciplinary Sciences, Fudan University, Shanghai 200433, China}
\affil[e]{Shanghai Institute for Mathematics and Interdisciplinary Sciences, Block A, International Innovation Plaza, No. 657 Songhu Road, Yangpu District, Shanghai 200433, China}

\date{}

\begin{document}

\maketitle

\begin{abstract}

    We study the ring of functions on the (classical and quantized) phase space of 2-dimensional BF theory with the gauge group $\mathrm{GL}_N$ coupled to a 1-dimensional quantum mechanics with global symmetry $\mathrm{GL}_K$. These functions are gauge-invariant local observables of the coupled system. We first construct the classical phase space of this system and describe its ring of functions, as well as their large-$N$ limit. We next compute the Hilbert series of these algebras for finite $ N$ and in the large-$N$ limit. We then study the quantization of this phase space and the deformation quantization of its ring of functions, elaborate its relation to the Yangian, and construct its coproduct. Finally, we identify these quantized algebras with the quantized Coulomb branch algebras of certain 3d $\mscr{N}=4$ quiver gauge theories.

\end{abstract}

\tableofcontents

\section{Introduction}
\label{sec:introduction}
Holography is one of the main active areas of research in finding a theory of quantum gravity \cite{tHooft199310,Susskind199409}. The prime example of this concept is the AdS/CFT Correspondence \cite{Maldacena199711,Witten199802}.

Recently, Costello and Li have formulated a twisted version of the AdS / CFT correspondence \cite{Costello201610,Costello201705}. According to this framework, the holography can be understood as a certain algebraic relation, known as Koszul duality, between the algebra of operators in the two sides of the correspondence (see \cite{Costello201303} for an earlier example and also \cite{IshtiaqueMoosavianZhou201809,CostelloGaiotto1821,GaiottoOh201907,RaghavendranYoo201910,LiTroost201911,CostelloPaquette202001,OhZhou202002,GaiottoAbajian202004,LiTroost202005,OhZhou202103,OhZhou202105,BudzikGaiotto202106,GaiottoLee202109,CostelloLi201606,CostelloLi201905,CostelloWilliams202110} for follow-up and related works). For a recent and very readable review of Koszul duality aimed at physicists, we refer the reader to \cite{PaquetteWilliams202110}. An instance of this twisted version has been studied in \cite{IshtiaqueMoosavianZhou201809}, where it was shown that the algebra of local operators in $2d$ BF theory with gauge group $\GL_N$ coupled to a $1d$ fermionic quantum mechanics with global symmetry $\GL_K$ (the boundary side) and the algebra of scattering states computed using Witten digrams of $4d$ Chern-Simons theory with gauge group $\GL_K$ (the bulk side) match and in the large-$N$ limit approach the Yangian (see \cite[Theorem 1]{IshtiaqueMoosavianZhou201809}).

In this paper, we study a closely-related system, which is $2d$ BF theory with gauge group $\GL_N$ coupled to a $1d$ \textit{bosonic} quantum mechanics with global symmetry $\GL_K$. We will show that the quantized algebra of gauge-invariant local operators in this coupled theory is a certain truncation of the Yangian of $\fgl_K$, and it approaches $Y_\hbar(\fgl_K)$ in the $N\to\infty$ limit, see Theorem \ref{thm:main} below.

The action functional of this system is given by
\begin{equation}\label{eq:action}
    S=S_{\text{BF}}+S_{\text{QM}}=\frac{1}{2\pi}\bigintsss_{\mathbb R^2_{x,w}}\text{Tr}(B F_A)+\bigintsss_{\mathbb R_{x}\times\{w=0\}}\sum_{a=1}^KI_a(\partial_x+A)J_a.
\end{equation}
Here $A$ is the $\GL_N$ connection on $\bR^2$, $F_A$ is the curvature form associated to $A$, $B$ is a scalar field on $\bR^2$ valued in $\End(\bC^N)$, $I_a$ is a scalar field on $\bR_{x}\times \{w=0\}$ valued in $(\bC^N)^*$, and $J_a$ is a scalar field on $\bR_{x}\times \{w=0\}$ valued in $\bC^N$. We would like to analyze the geometry of the phase space of the system \eqref{eq:action} and its quantization. One of the reasons we are studying the problem from the phase space perspective is that geometric tools allows us to make statements at finite $N$.

We choose the gauge $A=0$. Then the equations of motion are
\begin{align}\label{EOM}
\partial_x J_a=\partial_x I_a=0,\quad \partial_x B=0,\quad \partial_w B-\sum_{a=1}^KJ_aI_a\delta_{w=0}=0.
\end{align}
The solution is that $J_a$ and $I_a$ are constant along the line defect, $B$ is constant on the regions $w<0$ and $w>0$, and 
\begin{align}\label{moment map}
    B_{w>0}-B_{w<0}=\sum_{a=1}^KJ_aI_a.
\end{align}
So the phase space, denoted by $\cM(N,K)$, is parametrized by $B_{w>0}$, $J_a$ and $I_a$, modulo the $\GL_N$ action. This is the quiver variety associated to the framed quiver in Figure \ref{Fig: quiver}, i.e.
\begin{align}\label{eq: M(N,K)_def}
    \mathcal M(N,K)\cong \mathrm{Rep}(N,K)\sslash \GL_N,
\end{align}
where $\mathrm{Rep}(N,K):=\End(\bC^N)\oplus\Hom(\bC^N,\bC^K)\oplus \Hom(\bC^K,\bC^N)$ is the linear space of representations of quiver in Fig. \ref{Fig: quiver}. We study the geometry of $\mathcal M(N,K)$ in \S\ref{sec:geometry of the phase space}.

\begin{figure}[h!]
\centering
\begin{tikzpicture}[scale=.6,decoration={
    markings,
    mark=at position 0.55 with {\arrow{>}}}
    ] 
			\draw[line width=1pt] (0,0) circle (.5 cm)node[]{$N$};
			\draw[line width=1pt,postaction={decorate}] (-1.65,-2.2)node[yshift=+.8cm]{$J$} -- (-.35,-.35);
			\begin{scope}[xscale=-1]
			    \draw[line width=1pt,postaction={decorate}] (-.35,-.35) -- (-1.65,-2.2)node[yshift=+.8cm]{$I$};
			    \draw[line width=1pt] (-2.15,-2.2) rectangle ++(1,-1) node[xshift=+.3cm,yshift=+.3cm]{$K$};
			\end{scope}
			\draw[line width=1pt] (-2.15,-2.2) rectangle ++(1,-1) node[xshift=-.3cm,yshift=+.3cm]{$K$};
			\draw[line width=1 pt,looseness=7,postaction={decorate}] (-.35,.35) to [out=130,in=50] (.35,.35)node[xshift=-.2cm,yshift=1cm]{$B$};
	\end{tikzpicture}
	\caption{The quiver description of the fields contents.}\label{Fig: quiver}
\end{figure}

\subsection*{Summary of the results}

The basic logic of the paper is to first study the classical phase space $\mcal{M}(N,K)$ and its ring of functions $\mbb{C}[\mcal{M}(N,K)]$, and finally their large-$N$ limit. We then study modules for these algebras. Then, we considered the quantization of the classical phase space and the deformation quantization $\mbb{C}_\hbar[\mcal{M}(N,K)]$ of its ring of functions, which leads to the algebra. We study its structure, especially its coproduct and its identification with the Coulomb branch algebra of $3d$ $\mscr{N}=4$ theories. 

\smallskip Our main result is the following.

\begin{theorem}[{Corollary \ref{cor:M(N,K) as Coulomb branch}, Proposition \ref{prop:the surjective map between Yangian and deformed ring of functions}}]\label{thm:main}
There is an isomorphism between Poisson varieties
\begin{align*}
\cM(N,K)&\cong \text{Coulomb branch associated to the quiver in Figure \ref{fig:quiver for the Coulomb branch description_intro} with flavor symmetry}\\
&\cong \text{Beilinson-Drinfeld slice $\overline{\mathcal W}^{\underline{\lambda^*}}_0$ for $\GL_K$},
\end{align*}
where $\lambda=N\omega_1+N\omega_{K-1}$, and $\omega_i$ is the $i$-th fundamental coweight for $\GL_K$. There is an isomorphism between $\bC[\hbar]$ algebras
\begin{align*}
\bC_\hbar[\cM(N,K)]&\cong \text{quantized Coulomb branch algebra associated to the quiver in Figure \ref{fig:quiver for the Coulomb branch description_intro}}\\
&\quad \text{ with flavor symmetry}\\
&\cong \text{truncated Yangian $\mathbf Y^{\lambda}_0$}
\end{align*}
Moreover, there exists a surjective $\bC[\hbar]$ algebra homomorphism $$\rho_N:Y_\hbar(\fgl_K)\otimes \Lambda \twoheadrightarrow \bC_\hbar[\cM(N,K)]$$ such that $\bigcap_{N=1}^\infty \ker(\rho_N)=0$. Here $\Lambda:=\bC[p_i:i\in \bZ_{\ge 1}]$ is the polynomial ring of countably many generators, which is naturally identified with ring of symmetric polynomials with $p_i$ being the $i$-th power sum function.
\end{theorem}

We emphasize that the isomorphism between truncated Yangian $\mathbf Y^{\lambda}_0$ and quantized ring $\bC_\hbar[\cM(N,K)]$ is explicit, see the proof of Theorem \ref{theorem: quantum ideal}. We also provide a more conceptual (non-explicit) proof of the isomorphism $\mathbf Y^{\lambda}_0\cong \bC_\hbar[\cM(N,K)]$ using the tool of ring objects in the derived category of constructible sheaves on affine Grassmannian \cite{braverman2017ring}, see Remark \ref{rmk:another proof}.

\begin{figure}[h!]
    \centering
    \begin{tikzpicture}[scale=.5]
    
        \draw[line width=1pt] (-6,0) circle (.75cm);
        \draw[line width=1pt] (-5.25,0) -- (-4.25,0);
        \draw[line width=1pt] (-3.5,0) circle (.75cm);
        \draw[line width=1pt] (-2.75,0) -- (-1.75,0);
        \draw[line width=1pt,loosely dotted] (-1.75,0) -- (1.75,0);
        \draw[line width=1pt] (-6,-.75) -- (-6,-2.05);
        \draw[line width=1pt] (-6.65,-2.05) rectangle (-5.35,-3.35);
        
        \node at (-6, 0) {$N$};
        \node at (-6,-2.7) {$N$};
        \node at (-3.5,0) {$N$};

        \begin{scope}[xscale=-1]
        
        \draw[line width=1pt] (-6,0) circle (.75cm);
        \draw[line width=1pt] (-5.25,0) -- (-4.25,0);
        \draw[line width=1pt] (-3.5,0) circle (.75cm);
        \draw[line width=1pt] (-2.75,0) -- (-1.75,0);
        \draw[line width=1pt] (-6,-.75) -- (-6,-2.05);
        \draw[line width=1pt] (-6.65,-2.05) rectangle (-5.35,-3.35);

        \draw [line width=1pt, decorate,decoration={brace,amplitude=10pt,mirror,raise=.75cm}]
        (-6.8,0) -- (6.8,0) node[midway,yshift=1.4cm]{$K-1$ gauge nodes};

        \node at (-6, 0) {$N$};
        \node at (-6,-2.7) {$N$};
        \node at (-3.5,0) {$N$};
        \end{scope}

    \end{tikzpicture}
    \caption{The quiver whose Coulomb branch isomorphic to $\cM(N,K)$.}
    \label{fig:quiver for the Coulomb branch description_intro}
\end{figure}

\subsection*{Structure of the paper}
This paper is organized as follows.

\smallskip In \S\ref{sec:geometry of the phase space}, we investigate the geometry of the phase space of BF theory coupled to our quantum-mechanical system and study the algebra of functions in this phase space. The main results of this section are the following

\begin{enumerate}
    \item The first result is concerned with the structure of the phase space; we show that $\mathcal M(N,K)$ is a normal affine variety of dimension $2NK$. This is shown in Proposition \ref{prop:phase space is a normal affine variety}. 
    
    \item $\mbb{C}[\mcal{M}(N,K)]$, the algebra of functions on $\mcal{M}(N,K)$ is generated by the set $\{\Tr(B^n),I_aB^mJ_b:1\le n\le N,0\le m\le N-1,1\le a,b\le K\}$. Note that the operators $\text{Tr}(B^n)$ are dual to the gravitons while determinant $\det(B^n)$ and subdeterminant operators are dual to giant gravitons in the bulk. We then find that $\mcal{M}(N,1)\simeq\mbb{A}^{2N}$.

    \item Next, we define the following Poisson structure on $\mathcal{M}(N,K)$ by \begin{equation}\label{eq:Poisson bracket on the phase space (introduction)}
        \{J_{ia},I_{bj}\}=\delta_{ab}\delta_{ij},\quad 
        \{B_{mn},B_{pq}\}=\delta_{pn}B_{mq}-\delta_{mq}B_{pn},\quad
    \{B_{mn},I_{bj}\}=\{B_{mn},J_{ia}\}=0,
    \end{equation}
    and then define
    \begin{equation}
        T^{(n)}_{ab}:= I_aB^nJ_b=I_{ai_1}B_{i_1 i_2}B_{i_2 i_3}\cdots B_{i_n i_{n+1}}J_{i_{n+1}b},
    \end{equation}
    with the convention $T^{(-1)}_{ab}=\delta_{ab}$. Then we will see that
    \begin{equation}\label{eq:Poisson bracket in terms of T}
    \{T^{(p)}_{ab},T^{(q)}_{cd}\}=\sum_{i=-1}^{\min(p,q)-1}\left(T^{(p+q-1-i)}_{cb}T^{(i)}_{ad}-T^{(i)}_{cb}T^{(p+q-1-i)}_{ad}\right).
    \end{equation}
    In fact, this is the classical limit of the RTT relation \eqref{eq: RTT_intro}.

    \item A related side result of this section, whose details are expounded in Appendix \ref{sec:Relationship between two Poisson structures}, is the relationship between the natural Poisson brackets \eqref{eq:Poisson bracket on the phase space (introduction)} and the one inspired by the twisted holography computations in \cite{IshtiaqueMoosavianZhou201809}. We first derive the latter  (see Theorem \ref{thr:Poisson structure inspired by twisted holography} and the Poisson brackets \eqref{eq:Poisson structure of elementary fields from bosonic quantum mechanics}) from the classical limit of the algebra of gauge-invariant operators of the 2d BF theory coupled to a 1d quantum mechanics in loc. cit. We then show that they are related by a field redefinition (see relation \eqref{eq:relation between the natural Poisson brackets vs the one inspired by the twisted holography} and field redefinition \eqref{eq:redefinition of B field}).
\end{enumerate}

\smallskip In \S\ref{sec:large-N limit}, we consider the large-$N$ limit of $\mcal{M}(N,K)$ and its ring of functions $\mbb{C}[\mcal{M}(N,K)]$. The main results of this section as as follows
\begin{enumerate}
    \item We show that $\bigcup_N \mathcal{M}(N,K)$ is Zariski-dense in $L^-\GL_K\times L^-\GL_1$, where $L^-G$ is the negative loop group associated to an algebraic group $G$ defined in \eqref{eqn: loop group}. This result is the content of Theorem \ref{theorem: density 2}.
    
    \item Using this result, we then show that $\mathcal{M}(\infty,K)\cong L^-\GL_K\times L^-\GL_1$. This in turn would imply that 
    \begin{equation}
        \mathbb C[\mathcal{M}(\infty,K)]\cong \mathbb C[L^-\GL_K]\otimes \mathbb C[L^-\GL_1].
    \end{equation}
\end{enumerate}

\smallskip \S\ref{sec:the hilbert series of C[M(N,K)]} is devoted to study of Hilbert series of modules for $\mbb{C}[\mcal{M}(N,K)]$. The main result of this section is the computation of Hilbert series for $\mbb{C}[\mcal{M}(N,K)]$ in Theorem \ref{thr:the Hilbert series of C[M(N,K)]} and its large-$N$ limit $\mathbb C[\mathcal{M}(\infty,K)]$ in Proposition \ref{prop:the Hilbert series in the large-N limit}.

\smallskip In \S\ref{sec:quanization of phase space}, we study the quantization $\mbb{C}_\hbar[\mcal{M}(N,K)]$ of the ring of functions $\mbb{C}[\mcal{M}(N,K)]$ in the phase space. Quantization amounts to replace the Poisson brackets \eqref{eq:Poisson bracket on the phase space (introduction)} with commutators and studying the resulting algebras. The main results of this section are as follows.
\begin{enumerate}
    \item We first prove the commutator of $T^{(n)}_{ab}$:
    \begin{equation}\label{eq: RTT_intro}
[T^{(p)}_{ab},T^{(q)}_{cd}]=\hbar\sum_{i=-1}^{\min(p,q)-1}\left(T^{(i)}_{cb}T^{(p+q-1-i)}_{ad}-T^{(p+q-1-i)}_{cb}T^{(i)}_{ad}\right).
\end{equation}
This is equivalent to the RTT relation if one defines the generating functions $T_{ab}(z)$ (the RTT generators) of $T^{(n)}_{ab}$ by the following power-series expansion at $z\to \infty$
\begin{align*}
    T_{ab}(z):=\sum_{n\ge -1} T^{(n)}_{ab}z^{-n-1}=\delta_{ab}+I_a\frac{1}{z-B}J_b.
\end{align*}

\item We next construct two different surjective algebra maps from $Y_\hbar(\fgl_K)\otimes \Lambda$ to $\mathbb C_{\hbar}[\cM(N,K)]$, the first denoted $\rho_N$ is constructed in Proposition \ref{prop:the surjective map between Yangian and deformed ring of functions}, and the second denoted $\widetilde\rho_N$ is constructed in \S\ref{sec:another map between Yangian and deformed ring of functions}. Then we prove the key technical result Theorem \ref{theorem: quantum ideal} which characterizes the kernel of $\rho_N$. The proof of Theorem \ref{theorem: quantum ideal} leads to our main result Corollary \ref{cor:M(N,K) as Coulomb branch}.

\item In \S\ref{sec:quantized coproduct}, we construct the coproducts of $Y_\hbar(\fgl_K)\otimes \Lambda$ and $\bC_\hbar[\cM(N,K)]$, and they are compatible in the sense that the following diagram
\begin{equation}
\begin{tikzcd}
Y_\hbar(\fgl_K)\otimes \Lambda \ar[r,"\underset{n_1,n_2}{\mathbf \Delta}"] \ar[d,"\rho_{n_1+n_2}" '] & \left(Y_\hbar(\fgl_K)\otimes \Lambda\right)\otimes \left(Y_\hbar(\fgl_K)\otimes \Lambda\right) \ar[d,"\rho_{n_1}\otimes\rho_{n_2}"] \\
\bC_\hbar[\cM(n_1+n_2,K)] \ar[r,"\underset{n_1,n_2}{\Delta}"] & \bC_\hbar[\cM(n_1,K)]\otimes \bC_\hbar[\cM(n_2,K)]
\end{tikzcd}
\end{equation}
is commutative.

\item Finally, we explain the identification between the quantized ring of functions in phase space $\mbb{C}_\hbar[\mcal{M}(N,K)]$ and the Coulomb branch algebra of certain $3d$ $\mscr{N}=4$ quiver gauge theories. This leads to a conceptual proof of Corollary \ref{cor:M(N,K) as Coulomb branch}, see Remark \ref{rmk:another proof}.
\end{enumerate}

\smallskip Some details are relegated to the appendices. The Hall--Littlewood polynomial has been reviewed in Appendix \ref{sec: Hall-Littlewood Polynomials}. Geometrization of the Jing operators, which are used in giving a vertex-algebra definition of the Hall--Littlewood polynomials, is explained in Appendix \ref{appsec:affine grassamannians and geometrization of Jing operators}.

\section{Geometry of the Phase Space \texorpdfstring{$\mathcal M(N,K)$}{}}
\label{sec:geometry of the phase space}
This section is devoted to the elucidation of the geometry of the phase space $\mcal{M}(N,K)$ and its ring of functions.

\subsection{Generators of \texorpdfstring{$\mathbb C[\mathcal M(N,K)]$}{C[M(N,K)]}}\label{subsec:generators}

By invariant theory, the algebra of functions on $\mathcal M(N,K)$, denoted by $\mathbb C[\mathcal M(N,K)]$, is generated by 
\begin{align}\label{eq: gen}
    \Tr(B^n)=B_{i_1 i_2}B_{i_2 i_3}\cdots B_{i_n i_1},\;\text{ and }\; I_aB^mJ_b=I_{a i_1}B_{i_1 i_2}\cdots B_{i_m i_{m+1}}J_{i_{m+1} b},
\end{align}
where $n\ge 1$, $m\ge 0$, and $1\le a,b\le K$. In fact, the trace relations \footnote{The identity $\det(\Id-t\cdot B)=\exp\left(-\sum_{n=1}^{\infty}\frac{t^n}{n}\Tr(B^n)\right)$ implies that the right-hand-side is a polynomial of degree $N$, thus $\Tr(B^r)$ $(r>N)$ can be expressed as polynomials of $\{\Tr(B^n):1\le n\le N\}$. The Caylay-Hamilton identity $\det(x\cdot\Id-B)|_{x=B}=0$ shows that $B^s$ $(s\ge N)$ can be expressed as polynomials of $\{B^m:0\le m\le N-1\}$ and $\{\Tr(B^n):1\le n\le N\}$.} guarantees that the subset $\{\Tr(B^n),I_aB^mJ_b: 1\le n\le N, 0\le m\le N-1,1\le a,b\le K\}$ of generators \eqref{eq: gen} is sufficient to generate $\mathbb C[\mathcal M(N,K)]$.

When $K=1$, it turns out that $\{\Tr(B^n),IB^mJ: 1\le n\le N, 0\le m\le N-1\}$ freely generates $\mathbb C[\mathcal M(N,1)]$ so that we have $\mathcal M(N,1)\cong \mathbb A^{2N}$. In fact, we will see shortly in Proposition \ref{prop:phase space is a normal affine variety} that $\dim \mathcal M(N,1)=2N$. Since the map $\mathcal M(N,1)\to \mathbb A^{2N}$ induced by the generators $\{\Tr(B^n),IB^mJ: 1\le n\le N, 0\le m\le N-1\}$ is closed embedding, this map must be an isomorphism by dimensional reason.

For general $K$, let us fix a pair of integers $a,b$, then the functions $\Tr(B^n), I_aB^mJ_b$ give rise to a morphism $\eta_{ab}:\mathcal M(N,K)\to \mathcal M(1,K)$ sending a triple $(B,J,I)$ to $(B,J_b,I_a)$. From the above discussions, we have the following result.
\begin{proposition}\label{Prop: closed embedding}
The product of $\eta_{ab}$ is a closed embedding
\begin{align}\label{eq:closed embedding}
    \eta:=\prod_{1\le a,b\le K} \eta_{ab}: \mathcal M(N,K) \longhookrightarrow \mathcal M(N,1)\times _{\mathbb A^{(N)}}\cdots \times _{\mathbb A^{(N)}} \mathcal M(N,1),
\end{align}
where the right hand side has $K^2$ copies of $\mathcal M(N,1)$.
\end{proposition}

\subsection{Factorization} There is an obvious morphism: 
\begin{align}\label{eqn: factorization}
    \mathfrak{f}_{N_1,N_2}:\mathcal M(N_1,K)\times \mathcal M(N_2,K)&\longrightarrow \mathcal M(N_1+N_2,K),\\
    (B^{(1)},J^{(1)},I^{(1)})\times (B^{(2)},J^{(2)},I^{(2)})&\mapsto \left(
    \begin{bmatrix}
    B^{(1)} & 0\\
    0 & B^{(2)}
    \end{bmatrix},
    \begin{bmatrix}
    J^{(1)}\\
    J^{(2)}
    \end{bmatrix},
    \begin{bmatrix}
    I^{(1)} & I^{(2)}
    \end{bmatrix}.
    \right)
\end{align}
Consider the natural projection
\begin{align}
    \Phi_{N}: \mathcal M(N,K)\longrightarrow \mathbb A^{(N)}.
\end{align}
Here $\Phi_{N}$ maps a triple $(B,J,I)$ to the coefficients of the characteristic polynomial of $B$, and $\mathbb A^{(N)}$ is the $N$-th symmetric product of affine line $\mathbb A^1$, which parametrizes coefficients of the characteristic polynomial of $B$. Denote by $\left(\mathbb A^{(N_1)}\times \mathbb A^{(N_2)}\right)_{\mathrm{disj}}$ the open subset of $\mathbb A^{(N_1)}\times \mathbb A^{(N_2)}$ such that eigenvalues of $B^{(1)}$ is disjoint from eigenvalues of $B^{(2)}$. Analogous to the $K=1$ case discussed in \cite{finkelberg2014quantization}, we have the following factorization isomorphism
\begin{proposition}
The restriction of $\mathfrak{f}_{N_1,N_2}$ on $\left(\mathbb A^{(N_1)}\times \mathbb A^{(N_2)}\right)_{\mathrm{disj}}$ is isomorphism:
\begin{align*}
    \mathfrak{f}_{N_1,N_2}:\left(\mathcal M(N_1,K)\times \mathcal M(N_2,K)\right)_{\mathrm{disj}}\cong \mathcal M(N_1+N_2,K) \times _{\mathbb A^{(N_1+N_2)}} \left(\mathbb A^{(N_1)}\times \mathbb A^{(N_2)}\right)_{\mathrm{disj}}.
\end{align*}
Here $\left(\mathcal M(N_1,K)\times \mathcal M(N_2,K)\right)_{\mathrm{disj}}$ is the restriction of $\mathcal M(N_1,K)\times \mathcal M(N_2,K)$ on $\left(\mathbb A^{(N_1)}\times \mathbb A^{(N_2)}\right)_{\mathrm{disj}}$.
\end{proposition}

\begin{remark}
The factorization isomorphism $\mathfrak{f}_{N_1,N_2}$ is compatible with the embedding $\eta$ in \eqref{eq:closed embedding}.
\end{remark}

\begin{proposition}\label{prop:phase space is a normal affine variety}
$\mathcal M(N,K)$ is a normal affine variety of dimension $2NK$.
\end{proposition}

\begin{proof}
$\mathcal M(N,K)$ is normal and affine since is the quotient of an affine space by $\GL_N$, we only need to show that its dimension is $2NK$. By the factorization isomorphism, it suffices to show that $\dim \mathcal M(1,K)=2K$. Note that $\mathcal M(1,K)$ is isomorphic to the $\mathbb A^1$ times the variety of $K\times K$ matrices with rank $\le 1$, and it is known that the latter has dimension $2K-1$ \cite{bruns2006determinantal}.
\end{proof}

\begin{proposition}\label{prop: Phi_N is flat}
The morphism $\Phi_{N}: \mathcal M(N,K)\rightarrow \mathbb A^{(N)}$ is flat.
\end{proposition}

\begin{proof}
Since the map $\fgl_N\to \mathbb A^{(N)}$ sending a matrix $B$ to the coefficients of the characteristic polynomial of $B$ is flat \cite[\S 3.2]{Slodowy_1980}. Therefore the morphism $\mathrm{Rep}(N,K)\to \mathbb A^{(N)}$ sending $(B,J,I)$ to the coefficients of the characteristic polynomial of $B$ is flat. Since $\GL_N$ is reductive, the $\GL_N$-invariant subspace $\mathbb C[\mathcal M(N,K)]$ is a direct summand of $\mathbb C[\mathrm{Rep}(N,K)]$ as $\mathbb C[\mathbb A^{(N)}]$-module. Thus $\Phi_{N}: \mathcal M(N,K)\rightarrow \mathbb A^{(N)}$ is flat.
\end{proof}

\subsection{Singularities and Resolution} 

Let us pick a character $\zeta:\GL_N\to \bC^{\times}$ such that $\zeta(g)=\det(g)$. According to \cite[Proposition 3.1]{king1994moduli}, a point in $\mathrm{Rep}(N,K)$ is $\zeta$-semistable if and only if the following condition is satisfied:
\begin{itemize}
    \item[] if $S\subseteq\mathbb C^N$ is a linear subspace such that $B(S)\subseteq S$, and $\im(J)\subseteq S$, then $S=\mathbb C^N$.
\end{itemize}
It is straightforward to see that a point in $\mathrm{Rep}(N,K)$ is $\zeta$-semistable if and only if it is $\zeta$-stable. The $\zeta$-stable locus in $\mathrm{Rep}(N,K)$ is denoted $\mathrm{Rep}(N,K)^{\zeta-s}$. It is easy to see that $\mathrm{Rep}(N,K)^{\zeta-s}$ is nonempty. Define the GIT quotient
\begin{align}\label{eq: M^s(N,K)_def}
    \mathcal M^\zeta(N,K):=\mathrm{Rep}(N,K)\sslash_\zeta \GL_N=\mathrm{Rep}(N,K)^{\zeta-s}\sslash \GL_N.
\end{align}
Then $\mathcal M^\zeta(N,K)$ is a smooth variety of dimension $2NK$, since $\GL_N$ acts on $\mathrm{Rep}(N,K)^{\zeta-s}$ freely. According to \cite[\S 2]{king1994moduli} there is a natural projective morphism 
\begin{align}\label{eq: res of sing}
    f:\mathcal{M}^\zeta(N,K)\longrightarrow \mathcal{M}(N,K).
\end{align}

\begin{proposition}\label{prop: res of sing}
The projective morphism $f:\mathcal{M}^\zeta(N,K)\to \mathcal{M}(N,K)$ is a resolution of singularities. Moreover, we have $\mathcal O_{\mathcal M(N,K)}\cong \mathbf Rf_*\mathcal O_{\mathcal M^\zeta(N,K)}$, i.e.
\begin{itemize}
    \item[(1)] $\mathbf R^if_*\mathcal O_{\mathcal M^\zeta(N,K)}=0$ for $i>0$,
    \item[(2)] the natural homomorphism $\mathcal O_{\mathcal M(N,K)}\to f_*\mathcal O_{\mathcal M^\zeta(N,K)}$ is an isomorphism.
\end{itemize}
\end{proposition}

\begin{proof}
Let us prove the second statement first. Since $\cM(N,K)$ is affine, it suffices to show that
\begin{align}\label{eq: quatization theorem}
H^i(\mathcal M^\zeta(N,K),\mathcal O_{\mathcal M^\zeta(N,K)})=
\begin{cases}
    \bC[\mathcal M(N,K)] &, \text{ if }i=0,\\
    0 &, \text{ if }i>0.
\end{cases}
\end{align}
Then \eqref{eq: quatization theorem} is a special case of \cite[Theorem 3.29]{Halpern_Leistner_2014} by taking $F^{\bullet}=G^{\bullet}=\mathcal O_{[\mathrm{Rep}(N,K)/\GL_N]}$. 

To show that $f:\mathcal{M}^\zeta(N,K)\to \mathcal{M}(N,K)$ is a resolution of singularities, it is enough to show that $f$ is birational, i.e. there exists an open subscheme $U\subset \cM(N,K)$ such that $f$ induces isomorphism $f^{-1}(U)\cong U$. As we have shown that the natural homomorphism $\mathcal O_{\mathcal M(N,K)}\to f_*\mathcal O_{\mathcal M^\zeta(N,K)}$ is an isomorphism, $f$ must be a dominant morphism, i.e. the image of $f$ is dense in $\cM(N,K)$. Since $f$ is proper, it follows that $f$ is surjective. Then there exists an open subscheme $U\subset \cM(N,K)$ such that $f^{-1}(U)\to U$ is finite, because $\dim \cM^\zeta(N,K)=2NK=\dim\cM(N,K)$. The isomorphism $\mathcal O_{\mathcal M(N,K)}\cong f_*\mathcal O_{\mathcal M^\zeta(N,K)}$ implies that $f^{-1}(U)\to U$ is an isomorphism.
\end{proof}

Recall that an algebraic variety $X$ is said to have \textit{rational singularities} if for every resolution of singularities $\pi:Y\to X$ we have $\cO_{X}\cong \mathbf R\pi_*\cO_{Y}$. An algebraic variety $X$ is said to have \textit{Gorenstein singularities} if the dualizing complex $\omega_X$ is a locally free sheaf.

\begin{corollary}\label{cor: rational Gorenstein}
$\mathcal M(N,K)$ has rational and Gorenstein singularities.
\end{corollary}

\begin{proof}
By \cite[Theorem 5.10]{kollar1998birational}, an algebraic variety $X$ has rational singularities if and only if there exists a resolution of singularities $\pi:Y\to X$ such that $\cO_{X}\cong \mathbf R\pi_*\cO_{Y}$. In particular, $\mathcal M(N,K)$ has rational singularities by Proposition \ref{prop: res of sing}. By Lemma \ref{lem: triv can bun} below, the canonical line bundle on $\cM^\zeta(N,K)$ is trivial, i.e. $\mathcal K_{\mathcal M^\zeta(N,K)}\cong \cO_{\mathcal M^\zeta(N,K)}$. Then the dualizing sheaf $\omega_{\mathcal M(N,K)}$ is computed by $$\omega_{\mathcal M(N,K)}\cong \mathbf Rf_* \mathcal K_{\mathcal M^\zeta(N,K)}\cong \mathbf Rf_*\mathcal O_{\mathcal M^\zeta(N,K)}\cong \mathcal O_{\mathcal M(N,K)},$$which is a line bundle. Thus $\mathcal M(N,K)$ has Gorenstein singularities.
\end{proof}

\begin{lemma}\label{lem: triv can bun}
The canonical line bundle on $\mathcal{M}^\zeta(N,K)$ is trivial.
\end{lemma}
\begin{proof}
Denote by $\mathcal V$ the tautological sheaf on $\mathcal{M}^\zeta(N,K)$, which is the descent of $\mathbb C^N$ along the quotient $\mathrm{Rep}(N,K)^{\zeta-s}\to \mathcal{M}^\zeta(N,K)$, and denote by $W$ the framing vector space, then there is a short exact sequence
\begin{align}
    0\longrightarrow \End(\mathcal V)\longrightarrow \End(\mathcal V)\oplus W\otimes \mathcal V^*  \oplus W^*\otimes \mathcal V \longrightarrow T_{\mathcal{M}^\zeta(N,K)} \longrightarrow 0.
\end{align}
Here $T_{\mathcal{M}^\zeta(N,K)}$ is the tangent sheaf of $\mathcal{M}^\zeta(N,K)$. From this short exact sequence we get
\begin{align*}
    \mathcal K_{\mathcal{M}^\zeta(N,K)}=\det T_{\mathcal{M}^\zeta(N,K)}^*\cong \det (W\otimes \mathcal V^*)\otimes \det(W^*\otimes \mathcal V)\cong \mathcal O_{\mathcal{M}^\zeta(N,K)}.
\end{align*}
\end{proof}

\subsection{A Poisson Structure on \texorpdfstring{$\mcal{M}(N,K)$}{}}\label{sec:a natural Poisson structure}

Let us introduce a Poisson structure on the space of $(B,J,I)$ as following
\begin{eqgathered}
    \label{Poisson before quotient}
\{J_{ia},I_{bj}\}=\delta_{ab}\delta_{ij},\qquad \{B_{mn},I_{bj}\}=\{B_{mn},J_{ia}\}=0,
\\
\{B_{mn},B_{pq}\}=\delta_{pn}B_{mq}-\delta_{mq}B_{pn}.
\end{eqgathered}
Here we treat $J,I$ as usual bosonic variables, i.e. commute instead of anti-commute with each other. This Poisson structure comes from the classical limit of $U_{\hbar}(\mathfrak{gl}_N)\otimes \mathrm{Weyl}^{\otimes NK}_{\hbar}$, where $\mathrm{Weyl}^{\otimes NK}_{\hbar}$ is the Weyl algebra generated by $J,I$. It is easy to see that the Poisson structure is equivariant under the $\GL_N$ action, so it descends to $\mathcal M(N,K)$.

\begin{remark}
This is \textit{not} the Poisson structure for the Zastava space. In fact, when $K=1$, this Poisson structure on $\mathbb C[\mathcal M(N,1)]$ is trivial, see the Theorem \ref{prop: Poisson structure} below.
\end{remark}

Define $T^{(n)}_{ab}=I_aB^nJ_b$, and we use the convention $T^{(-1)}_{ab}=\delta_{ab}$, then denote by $T_{ab}(z)$ the power series expanded at $z\to \infty$:
\begin{align*}
    T_{ab}(z)=\sum_{n\ge -1} T^{(n)}_{ab}z^{-n-1}=\delta_{ab}+I_a\frac{1}{z-B}J_b.
\end{align*}

\begin{proposition}\label{prop: Poisson structure}
The Poisson brackets between $T_{ab}^{(k)}$ are:
\begin{align}\label{eqn: Poisson structure of RTT}
    \{T^{(p)}_{ab},T^{(q)}_{cd}\}=\sum_{i=-1}^{\min(p,q)-1}\left(T^{(p+q-1-i)}_{cb}T^{(i)}_{ad}-T^{(i)}_{cb}T^{(p+q-1-i)}_{ad}\right).
\end{align}
And for all $n\ge 1$, $\mathrm{Tr}(B^n)$ is Poisson central.
\end{proposition}

\begin{proof}
This is the classical limit of \eqref{eqn: commutators}, which will be proven independently.
\end{proof}

\begin{remark}[The Poisson Structure from Twisted Holography]
    The algebra of gauge-invariant operators of the systems consisting of 2d BF theory coupled to 1d quantum mechanics can be constructed using Feynman diagram computations \cite{IshtiaqueMoosavianZhou201809}. From this, one can construct a Poisson structure on $\C[\mcal{M}(N,K)]$. Hence, a natural question is the connection between this Poisson structure and the one in \eqref{Poisson before quotient}. It can be shown that (a bosonic version of) the former is the same as the latter. The details of this has been expounded in Appendix \ref{sec:Relationship between two Poisson structures}.\qed 
\end{remark}

\subsection{Multiplication Morphism} 
Apart from the obvious factorization map \eqref{eqn: factorization}, there is another map
\begin{align}\label{eqn: multiplication}
    \mathfrak{m}_{N_1,N_2}:\mathcal M(N_1,K)\times \mathcal M(N_2,K)&\longrightarrow \mathcal M(N_1+N_2,K),\\
    (B^{(1)},J^{(1)},I^{(1)})\times (B^{(2)},J^{(2)},I^{(2)})&\mapsto \left(
    \begin{bmatrix}
    B^{(1)} & J^{(1)} I^{(2)}\\
    0 & B^{(2)}
    \end{bmatrix},
    \begin{bmatrix}
    J^{(1)}\\
    J^{(2)}
    \end{bmatrix},
    \begin{bmatrix}
    I^{(1)} & I^{(2)}
    \end{bmatrix}
    \right).
\end{align}
We have the following elementary property of the multiplication morphism.
\begin{proposition}\label{prop: multiplication is dominant}
The multiplication morphism $\mathfrak{m}_{N_1,N_2}$ is dominant.
\end{proposition}

\begin{proof}

It suffices to prove that the composition $\mathfrak{f}_{N_1,N_2}^{-1}\circ \mathfrak{m}_{N_1,N_2}$ is dominant when restricted on $\left(\mathbb A^{(N_1)}\times \mathbb A^{(N_2)}\right)_{\mathrm{disj}}$. First of all, we construct a $\GL_{N_1}\times \GL_{N_2}$ equivariant map $$\widetilde {\mathfrak{m}}_{N_1,N_2}: \left(\mathrm{Rep}(N_1,K)\times \mathrm{Rep}(N_2,K)\right)_{\mathrm{disj}}\longrightarrow \left(\mathrm{Rep}(N_1,K)\times \mathrm{Rep}(N_2,K)\right)_{\mathrm{disj}},$$such that $\widetilde {\mathfrak{m}}_{N_1,N_2}$ descends to $\mathfrak{f}_{N_1,N_2}^{-1}\circ \mathfrak{m}_{N_1,N_2}$ after taking the quotient by $\GL_{N_1}\times \GL_{N_2}$. The construction is as follows. If the spectra of $B_1$ and $B_2$ are disjoint from each other, then linear map $\mathrm{Mat}(N_1,N_2)\to\mathrm{Mat}(N_1,N_2), X\mapsto B_1X-XB_2$ is an isomorphism. Let $A$ be the unique $N_1\times N_2$ matrix such that
$$B^{(1)}A-AB^{(2)}=J^{(1)} I^{(2)}\;\text{holds.}$$
Then we can use the matrix
\begin{align*}
    \begin{bmatrix}
    \Id & A\\
    0 & \Id
    \end{bmatrix}\; \text{to diagonalize}\; 
    \begin{bmatrix}
    B^{(1)} & J^{(1)} I^{(2)}\\
    0 & B^{(2)}
    \end{bmatrix}
\end{align*}
and it accordingly maps $\left[ I^{(1)} ,I^{(2)} \right]$ to $\left[ I^{(1)},I^{(2)} -I^{(1)}A  \right]$ and $\left[J^{(1)},J^{(2)} \right]^{\mathrm{t}}$ to $\left[ J^{(1)}+AJ^{(2)} ,J^{(2)}\right]^{\mathrm{t}}$. Hence we define $\widetilde {\mathfrak{m}}_{N_1,N_2}$ as
\begin{align}
    (B^{(1)},J^{(1)},I^{(1)})\times (B^{(2)},J^{(2)},I^{(2)})\mapsto (B^{(1)},J^{(1)}+AJ^{(2)},I^{(1)})\times (B^{(2)},J^{(2)},I^{(2)}-I^{(1)}A).
\end{align}
Notice that the tangent map $d\widetilde {\mathfrak{m}}_{N_1,N_2}$ is an isomorphism at any point $(B^{(1)},0,I^{(1)})\times (B^{(2)},J^{(2)},0)$, so $\widetilde {\mathfrak{m}}_{N_1,N_2}$ is generically \' etale thus it is dominant. Then it follows that $\mathfrak{m}_{N_1,N_2}$ is dominant.
\end{proof}

\begin{proposition}\label{prop: multiplication is Poisson}
The multiplication morphism $\mathfrak{m}_{N_1,N_2}$ has following properties
\begin{itemize}
    \item[(1)] $\mathfrak{m}_{N_1,N_2}$ is Poisson,
    \item[(2)] $\mathfrak{m}_{N_1+N_2,N_3}\circ(\mathfrak{m}_{N_1,N_2}\times \mathrm{Id})=\mathfrak{m}_{N_1,N_2+N_3}\circ(\mathrm{Id}\times \mathfrak{m}_{N_2,N_3})$, i.e. multiplication is associative.
\end{itemize}
\end{proposition}
The proposition will be evident once we make connection to the multiplication map on the loop group in the next section. Note that the factorization map $\mathfrak{f}_{N_1,N_2}$ is not Poisson in general.

\subsection{Embedding \texorpdfstring{$\mathcal M(N,K)\hookrightarrow \mathcal M(N',K)$}{M(N,K) in M(N',K)}} Suppose that $N<N'$, then we have a morphism
\begin{align}\label{eqn: embedding}
    \iota_{N,N'}:\mathcal M(N,K)&\longrightarrow \mathcal M(N',K),\\
    (B,J,I)&\mapsto \left(
    \begin{bmatrix}
    B & 0\\
    0 & 0
    \end{bmatrix},
    \begin{bmatrix}
    J\\
    0
    \end{bmatrix},
    \begin{bmatrix}
    I & 0
    \end{bmatrix}
    \right).
\end{align}
Note that $\iota_{N,N'}^*(\mathrm{Tr}(B^n))=\mathrm{Tr}( B^n),\iota_{N,N'}^*(T_{ab}^{(m)})=T_{ab}^{(m)}$, so $\iota_{N,N'}^*$ is surjective, thus $\iota_{N,N'}$ is a closed embedding.

\begin{proposition}\label{prop: properties of embedding}
The embedding $\iota_{N,N'}$ has following properties
\begin{itemize}
    \item[(1)] $\iota_{N',N''}\circ\iota_{N,N'}=\iota_{N,N''}$,
    \item[(2)] $\iota_{N,N'}$ is Poisson,
    \item[(3)] $\mathfrak{m}_{N_1',N_2'}\circ\left(\iota_{N_1,N'_1}\times \iota_{N_2,N'_2}\right)=\iota_{N_1+N_2,N_1'+N_2'}\circ \mathfrak{m}_{N_1,N_2}$.
\end{itemize}
\end{proposition}

\begin{proof}
Property (1) is obvious from definition of $\iota_{N,N'}$, (2) is a corollary of Proposition \ref{prop: Poisson structure}, only (3) needs explanation. Using property (1), the proof of (3) reduces to the cases of either $N_1'=N_1, N_2'=N_2+1$ or $N_1'=N_1+1, N_2'=N_2$. The first case is obvious from the definition of embedding and multiplication morphism, so we only need to consider the case when $N_1'=N_1+1, N_2'=N_2$. It amounts to showing that
\begin{align*}
    \left(
    \begin{bmatrix}
    B^{(1)} & 0 & J^{(1)} I^{(2)} \\
    0 & 0 & 0 \\
   0 & 0 & B^{(2)} 
    \end{bmatrix},
    \begin{bmatrix}
    J^{(1)}\\
    0\\
    J^{(2)}
    \end{bmatrix},
    \begin{bmatrix}
    I^{(1)} & 0 & I^{(2)}
    \end{bmatrix}\right)
\end{align*}
is equivalent to 
\begin{align*}
    \left(
    \begin{bmatrix}
    B^{(1)} & J^{(1)} I^{(2)} & 0 \\
    0 & B^{(2)}  & 0 \\
    0 & 0 & 0 
    \end{bmatrix},
    \begin{bmatrix}
    J^{(1)}\\
    J^{(2)}\\
    0
    \end{bmatrix},
    \begin{bmatrix}
    I^{(1)} & I^{(2)} & 0
    \end{bmatrix}
    \right)
\end{align*}
under the action of some matrix $W\in\GL_{N_1+N_2+1}$. It is elementary to check that
\begin{align*}
    W=\begin{bmatrix}
    \mathrm{Id}_{N_1} & 0  \\
    0 & w_{N_2}w_{N_2-1}\cdots w_1  
    \end{bmatrix}
\end{align*}
does the job, where $w_i\in \GL_{N_2+1}$ switches row $i$ and row $i+1$.
\end{proof}

\subsection{\texorpdfstring{$\cM(N,K)$}{M(N,K)} as a Hamiltonian Reduction}\label{subsec:Hamiltonian reduction}

The phase space $\cM(N,K)$ can be equivalently described by a Hamiltonian reduction. Namely, we consider an enlarged space of representations 
\begin{align}
    \widetilde{\mathrm{Rep}}(N,K):=\mathrm{Rep}(N,K)\oplus \End(\bC^N),
\end{align}
where the $N\times N$ matrices in the newly-added component $\End(\bC^N)$ will be denoted $\widetilde{B}$. $\End(\bC^N)$ is endowed with the Poisson structure $$\{\widetilde B_{mn},\widetilde B_{pq}\}=\delta_{pn}\widetilde B_{mq}-\delta_{mq}\widetilde B_{pn}.$$
The action of gauge group $\GL_N$ on $\widetilde{\mathrm{Rep}}(N,K)$ is Hamiltonian with the moment map
\begin{align}
    \mu:\widetilde{\mathrm{Rep}}(N,K)\to \fgl_N^*,\quad \mu(B,J,I,\widetilde B)=B+\widetilde{B}-JI.
\end{align}
We note that the equation $\mu=0$ solves $\widetilde{B}$ in terms of $(B,J,I)$, so $\mathrm{Rep}(N,K)\cong \mu^{-1}(0)$. Therefore, $\cM(N,K)$ is isomorphic to the Hamiltonian reduction $\mu^{-1}(0)\sslash \GL_N$.

\begin{remark}
Recall that $B$ is $B_{w>0}$ in the solution to EOM of BF theory, see \eqref{moment map}. Then $\widetilde{B}$ has a natural identification with $-B_{w<0}$ in \eqref{moment map}.
\end{remark}

\section{Large-\texorpdfstring{$N$}{} Limit}
\label{sec:large-N limit}

In this section we use the embeddings $\iota_{N,N'}: \cM(N,K)\hookrightarrow\cM(N',K)$ constructed in the previous section to define the large-$N$ limit of the family $\cM(N,K)$ as the spectrum of $\mathbb C^{\times}$-finite elements in the inverse limit of algebras $\mathbb C[\cM(N,K)]$, and show that the large-$N$ limit is isomorphic to the Poisson group $L^-(\GL_K\times \GL_1)$, defined below. It is known that $L^-(\GL_K\times \GL_1)$ quantizes to the Yangian $Y_\hbar(\fgl_K)\otimes \Lambda$, and we will explore the quantized version of the large-$N$ limit in the next section.

\begin{definition}\label{defn: large N limit of M(N,K)}
Define $\mathbb C[\mathcal{M}(\infty,K)]$ to be the subalgebra of $\underset{\substack{\longleftarrow\\N}}{\lim}\: \mathbb C[\mathcal M(N,K)]$ generated by $\{T^{(n)}_{ab},\Tr(B^m):n\in \bZ_{\ge 0}, m\in \bZ_{\ge 1},1\le a,b\le K\}$. Then define $\mathcal{M}(\infty,K)=\Spec \mathbb C[\mathcal{M}(\infty,K)]$.
\end{definition}

Define $L^{-}\GL_K$ to be the group of power series
\begin{align}\label{eqn: loop group}
    \Id+\sum_{i=1}^{\infty}g_i z^{-i},\; g_i\in \mathfrak{gl}_{K}.
\end{align}
Here the group structure on $L^{-}\GL_K$ is the multiplication of power series in matrices. $L^{-}\GL_K$ is endowed with scheme structure of an infinite dimensional affine space. Consider the morphism
\begin{align}
    i_N=(\pi_N,\varphi_N): &\mathcal M(N,K)\rightarrow L^{-}\GL_K\times L^{-}\GL_1,\\
    (B,J,I)&\mapsto \left(\Id+I\frac{1}{z-{B}}J,\:\frac{1}{z^N}\det(z-B)\right),
\end{align}
which is a closed embedding because $T^{(n)}_{ab}$ and $\Tr(B^m)$ generate $\mathbb C[\mathcal{M}(N,K)]$. Here $(z-{B})^{-1}$ is expanded as a power series of matrices in $z^{-1}$. It is known that $L^{-}\GL_K$ is a Poisson-Lie group scheme whose Poisson structure comes from the Manin triple \cite[\S 2.3]{kamnitzer2014yangians} $$(\mathfrak{gl}_K(\!(z^{-1})\!),z^{-1}\mathfrak{gl}_K[\![z^{-1}]\!],\mathfrak{gl}_K[z]).$$ Explicitly, let $T_{ab}^{(n)}, n\ge -1$ be the function on $L^{-}\GL_K$ that takes the value of $ab$ component of $g_{n+1}$ and we use the convention that $T_{ab}^{(-1)}=\delta_{ab}$, then the Poisson structure on $L^{-}\GL_K$ is determined by the equation
\begin{align}\label{eqn: Poisson structure of loop group}
    (u-v)\{T_{ab}(u),T_{cd}(v)\}=T_{ad}(v)T_{cb}(u)-T_{ad}(u)T_{cb}(v),\; \text{where}\; T_{ab}(u)=\sum_{i=-1}^{\infty}T^{(i)}_{ab}u^{-i-1}.
\end{align}
Compare equation \eqref{eqn: Poisson structure of loop group} with equation \eqref{eqn: Poisson structure of RTT}, and we have
\begin{proposition}\label{prop: projection is Poisson}
The morphism $i_N: \mathcal M(N,K)\rightarrow L^{-}\GL_K\times L^{-}\GL_1$ is Poisson.
\end{proposition}

\begin{proposition}\label{prop: projection compatible with multiplication}
$i_N$ is compatible with embedding $\iota_{N,N'}$ and multiplication $\mathfrak{m}_{N_1,N_2}$, i.e.
\begin{itemize}
    \item[(1)] $i_{N'}\circ \iota_{N,N'}=i_N$,
    \item[(2)] $i_{N_1+N_2}\circ\mathfrak{m}_{N_1,N_2}= \mathfrak{m}\circ(i_{N_1}\times i_{N_2})$.
\end{itemize}
Here $\mathfrak{m}:L^{-}(\GL_K\times\GL_1)\times L^{-}(\GL_K\times\GL_1)\to L^{-}(\GL_K\times\GL_1)$ is the multiplication map of the group $L^{-}(\GL_K\times\GL_1)$.
\end{proposition}

\begin{proof}
(1) is obvious from definition. (2) can be shown by direct computation. If $(B^{(1)},J^{(1)},I^{(1)})$ is a point in $ \mathcal M(N_1,K)$ and $ (B^{(2)},J^{(2)},I^{(2)})$ is a point in $ \mathcal M(N_2,K)$, then $\pi_{N_1+N_2}\circ\mathfrak{m}_{N_1,N_2}$ maps this pair of representations to
\begin{align*}
    &\Id+\begin{bmatrix}
    I^{(1)} & I^{(2)}
    \end{bmatrix}
    \left(z-
    \begin{bmatrix}
    B^{(1)} & J^{(1)} I^{(2)}\\
    0 & B^{(2)}
    \end{bmatrix}
    \right)^{-1}
    \begin{bmatrix}
    J^{(1)}\\
    J^{(2)}
    \end{bmatrix}\\
    &=\Id+\begin{bmatrix}
    I^{(1)} & I^{(2)}
    \end{bmatrix}
    \left(z-
    \begin{bmatrix}
    {B}^{(1)} & J^{(1)} I^{(2)}\\
    0 & {B}^{(2)}
    \end{bmatrix}
    \right)^{-1}
    \begin{bmatrix}
    J^{(1)}\\
    J^{(2)}
    \end{bmatrix}\\
    &=\Id+I^{(1)}\frac{1}{z-{B}^{(1)}}J^{(1)}+I^{(2)}\frac{1}{z-{B}^{(2)}}J^{(2)}\\
    &~+\sum_{i,j=0}^{\infty}I^{(1)}\left({B}^{(1)}\right)^iJ^{(1)}I^{(2)}\left({B}^{(2)}\right)^jJ^{(2)}z^{-i-j-2}\\
    &=\left(\Id+I^{(1)}\frac{1}{z-{B}^{(1)}}J^{(1)}\right)\left(\Id+I^{(2)}\frac{1}{z-{B}^{(2)}}J^{(2)}\right).
\end{align*}
And we also have $\varphi_{N_1+N_2}\circ\mathfrak{m}_{N_1,N_2}= \mathfrak{m}\circ(\varphi_{N_1}\times \varphi_{N_2})$ since determinant of a block diagonal matrix is the product of determinants of each block.
\end{proof}

\begin{proof}[Proof of Proposition \ref{prop: multiplication is Poisson}]
(1) follows from Proposition \ref{prop: projection is Poisson} and the fact that the Poisson structure on $L^{-}(\GL_K\times \GL_1)$ makes it a Poisson-Lie group, i.e. $\mathfrak m$ is Poisson. (2) is a direct consequence of Proposition \ref{prop: projection compatible with multiplication}.
\end{proof}

Since $i_N$ is compatible with $\iota_{N,N'}$ and the generators of $\bC[L^{-}\GL_K\times L^{-}\GL_1]$ are mapped to polynomials in $\{T^{(n)}_{ab},\Tr(B^m):n,m\in \bZ_{\ge 0},1\le a,b\le K\}$, we obtain a morphism
\begin{align}\label{eq: i_infty}
    i_\infty: \cM(\infty,K)\longrightarrow L^{-}\GL_K\times L^{-}\GL_1.
\end{align}

\begin{theorem}\label{theorem: density 2}
$\bigcup_N \mathcal{M}(N,K)$ is Zariski-dense in $L^-\GL_K\times L^-\GL_1$.
\end{theorem}

\begin{proof}
Define $L^-_N\GL_K$ to be the closed subscheme of $L^-\GL_K$ consisting of Laurent polynomials of the form
\begin{align}
    \Id+\sum_{i=1}^N g_iz^{-i},\; g_i\in \fgl_K.
\end{align}
It is easy to see that $\bigcup_N L^-_N\GL_K\times L^-_N\GL_1$ is Zariski-dense in $L^-\GL_K\times L^-\GL_1$.
To prove the theorem, it suffices to show that for every $N$, there exists $N'$ such that $L^-_N\GL_K\times L^-_N\GL_1$ is a closed subscheme of $\mathcal M(N',K)$. Denote by ${\mathfrak{m}}=\mathfrak{m}_{L^-\GL_K}\times \mathfrak{m}_{L^-_N\GL_1}$ the multiplication map on $L^-\GL_K\times L^-\GL_1$. We make two observations 
\begin{itemize}
    \item[(1)] $L^-_1\GL_K\times \{1\}\subseteq \mathcal M(K,K)$. This is because $S\times \{1\}\subseteq \mathcal M(K,K)$, where $S$ is the subvariety of $L^-_1\GL_K$:
    \begin{align*}
    \Id+\frac{g}{z},\; g\in \mathfrak{gl}_K\;\text{such that}\; \mathrm{rank}(g)\le 1.
    \end{align*}
    then we can apply the multiplication ${\mathfrak{m}}$ $K$ times to obtain $L^-_1\GL_K\times \{1\}$, more precisely, we have following linear algebra fact:
\begin{itemize}
    \item Every matrix $M\in \mathfrak{gl}_K$ can be written as a linear combination $M=X_1+\cdots+X_K$ such that $\mathrm{rank}(X_i)\le 1$ and $X_iX_j=0$ if $i<j$.
\end{itemize}
This can be interpreted as 
\begin{align*}
    \Id+\frac{M}{z}=\left(\Id+\frac{X_1}{z}\right)\cdots\left(\Id+\frac{X_K}{z}\right),
\end{align*}
which is exactly what we want to show. To show this fact, we notice that the statement is true for $M$ if and only if it is true for $AMA^{-1}$ for some $A\in \GL_K$, so without loss of generality, we assume that $M$ is a Jordan block $J_{\lambda}$, and then take $X_i=a_i^{\mathrm{t}}b_i$, where 
\begin{align*}
    a_i&=(0,\cdots,0,\lambda,1,0,\cdots,0),\; i<K\;\text{and}\;i\text{-th component is }\lambda,\\
    a_K&=(0,\cdots,0,\lambda),\\
    b_i&=(0,\cdots,0,1,0,\cdots,0),\; i\text{-th component is }1.
\end{align*}
If $M$ is a direct sum of Jordan blocks, then we take $X_i$ associated to each individual block.
    \item[(2)] $ \{1\}\times L^-_1\GL_1\subseteq \mathcal M(1,K)$, because $\{1\}\times L^-_1\GL_1$ is exactly the image of the subset $\{(B,J,I)=(b,0,0):b\in \bC\}\subset \mathcal M(1,K)$ under the map $i_1$.
    \item[(3)] The multiplication map $\mathfrak{m}_{L^-\GL_K}:L^-_1\GL_K\times L^-_N\GL_K\to L^-_{N+1}\GL_K$ is dominant. In fact, the tangent map $d\mathfrak{m}_{L^-\GL_K}$ at the point $\left(1, 1+1/z+\cdots+1/z^N\right)$ is 
\begin{align}\label{eqn: tangent map}
    \left(\frac{X}{z},\frac{Y_1}{z},\cdots, \frac{Y_i}{z^i},\cdots,\frac{Y_N}{z^N}\right)\mapsto \frac{X+Y_1}{z},\cdots,\frac{X+Y_N}{z^N},\frac{X}{z^{N+1}}
\end{align}
where left hand side is a tangent vector at $\left(1, 1+1/z+\cdots+1/z^N\right)$, and right hand side is a tangent vector at $1+1/z+\cdots+1/z^N\in L^-_{N+1}\GL_K$. Since $X,Y_1,\cdots,Y_N$ take value in all matrices in $\mathfrak{gl}_K$, the linear map \eqref{eqn: tangent map} is surjective and thus is an isomorphism by dimension counting. It follows that $\mathfrak{m}_{L^-\GL_K}:L^-_1\GL_K\times L^-_N\GL_K\to L^-_{N+1}\GL_K$ is \' etale at the point $\left(1, 1+1/z+\cdots+1/z^N\right)$, thus it is generically \' etale and dominant.
\end{itemize}
Combine (1) and (2) and use the multiplication ${\mathfrak{m}}$ (which is compatible with the multiplications of $\mathcal M(N,K)$), then we have an inclusion $L^-_1\GL_K\times L^-_1\GL_1\subset \mathcal M(K+1,K)$. (3) implies that the morphism $${\mathfrak{m}}: \left(L^-_1\GL_K\times L^-_1\GL_1\right)\times \left(L^-_N\GL_K\times L^-_N\GL_1\right) \rightarrow L^-_{N+1}\GL_K\times L^-_{N+1}\GL_1$$ is dominant. By induction on $N$, we have inclusions $L^-_N\GL_K\times L^-_N\GL_1\subseteq \mathcal M((K+1)N,K)$. This concludes the proof.
\end{proof}

\begin{corollary}
$\mathcal{M}(\infty,K)\cong L^-\GL_K\times L^-\GL_1$, i.e. 
\begin{align}
    \mathbb C[\mathcal{M}(\infty,K)]\cong \mathbb C[L^-\GL_K]\otimes \mathbb C[L^-\GL_1].
\end{align}
\end{corollary}

\begin{proof}
The algebra map $\mathbb C[L^-\GL_K]\otimes \mathbb C[L^-\GL_1]\to \mathbb C[\mathcal{M}(\infty,K)]$ is surjective because all the generators $\{T^{(n)}_{ab},\Tr(B^m):n\in \bZ_{\ge 0}, m\in \bZ_{\ge 1},1\le a,b\le K\}$ are in the image. This map is also injective because the algebra map $\mathbb C[L^-\GL_K]\otimes \mathbb C[L^-\GL_1]\to \underset{\substack{\longleftarrow\\N}}{\lim}\: \mathbb C[\mathcal M(N,K)]$ is injective as $\bigcup_N \mathcal{M}(N,K)$ is Zariski-dense in $L^-\GL_K\times L^-\GL_1$.
\end{proof}

\section{Hilbert Series of \texorpdfstring{$\mathbb C[\mathcal M(N,K)]$}{}}
\label{sec:the hilbert series of C[M(N,K)]}

Recall that we have a resolution of singularities $f:\mathcal{M}^\zeta(N,K)\to\mathcal{M}(N,K)$ by Proposition \ref{prop: res of sing}, where $\mathcal{M}^\zeta(N,K)$ is the moduli space of $\zeta$-stable representations of the quiver in the Figure \ref{Fig: quiver}. The action of gauge group $\GL_N$ on the space of $\zeta$-stable representations $\mathrm{Rep}(N,K)^{\zeta-s}$ is free, so the quotient map $\mathrm{Rep}(N,K)^{\zeta-s}\to \mathcal{M}^\zeta(N,K)$ is a principal $\GL_N$-bundle. The gauge node vector space $\mathbb C^N$ is a trivial bundle on $\mathrm{Rep}(N,K)$ which is endowed with a non-trivial equivariant structure under the action of $\GL_N$, then it descend to a locally free sheaf $\mathcal V$ on $\mathcal{M}^\zeta(N,K)$ since the $\GL_N$ action on the stable locus is free. $\mathcal V$ is called the \textit{tautological sheaf}, and its determinant line bundle is called the \textit{tautological line bundle} which will be denoted $\mathsf{Det}$. 

By the geometric invariant theory, $\mathcal{M}^\zeta(N,K)$ can be presented as a projective spectrum \cite[\S 2]{king1994moduli}
\begin{align*}
    \mathcal{M}^\zeta(N,K)=\Proj \left(\bigoplus_{m=0}^{\infty}\bC[\mathrm{Rep}(N,K)]^{\GL_N, \chi^m}\right),
\end{align*}
where $\bC[\mathrm{Rep}(N,K)]^{\GL_N, \chi^m}$ is the subspace of $\bC[\mathrm{Rep}(N,K)]$ on which $\GL_N$ acts via the character $\chi^m$. Then $\mathsf{Det}$ is identified with $\mathcal O(1)$ of the above projective spectrum. In particular $\mathsf{Det}$ is an ample line bundle on $\mathcal{M}^\zeta(N,K)$. 

Recall the following Grauert-Riemenschneider vanishing theorem.
\begin{theorem}[{\cite[Corollary 2.68]{kollar1998birational}}]\label{theorem: GR vanishing}
Let $h: X\longrightarrow Y$ be a resolution of singularities in characteristic zero, and let $\mathcal L$ be an ample line bundle on $X$, then $\mathbf R^ih_*(\mathcal K_{X}\otimes \mathcal L)=0$ for $i>0$. Here $\mathcal K_X$ is the canonical line bundle of $X$.
\end{theorem}

Applying the above theorem to the resolution of singularities $f:\mathcal{M}^\zeta(N,K)\to\mathcal{M}(N,K)$ and the ample line bundle $\mathsf{Det}$, we get the following cohomology vanishing result.

\begin{lemma}\label{lem: cohomology vanishing 1}
\begin{align}
    H^i(\mathcal{M}^\zeta(N,K),\mathsf{Det}^{\otimes n})=0,\; \text{for all }\; i>0  \;\text{and }\; n\ge 0.
\end{align}
\end{lemma}

\begin{proof}
By Lemma \ref{lem: triv can bun}, $\mathcal K_{\mathcal{M}^\zeta(N,K)}\cong \mathcal O_{\mathcal{M}^\zeta(N,K)}$. Then the lemma follows from Theorem \ref{theorem: GR vanishing}.
\end{proof}

\begin{definition}
For a nonnegative integer $n$, define the $\mathbb C[\mathcal{M}(N,K)]$ module of level $n$, denoted $\Gamma(N,K,n)$, to be the global section of $n$-th power of tautological line bundle, i.e.
\begin{align}
    \Gamma(N,K,n):=\Gamma(\mathcal{M}^\zeta(N,K),\mathsf{Det}^{\otimes n}).
\end{align}
\end{definition}

In this section, we compute the Hilbert series of $\mathbb C[\mathcal{M}(N,K)]$ and $\Gamma(N,K,n)$. Before starting, let us introduce some notations and explain what we are going to compute.

The quiver in Figure \ref{Fig: quiver} admits an action of $\GL_K\times \mathbb C^{\times}_q\times \mathbb C^{\times}_t$, where $\GL_K$ is the flavour symmetry which acts on the framing vector space, $\mathbb C^{\times}_q$ scales $B$ by $B\mapsto q^{-1}B$, and $\mathbb C^{\times}_t$ scales $I$ by $I\mapsto t^{-1}I$. The convention of the inverse $q^{-1}$ and $t^{-1}$ is such that the functions $\mathrm{Tr}(B^n)$ and $I_aB^mJ_b$ scales by $q^n$ and $q^mt$ respectively (since functions are dual to the space). Although $\mathbb C[\mathcal{M}(N,K)]$ is infinite dimensional, every $\mathbb C^{\times}_q\times \mathbb C^{\times}_t$-weight space of $\mathbb C[\mathcal{M}(N,K)]$ is finite dimensional, thus it makes sense to regard $\mathbb C[\mathcal{M}(N,K)]$ as an element in $K_{\GL_K}(\mathrm{pt})[\![q,t]\!]$. Similarly, the same properties hold for $\Gamma(N,K,n)$. The goal of this section is to compute the these elements.

\begin{definition}
Denote the complexified $\GL_K$-equivariant K-theory of a point by $$K_{\GL_K}(\mathrm{pt})=\bC[x_1^{\pm },\cdots,x_K^{\pm }]^{\mathfrak S_K},$$ where $\mathfrak S_K$ is the permutation group acting on $x_1,\cdots x_K$. We use shorthand notation $f(x)$ for a function of $x_1,\cdots, x_K$, and $f(x^{-1})=f(x_1^{-1},\cdots,x_K^{-1})$. Denote by $Z_{N,K}(x;q,t)$ the element of $\mathbb C[\mathcal{M}(N,K)]$ in $K_{\GL_K}(\mathrm{pt})[\![q,t]\!]$, and denote by $Z^{(n)}_{N,K}(x;q,t)$ the element of $\Gamma(N,K,n)$ in $K_{\GL_K}(\mathrm{pt})[\![q,t]\!]$.
\end{definition}

By Lemma \ref{lem: cohomology vanishing 1}, we have $Z^{(n)}_{N,K}(x;q,t)=\chi(\cM^\zeta(N,K),\mathsf{Det}^{\otimes n})$. The case $K=1$ is trivial: The functions $\Tr(B),\cdots ,\Tr(B^N), IJ, IBJ,\cdots,IB^{N-1}J$ give rise to an isomorphism $\mathcal M(N,1)\cong \mathbb A^{2N}$. The Lemma \ref{lemma: bundle over quot scheme} below, together with the fact that the Hilbert-Chow map for Hilbert scheme of points on smooth curve is isomorphism, implies that $\mathcal{M}^\zeta(N,K)\cong \mathcal M(N,K)$. In fact, $\mathsf{Det}$ in this case is a trivial bundle, with $\mathbb C^{\times }_q\times \mathbb C^{\times }_t$-weight $(1,0)$, thus
\begin{align}
    Z^{(n)}_{N,1}(x;q,t)=q^nZ_{N,1}(x;q,t)=\frac{q^n}{(q;q)_N(t;q)_N}.
\end{align}
Here we use the $q$-Pochhammer symbol
$$(a;q)_n=(1-a)(1-aq)\cdots(1-aq^{n-1}).$$
The case when $K>1$ is harder. In principal, one can use the localization technique to get a formula for  $\chi(\mathcal{M}^\zeta(N,K),\mathsf{Det}^{\otimes n})$ in terms of summation over fixed points, but it involves complicated denominators that are hard to extract the power series in $q$ and $t$ explicitly. Our strategy is to reduce the computation to Euler character of vector bundles on Quot scheme, which is related to the affine Grassmannian of $\GL_K$, and apply the results on the geometry of the affine Grassmannian of $\GL_K$ to finish the calculation. We present the result here first and explain the calculation in steps afterwards.

\begin{theorem}\label{thr:the Hilbert series of C[M(N,K)]}
The Hilbert series of $\Gamma(N,K,n)=\Gamma(\mathcal{M}^\zeta(N,K),\mathsf{Det}^{\otimes n})$ is
\begin{align}
    Z^{(n)}_{N,K}(x;q,t)=\frac{1}{(q;q)_N}\sum_{\mu}t^{|\mu|}H_{\mu+(n^N)}(x;q)h_{\mu_1}(x^{-1})\cdots h_{\mu_N}(x^{-1}).
\end{align}
Here the summation is over arrays $\mu=(\mu_1,\cdots,\mu_N)\in \mathbb Z_{\ge 0}^N$, $(n^N)$ is the array consisting of $N$ copies of $n$, i.e. $(n^N)=(n,n,\cdots,n)$, and $|\mu|:=\sum_{i=1}^N\mu_i$, and $h_{k}(x)$ is the complete homogeneous symmetric polynomials of degree $k$, and $H_{\lambda}(x;q)$ is the generalized transformed Hall--Littlewood polynomial of the array $\lambda$, defined in \eqref{definition: transformed HL}.
\end{theorem}

\subsection{Reduction Steps}
Recall that the stability condition in the definition of $\mathcal{M}^\zeta(N,K)$ is that if $V\subset\mathbb C^N$, $B(V)\subseteq V$ and $\im(J)\subseteq V$ then $V=\mathbb C^N$, in particular the sub-quiver consisting of arrows $(B,J)$ is stable under the same stability condition, so we have:
\begin{lemma}\label{lemma: bundle over quot scheme}
The moduli of stable representations $\mathcal{M}^\zeta(N,K)$ is a vector bundle over the Quot scheme of $\mathbb A^1$ which parametrizes length $N$ quotients of $\mathcal O_{\mathbb A^1}^{\oplus K}$, denoted by $\mathrm{Quot}^N(\mathbb A^1,\mathcal O_{\mathbb A^1}^{\oplus K}):$
\begin{center}
\begin{tikzcd}
 \mathcal{M}^\zeta(N,K)=\mathbb V(\mathcal V\otimes W^*) \arrow[d,"p"] \\
 \mathrm{Quot}^N(\mathbb A^1,\mathcal O_{\mathbb A^1}^{\oplus K})
\end{tikzcd}
\end{center}
Here $\mathcal V$ is the tautological sheaf on $\mathrm{Quot}^N(\mathbb A^1,\mathcal O_{\mathbb A^1}^{\oplus K})$, and $W$ is the framing vector space.
\end{lemma}

\begin{proof}
Consider a point $(B,J,I)\in \mathcal{M}^\zeta(N,K)$, the action $B$ on $\mathbb C^N$ makes it into a $\mathbb C[z]$-module such that $z$ acts as $B$. The stability on $(B,J)$ is equivalent to that $\mathbb C^N$ is a quotient module of a free module of rank $K$. This gives rise to a morphism $p:\mathcal{M}^\zeta(N,K)\to \mathrm{Quot}^N(\mathbb A^1,\mathcal O_{\mathbb A^1}^{\oplus K})$, and the extra information in $\mathcal{M}^\zeta(N,K)$ compared to the Quot scheme is a homomorphism from the universal quotient $\mathcal V$ to the framing vector space $W$, so $\mathcal{M}^\zeta(N,K)$ is represented by $\mathbb V(\mathcal V\otimes W^*)$.
\end{proof}

\noindent Lemma \ref{lemma: bundle over quot scheme} implies that
\begin{align}
    \chi(\mathcal{M}^\zeta(N,K),\mathsf{Det}^{\otimes n})=\sum_{m=0}^{\infty}t^m\chi(\mathrm{Quot}^N(\mathbb A^1,\mathcal O_{\mathbb A^1}^{\oplus K}),\mathrm{Sym}^m(\mathcal V\otimes W^*)\otimes \mathsf{Det}^{\otimes n}).
\end{align}
Here in each summand, $\chi(\mathrm{Quot}^N(\mathbb A^1,\mathcal O_{\mathbb A^1}^{\oplus K}),\mathrm{Sym}^m(\mathcal V\otimes W^*))$ is in $K_{\GL_K}(\mathrm{pt})[\![q]\!]$. So the computation of $\mathbb C[\mathcal M(N,K)]$ boils down to the computation of equivariant Euler characters of sheaves on the Quot scheme.

The Quot scheme has a nice structure: there is morphism $h:\mathrm{Quot}^N(\mathbb A^1,\mathcal O_{\mathbb A^1}^{\oplus K})\longrightarrow \mathbb A^{(N)}$ where $\mathbb A^{(N)}$ is the $N$-th symmetric product of $\mathbb A^{1}$, which is identified with the Hilbert scheme of $N$ points on $\mathbb A^{1}$ and $h$ is the Hilbert-Chow morphism for the Quot scheme. In the language of quivers, $h$ maps $(B,J)$ to the set of eigenvalues of $B$ (counted with multiplicities), regarded as a divisor of degree $N$ in $\mathbb A^1$.

\begin{lemma}\label{Lemma: central fiber}
The central fiber $h^{-1}(0)$ of the Hilbert-Chow morphism $h:\mathrm{Quot}^N(\mathbb A^1,\mathcal O_{\mathbb A^1}^{\oplus K})\longrightarrow \mathbb A^{(N)}$, endowed with reduced scheme structure, is isomorphic to the affine Grassmannian Schubert variety $\overline{\Gr}^{N\omega_1}_{\GL_K}$ \footnote{Here we endow all affine Grassmannian Schubert varieties with reduced scheme structures.}. Here $\omega_1=(1,0,\cdots,0)$ is the first fundamental coweight of $\GL_K$.
\end{lemma}

\begin{proof}
The central fiber $h^{-1}(0)$ is the moduli space of submodules of $\mathbb C[z]^{\oplus K}$ whose cokernels are finite of length $N$ and are supported at $0$. Since finite quotients of $\mathbb C[z]^{\oplus K}$ are equivalent to finite quotients of $\mathbb C[\![z]\!]^{\oplus K}$, $h^{-1}(0)$ is the moduli space of submodules of $\mathbb C[\![z]\!]^{\oplus K}$ whose cokernels are finite of length $N$. This implies that $h^{-1}(0)$ is isomorphic to a closed subscheme of $\Gr_{\GL_K}$ which is set-theoretically identified with $\overline{\Gr}^{N\omega_1}_{\GL_K}$, thus $h^{-1}(0)_{\mathrm{red}}\cong \overline{\Gr}^{N\omega_1}_{\GL_K}$.
\end{proof}

\begin{proposition}\label{prop: flatness of Hilbert-Chow}
The Hilbert-Chow morphism $h:\mathrm{Quot}^N(\mathbb A^1,\mathcal O_{\mathbb A^1}^{\oplus K})\longrightarrow \mathbb A^{(N)}$ is flat.
\end{proposition}

\begin{proof}
By the deformation theory, $\mathrm{Quot}^N(\mathbb A^1,\mathcal O_{\mathbb A^1}^{\oplus K})$ is smooth of dimension $NK$. $h^{-1}(0)_{\mathrm{red}}\cong \overline{\Gr}^{N\omega_1}_{\GL_K}$ has dimension $(K-1)N$, which equals to $\dim \mathrm{Quot}^N(\mathbb A^1,\mathcal O_{\mathbb A^1}^{\oplus K})-\dim \mathbb A^{(N)}$, thus $h$ is flat along $h^{-1}(0)$ by miracle flatness theorem \cite[\href{https://stacks.math.columbia.edu/tag/00R4}{Tag 00R4}]{stacks-project}. Since flatness is an open condition, $h$ is flat in an open neighborhood of $h^{-1}(0)$. Since Hilbert-Chow morphism $h$ is proper, there is an open neighborhood $U$ of $0\in \mathbb A^{(N)}$ such that $h|_{h^{-1}(U)}$ is flat. Finally $h$ is equivariant under the $\mathbb C^{\times}$ action on $\mathbb A^1$ which contracts $\mathbb A^{(N)}$ to $0$, so the flatness is transported from $U$ to the whole $\mathbb A^{(N)}$.
\end{proof}

Proposition \ref{prop: flatness of Hilbert-Chow} provides a tool that reduces the computation of Euler character to the central fiber. In effect, to compute $\chi(\mathrm{Quot}^N(\mathbb A^1,\mathcal O_{\mathbb A^1}^{\oplus K}),\mathcal F)$ for a locally free sheaf $\mathcal F$, we can apply $\mathbb C^{\times}_q$-localization to its derived pushforward $Rh_*(\mathcal F)$:
\begin{align}\label{eqn: localization on A^N}
    \chi(\mathrm{Quot}^N(\mathbb A^1,\mathcal O_{\mathbb A^1}^{\oplus K}),\mathcal F)=\chi(\mathbb A^{(N)},\mathbf Rh_*(\mathcal F))=\frac{\chi(h^{-1}(0),\mathcal F|_{h^{-1}(0)})}{\prod_{i=1}^N(1-q^i)},
\end{align}
where in the last equation we use the proper base change (since $\mathcal F$ is flat over $\mathbb A^{(N)}$ by Proposition \ref{prop: flatness of Hilbert-Chow}), and the denominator comes from the tangent space of $\mathbb A^{(N)}$ at $0$ which has $\mathbb C^{\times}_q$-weights $-1,\cdots,-N$.

\begin{proposition}\label{prop: central fiber}
The central fiber $h^{-1}(0)$ is isomorphic to $\overline{\Gr}^{N\omega_1}_{\GL_K}$ as a scheme.
\end{proposition}

\begin{proof}
In view of Lemma \ref{Lemma: central fiber}, the proposition is equivalent to that $h^{-1}(0)$ is reduced. Since $h$ is flat with domain and codomain being smooth, $h^{-1}(0)$ is a Cohen-Macaulay scheme, therefore it is enough to show that $h^{-1}(0)$ is generically reduced. We claim that $h$ is smooth at the point $z^{N \omega_1}$. Assume that the claim is true, then $h$ is smooth in an open neighborhood of $z^{N \omega_1}$, thus $h^{-1}(0)$ is generically reduced.

The claim follows from the deformation theory of $\mathrm{Quot}^N(\mathbb A^1,\mathcal O_{\mathbb A^1}^{\oplus K})$. Namely, if $e_1,\cdots,e_K$ is the basis of $\mathcal O_{\mathbb A^1}^{\oplus K}$, then $z^{N\omega_1}$ corresponds to short exact sequence
\begin{align*}
    0\longrightarrow \mathcal E\longrightarrow \mathcal O_{\mathbb A^1}^{\oplus K}\longrightarrow Q\longrightarrow 0
\end{align*}
such that $\mathcal E$ is the subsheaf of $\mathcal O_{\mathbb A^1}^{\oplus K}$ generated by $z^{N}e_1,e_2,\cdots,e_K$. Then the tangent space of $\mathrm{Quot}^N(\mathbb A^1,\mathcal O_{\mathbb A^1}^{\oplus K})$ at $z^{N\omega_1}$ is 
\begin{align*}
    \Hom_{\mathcal O_{\mathbb A^1}}(\mathcal E,Q).
\end{align*}
In particular, the tangent space contains $\Hom_{\mathcal O_{\mathbb A^1}}(z^{N}\mathcal O_{\mathbb A^1}, \mathcal O_{\mathbb A^1}/z^{N}\mathcal O_{\mathbb A^1})$ as a subspace, and the latter projects isomorphically onto the tangent space of $\mathbb A^{(N)}$ at $0$. In particular, the tangent map at $z^{N\omega_1}$ is surjective, thus $h$ is smooth at $z^{N\omega_1}$.
\end{proof}

By the Lemma \ref{lem: cohomology vanishing 1} and localization formula \eqref{eqn: localization on A^N}, we reduce the calculation to
\begin{align}
Z^{(n)}_{N,K}(x;q,t)=\chi(\mathcal{M}^\zeta(N,K),\mathsf{Det}^{\otimes n})
    =\frac{1}{(q;q)_N}\sum_{m=0}^{\infty}t^m\chi(\overline{\Gr}^{N\omega_1}_{\GL_K},\mathrm{Sym}^m(\mathcal V\otimes W^*)\otimes \mathcal O(n)),
\end{align}
where $\mathcal V$ is the restriction of the universal quotient sheaf to $\overline{\Gr}^{N\omega_1}_{\GL_K}$.

\subsection{Calculation on Affine Grassmannian} It remains to do the calculation on affine Grassmannian for 
\begin{align*}
    \sum_{m=0}^{\infty}t^n\chi(\overline{\Gr}^{N\omega_1}_{\GL_K},\mathrm{Sym}^m(\mathcal V\otimes W^*)\otimes \mathcal O(n))=\chi(\overline{\Gr}^{N\omega_1}_{\GL_K},S^{\bullet}_t(\mathcal V\otimes W^*)\otimes \mathcal O(n)).
\end{align*}
Here we use the notation $S^{\bullet}_t(\mathcal V\otimes W^*)=\bigoplus_{m\ge 0} t^n \:\mathrm{Sym}^m(\mathcal V\otimes W^*)$.
To start with, note that there is a convolution map on $\Gr_{\GL_K}$:
\begin{align*}
    m:\Gr^{\omega_1}_{\GL_K}\widetilde{\times} \:\overline{\Gr}^{(N-1)\omega_1}_{\GL_K}\longrightarrow \overline{\Gr}^{N\omega_1}_{\GL_K},
\end{align*}
see appendix \eqref{eqn: convolution} for definition of the convolution product. The key property of the convolution product is that 
\begin{align}\label{eqn: pushforwrd of structure sheaf}
    \mathcal O\cong \mathbf Rm_*\mathcal O.
\end{align}
See the proof of appendix \ref{corollary: character of O(k)} for an explanation of this isomorphism. Here $\mathcal O$ is the structure sheaves, we omit the subscripts labelling the domain and codomain, since the meaning of the homomorphism is clear. In view of \eqref{eqn: pushforwrd of structure sheaf}, we have
\begin{align*}
    \chi(\overline{\Gr}^{N\omega_1}_{\GL_K},S^{\bullet}_t(\mathcal V\otimes W^*)\otimes \mathcal O(n))=\chi\left(\Gr^{\omega_1}_{\GL_K}\widetilde{\times} \:\overline{\Gr}^{(N-1)\omega_1}_{\GL_K},S^{\bullet}_t(m^*\mathcal V\otimes W^*)\otimes \mathcal O(n)\right).
\end{align*}
Let us write $\mathcal V_N$ for $\mathcal V$ to indicate the rank of the gauge group.
\begin{lemma}\label{lemma: pullback of V}
There is a short exact sequence
\begin{align*}
    0\longrightarrow \widetilde {\mathcal V}_{N-1}\longrightarrow m^*\mathcal V_N \longrightarrow p^*\mathcal V_1\longrightarrow 0,
\end{align*}
where $p:\Gr^{\omega_1}_{\GL_K}\widetilde{\times}\: \overline{\Gr}^{(N-1)\omega_1}_{\GL_K}\to \overline{\Gr}^{N\omega_1}_{\GL_K}$ is the projection to the first component map, and $\widetilde {\mathcal V}_{N-1}$ is the sheaf $\GL_K(\mathscr K) \overset{\GL_n(\mathscr O)}{\times} {\mathcal V}_{N-1}$.
\end{lemma}

\begin{proof}
$\mathcal V_N$ is the universal quotient of $\mathbb C[\![z]\!]^{\oplus K}$. Denote the kernel by $L_N$. Then the pullback of $\mathcal V_N$ to the twisted product $\Gr^{\omega_1}_{\GL_K}\widetilde{\times}\: \overline{\Gr}^{(N-1)\omega_1}_{\GL_K}$ is by definition the extension of $\mathcal V_1$ by $\mathcal V_{N-1}$, except that in the definition of $\mathcal V_{N-1}$ the free module $\mathbb C[\![z]\!]^{\oplus K}$ is replaced by $L_1$ (this is the meaning of twist).
\end{proof}

Note that $\mathcal V_1$ is of rank one, so it is by definition the determinant line bundle $\mathcal O(1)$ on the affine Grassmannian $\Gr_{\GL_K}$ restricted on $\Gr^{\omega_1}_{\GL_K}$. The convolution map easily generalizes to multiple copies of $\Gr_{\GL_K}$:
\begin{align*}
    m: \Gr^{\omega_1}_{\GL_K}\widetilde{\times} \cdots\widetilde{\times} \Gr^{\omega_1}_{\GL_K}\longrightarrow \overline{\Gr}^{\omega_1}_{\GL_K},
\end{align*}
and we can apply Lemma \ref{lemma: pullback of V} recursively and see that $m^*\mathcal V_N$ is a consecutive extension of (twisted) $\mathcal O(1)$. Since we only care about the Euler character, we can forget about the extension structure and focus on the $K$-theory class, in other words, we have:
\begin{align}
    \begin{split}
        &\chi(\overline{\Gr}^{N\omega_1}_{\GL_K},S^{\bullet}_t(\mathcal V\otimes W^*)\otimes \mathcal O(n))\\
        =&\chi(\Gr^{\omega_1}_{\GL_K}\widetilde{\times} \cdots\widetilde{\times} \Gr^{\omega_1}_{\GL_K},S^{\bullet}_t((\mathcal O(1)+\widetilde {\mathcal O(1)}+\cdots+\widetilde {\mathcal O(1)})\otimes W^*)\otimes (\mathcal O(n)\widetilde{\boxtimes}\cdots \widetilde{\boxtimes}\mathcal O(n)))\\
        =&\sum_{\mu}t^{|\mu|}\chi(\Gr^{\omega_1}_{\GL_K}\widetilde{\times} \cdots\widetilde{\times} \Gr^{\omega_1}_{\GL_K},\mathcal O(\mu_1+n)\widetilde{\boxtimes}\cdots \widetilde{\boxtimes}\mathcal O(\mu_N+n))\chi(S^{\mu_1}(W^*))\cdots\chi(S^{\mu_N}(W^*)).
    \end{split}
\end{align}
Here the summation is over arrays $\mu=(\mu_1,\cdots,\mu_N)\in \mathbb Z_{\ge 0}^N$, $|\mu|=\sum_{i=1}^N\mu_i$, and $\chi(S^{k}(W^*))$ is the $\GL_K$-equivariant $K$-theory class of the $k$-th symmetric tensor product of $W^*$, where $W$ is the fundamental representation of $\GL_K$. It is well-known that $\chi(S^{k}(W^*))=h_{k}(x^{-1})$, where $h_{k}(x)$ is the complete homogeneous symmetric polynomial of degree $k$. Finally, the remaining part of the computation, which is the character of $\mathcal O(\mu_1+n)\widetilde{\boxtimes}\cdots \widetilde{\boxtimes}\mathcal O(\mu_N+n)$, is related to a well-understood family of symmetric functions, the transformed Hall--Littlewood polynomial. In fact we have
\begin{align}
    \chi(\Gr^{\omega_1}_{\GL_K}\widetilde{\times} \cdots\widetilde{\times} \Gr^{\omega_1}_{\GL_K},\mathcal O(\mu_1+n)\widetilde{\boxtimes}\cdots \widetilde{\boxtimes}\mathcal O(\mu_N+n))=H_{\mu+(n^N)}(x;q).
\end{align}
where $H_{\mu}(x;q)$ is the generalized transformed Hall--Littlewood polynomial of the array $\mu+(n^N)$, see \eqref{definition: transformed HL} for definition. For the derivation of this formula, see Corollary \ref{corollary: character of convolution} in the appendix.

\subsection{\texorpdfstring{$N\to \infty$}{N to infinity} limit} Recall that $\mathbb C[\mathcal{M}(\infty,K)]$ is the subalgebra of $\underset{\substack{\longleftarrow\\N}}{\lim}\: \mathbb C[\mathcal M(N,K)]$ generated by $T^{(n)}_{ab}$ and $\Tr(B^m)$, for all $n,m\in \mathbb Z_{\ge 0}$ and $1\le a,b\le K$ (Definition \ref{defn: large N limit of M(N,K)}).

\begin{lemma}\label{lemma: eigenvectors in the limit}
$\mathbb C[\mathcal{M}(\infty,K)]$ contains all $T\times \mathbb C^{\times}_q\times \mathbb C^{\times}_t$ eigenvectors in $\underset{\substack{\longleftarrow\\N}}{\lim}\: \mathbb C[\mathcal M(N,K)]$, where $T\subset \GL_K$ is the maximal torus.
\end{lemma}
\begin{proof}
We claim that for fixed $n\in\mathbb Z_{\ge 0}$, the dimension of $\mathbb C^{\times}_q$-weight $n$ space of $\mathbb C[\mathcal M(N,K)]$ stabilizes when $N\gg 0$, more precisely there exists $N$ such that for all $N'>N$ the kernel of $\mathbb C[\mathcal M(N',K)]\twoheadrightarrow \mathbb C[\mathcal M(N,K)]$ has $\mathbb C^{\times}_q$-weights $>n$. To see why this is true, we take $N$ such that $L^-_{n}\GL_K\times L^-_n\GL_1\subseteq\mathcal M(N,K)$ ($N$ can be $(n+1)K$ according to the proof of Theorem \ref{theorem: density 2}), then $\ker(\mathbb C[\mathcal M(N',K)]\twoheadrightarrow \mathbb C[\mathcal M(N,K)])$ is a subquotient of $\ker(\mathbb C[L^-\GL_K\times L^-\GL_1]\twoheadrightarrow \mathbb C[L^-_{n}\GL_K\times L^-_n\GL_1])$, and the latter is an ideal generated by elements of $\mathbb C^{\times}_q$-weights greater than $n$. 

Now assume that $a\in \underset{\substack{\longleftarrow\\N}}{\lim}\: \mathbb C[\mathcal M(N,K)]$ is a $T\times \mathbb C^{\times}_q\times \mathbb C^{\times}_t$ eigenvector, and let its $\mathbb C^{\times}_q$ be $n$. Then there exists $N$ such that for all $N'>N$ the kernel of $\mathbb C[\mathcal M(N',K)]\twoheadrightarrow \mathbb C[\mathcal M(N,K)]$ has $\mathbb C^{\times}_q$-weights greater than $n$. Consider the image of $a$ in $\mathbb C[\mathcal M(N,K)]$, denoted by $\overline{a}$, and take a $T\times \mathbb C^{\times}_q\times \mathbb C^{\times}_t$-equivariant lift of $\overline{a}$ along the projection $\mathbb C[L^-\GL_K\times L^-\GL_1]\twoheadrightarrow \mathbb C[\mathcal M(N,K)]$, and we denote the lift by $a'$, then $a-a'$ has $\mathbb C^{\times}_q$-weight $n$ and is zero in $\mathbb C[\mathcal M(N,K)]$, thus $a-a'$ is in the kernel of $\mathbb C[\mathcal M(N',K)]\twoheadrightarrow \mathbb C[\mathcal M(N,K)]$ for all $N'>N$, which forces $a=a'$ in $\mathbb C[\mathcal M(N',K)]$ because of weight consideration, therefore $a=a'$ in $\underset{\substack{\longleftarrow\\N}}{\lim}\: \mathbb C[\mathcal M(N,K)]$.
\end{proof}

\begin{proposition}\label{prop:the Hilbert series in the large-N limit}
The Hilbert series of $\mathbb C[\mathcal{M}(\infty,K)]$ equals to the $N\to \infty$ limit of Hilbert series of $\mathbb C[\mathcal{M}(N,K)]$, i.e.
\begin{align}\label{eqn: Hilbert series infinity 1}
    \mathbb C[\mathcal{M}(\infty,K)]=\frac{1}{(q;q)_{\infty}}\underset{N\to \infty}{\lim}\sum_{\mu}t^{|\mu|}H_{\mu}(x;q)h_{\mu_1}(x^{-1})\cdots h_{\mu_N}(x^{-1})
\end{align}
\end{proposition}

\begin{proof}
The $N\to \infty$ limit of Hilbert series of $\mathbb C[\mathcal{M}(N,K)]$ enumerates $T\times \mathbb C^{\times}_q\times \mathbb C^{\times}_t$ eigenvectors in $\underset{\substack{\longleftarrow\\N}}{\lim}\: \mathbb C[\mathcal M(N,K)]$, which is the same as $T\times \mathbb C^{\times}_q\times \mathbb C^{\times}_t$ eigenvectors in $\mathbb C[\mathcal{M}(\infty,K)]$, by Lemma \ref{lemma: eigenvectors in the limit}.
\end{proof}

On the other hand, $\mathbb C[\mathcal{M}(\infty,K)]$ is freely generated by $I_aB^nJ_{b},\Tr(B^m)$, which makes its Hilbert series easily computed by
\begin{align}\label{eqn: Hilbert series infinity 2}
    \mathrm{PE}\left((t+tq+tq^2+\cdots)\chi(\mathfrak{gl}_K)\right)\mathrm{PE}(q+q^2+\cdots).
\end{align}
Here $\chi(\mathfrak{gl}_K)$ is the character of the adjoint representation of $\GL_K$, and $\mathrm{PE}$ is the plethestic exponential. Note that $\chi(\mathfrak{gl}_K)$ can be written as a symmetric function $1+\frac{s_{\lambda_{\mathrm{ad}}}(x)}{m_K(x)}$, where $\lambda_{\mathrm{ad}}$ is the Young tableaux corresponding to the adjoint representation of $\SL_K$, and $s_{\lambda_{\mathrm{ad}}}(x)$ is the Schur function associated to $\lambda_{\mathrm{ad}}$, and $m_K(x)=x_1x_2\cdots x_K$. Moreover,
\begin{align*}
    \mathrm{PE}(q+q^2+\cdots)=\prod_{i=1}^{\infty}\frac{1}{1-q^i}=\frac{1}{(q;q)_{\infty}}.
\end{align*}
Compare equation \ref{eqn: Hilbert series infinity 1} with \ref{eqn: Hilbert series infinity 2}, we get the following interesting equation, which we do not know other way to prove.
\begin{corollary}
\begin{align}\label{eqn: large N limit of Hilbert series}
    \underset{N\to \infty}{\lim} \sum_{\mu}t^{|\mu|}H_{\mu}(x;q)h_{\mu_1}(x^{-1})\cdots h_{\mu_N}(x^{-1})=\mathrm{PE}\left(\frac{t}{1-q}\left(1+\frac{s_{\lambda_{\mathrm{ad}}}(x)}{m_K(x)}\right)\right).
\end{align}
Here $\lambda_{\mathrm{ad}}$ is the Young tableaux corresponding to the adjoint representation of $\SL_K$, and $s_{\lambda_{\mathrm{ad}}}(x)$ is the Schur function associated to $\lambda_{\mathrm{ad}}$, $m_K(x)=x_1x_2\cdots x_K$, and $\mathrm{PE}$ is the plethestic exponential.
\end{corollary}

\section{Quantization of \texorpdfstring{$\cM(N,K)$}{}}
\label{sec:quanization of phase space}

In this section we study the quantization of $\cM(N,K)$, namely we quantize the Poisson structure \eqref{Poisson before quotient} to the commutation relation:
\begin{align}\label{commutation relation}
    [J_{ia},I_{bj}]=\hbar\delta_{ab}\delta_{ij},\; [B_{mn},B_{pq}]=\hbar(\delta_{pn}B_{mq}-\delta_{mq}B_{pn}),\;[B_{mn},I_{bj}]=[B_{mn},J_{ia}]=0.
\end{align}
This is the algebra $U_{\hbar}(\mathfrak{gl}_N)\otimes \mathrm{Weyl}^{\otimes NK}_{\hbar}$, and we define the quantized ring of functions on the phase space $\mathbb C_{\hbar}[\cM(N,K)]$ by the invariant part $(U_{\hbar}(\mathfrak{gl}_N)\otimes \mathrm{Weyl}^{\otimes NK}_{\hbar})^{\GL_N}$. Since $\GL_N$ is reductive, we have $\mathbb C_{\hbar}[\cM(N,K)]/(\hbar)=\mathbb C[\cM(N,K)]$. Note that $\mathbb C_{\hbar}[\cM(N,K)]$ possesses a natural grading by setting
\begin{align}\label{grading}
    \deg(J)=0,\deg(I)=1,\deg(B)=1,\deg(\hbar)=1.
\end{align}

\begin{lemma}\label{lemma: flatness}
$\mathbb C_{\hbar}[\cM(N,K)]$ is flat over $\mathbb C[\hbar]$.
\end{lemma}

\begin{proof}
Since $U_{\hbar}(\mathfrak{gl}_N)\otimes \mathrm{Weyl}^{\otimes NK}_{\hbar}$ is flat over $\mathbb C[\hbar]$, the subalgebra $\mathbb C_{\hbar}[\cM(N,K)]$ is $\hbar$-torsion free, thus it is also flat over $\mathbb C[\hbar]$.
\end{proof}

\begin{remark}
On the moduli of $\zeta$-stable representations $\cM^\zeta(N,K)$, there is a notion of quantized structure sheaf. Namely, consider the completion of $U_{\hbar}(\mathfrak{gl}_N)\otimes \mathrm{Weyl}^{\otimes NK}_{\hbar}$ in the $\hbar$-adic topology, then localize this algebra in the Zariski topology of the affine space $\mathrm{Rep}(N,K)$, and by taking $\GL_N$-invariant on the open locus of stable representations $\mathrm{Rep}(N,K)^{\zeta-s}$, we get a sheaf of flat $\mathbb C[\![\hbar]\!]$-algebras on $\cM^\zeta(N,K)$, denoted by $\widehat{\mathcal O}_{\cM^\zeta(N,K)}$. By construction we have $\widehat{\mathcal O}_{\cM^\zeta(N,K)}/(\hbar)=\mathcal O_{\cM^\zeta(N,K)}$. This sheaf is related to $\mathbb C_{\hbar}[\cM(N,K)]$ as follows. By construction we have a natural homomorphism of algebras $\mathbb C_{\hbar}[\cM(N,K)]\to \Gamma(\cM^\zeta(N,K),\widehat{\mathcal O}_{\cM^\zeta(N,K)})$, which preserves the grading \eqref{grading}. On the other hand, by Proposition \ref{prop: res of sing}, we have
\begin{itemize}
    \item $H^i(\cM^\zeta(N,K),\widehat{\mathcal O}_{\cM^\zeta(N,K)})=0$, for $i>0$.
    \item $\Gamma(\cM^\zeta(N,K),\widehat{\mathcal O}_{\cM^\zeta(N,K)})$ is a flat $\mathbb C[\![\hbar]\!]$-algebra, and $\Gamma(\cM^\zeta(N,K),\widehat{\mathcal O}_{\cM^\zeta(N,K)})/(\hbar)=\mathbb C[\cM(N,K)]$.
\end{itemize}
Let us give a $\bZ_{\ge 0}$-grading to $U_{\hbar}(\mathfrak{gl}_N)\otimes \mathrm{Weyl}^{\otimes NK}_{\hbar}$ by setting
\begin{align*}
    \deg I=\deg J=1, \;\deg B=\deg \hbar=2.
\end{align*}
Since $\mathbb C_{\hbar}[\cM(N,K)]$ is generated by positive degree elements, we conclude that $\mathbb C_{\hbar}[\cM(N,K)]$ is naturally identified with the $\bC[\hbar]$-subalgebra of $\Gamma(\cM^\zeta(N,K),\widehat{\mathcal O}_{\cM^\zeta(N,K)})$ generated by homogeneous elements.
\end{remark}

The set of elements $$\{\mathrm{Tr}(B^n),T^{(m)}_{ab}:=I_aB^nJ_b: 1\le n\le N,0\le m\le N-1,1\le a,b\le K\}$$ generates $\mathbb C[\cM(N,K)]$ (\S\ref{subsec:generators}), so it also generate $\mathbb C_{\hbar}[\cM(N,K)]$. It is a straightforward computation to see that $\mathrm{Tr}(B^k)$ commutes with all elements in $U_{\hbar}(\mathfrak{gl}_N)\otimes \mathrm{Weyl}^{\otimes NK}_{\hbar}$, therefore $\mathrm{Tr}(B^k)$ is central in $\mathbb C_{\hbar}[\cM(N,K)]$. We use $T_{ab}(z)$ to denote the following power series:
\begin{align*}
    T_{ab}(z):=\sum_{n\ge -1} T^{(n)}_{ab}z^{-n-1}=\delta_{ab}+I_a\frac{1}{z-B}J_b.
\end{align*}
\begin{proposition}\label{prop: commutator}
The commutators between $T_{ab}^{(k)}$ are:
\begin{align}\label{eqn: commutators}
    [T^{(p)}_{ab},T^{(q)}_{cd}]=\hbar\sum_{i=-1}^{\min(p,q)-1}\left(T^{(i)}_{cb}T^{(p+q-1-i)}_{ad}-T^{(p+q-1-i)}_{cb}T^{(i)}_{ad}\right).
\end{align}
\end{proposition}

\begin{proof}
It is easy to see that \eqref{eqn: commutators} is equivalent to
\begin{align*}
    [T^{(p+1)}_{ab},T^{(q)}_{cd}]-[T^{(p)}_{ab},T^{(q+1)}_{cd}]=\hbar\left(T^{(p)}_{cb}T^{(q)}_{ad}-T^{(q)}_{cb}T^{(p)}_{ad}\right).
\end{align*}
We compute the left-hand-side of the above equation:
\begin{align*}
    &[T^{(p+1)}_{ab},T^{(q)}_{cd}]-[T^{(p)}_{ab},T^{(q+1)}_{cd}]=I_{am}I_{cr}\left([(B^{p+1})_{mn},(B^{q})_{rs}]-[(B^{p})_{mn},(B^{q+1})_{rs}]\right)J_{bn}J_{ds}\\
    &~=\hbar I_{am}I_{cr}\left(\sum_{i=1}^q(B^{i-1})_{rn}(B^{p+1+q-i})_{ms}-(B^{i+p})_{rn}(B^{q-i})_{ms}\right)J_{bn}J_{ds}\\
    &~-\hbar I_{am}I_{cr}\left(\sum_{i=1}^{q+1}(B^{i-1})_{rn}(B^{p+1+q-i})_{ms}-(B^{i+p-1})_{rn}(B^{q+1-i})_{ms}\right)J_{bn}J_{ds}\\
    &~=\hbar I_{am}I_{cr}\left((B^p)_{rn}(B^q)_{ms}-(B^q)_{rn}(B^p)_{ms}\right)J_{bn}J_{ds},
\end{align*}
which equals to the right-hand-side.
\end{proof}

\begin{remark}\label{Remark: RTT}
The commutators \eqref{eqn: commutators} is equivalent to the RTT equation
\begin{align}\label{eqn: Poisson structure of RTT 2}
    [T_{ab}(u),T_{cd}(v)]=\frac{\hbar}{u-v}\left(T_{cb}(u)T_{ad}(v)-T_{cb}(v)T_{ad}(u)\right).
\end{align}
\end{remark}

The classical embedding $\cM(L,K)\hookrightarrow \cM(N,K)$ for $L<N$ can be quantized as follows. Consider the left ideal of $U_{\hbar}(\mathfrak{gl}_N)\otimes \mathrm{Weyl}^{\otimes NK}_{\hbar}$ generated by $B_{ij}$ and $J_{ia}$ for all $L<i,j\le N$ and $1\le a\le K$, denote it by $I^0_{L,N}$, then $I_{L,N}:=\left(U_{\hbar}(\mathfrak{gl}_N)\otimes \mathrm{Weyl}^{\otimes NK}_{\hbar}\right)^{\GL_{N-L}}\cap I^0_{L,N}$ is a two-sided ideal in $\left(U_{\hbar}(\mathfrak{gl}_N)\otimes \mathrm{Weyl}^{\otimes NK}_{\hbar}\right)^{\GL_{N-L}}$, where $\GL_{N-L}$ acts on indices $L<i,j\le N$. It is easy to see that 
\begin{align*}
    \left(U_{\hbar}(\mathfrak{gl}_N)\otimes \mathrm{Weyl}^{\otimes NK}_{\hbar}\right)^{\GL_{N-L}}=\left(U_{\hbar}(\mathfrak{gl}_L)\otimes \mathrm{Weyl}^{\otimes LK}_{\hbar}\right)\oplus I_{L,N}
\end{align*}
as vector spaces, thus $U_{\hbar}(\mathfrak{gl}_L)\otimes \mathrm{Weyl}^{\otimes LK}_{\hbar}=\left(U_{\hbar}(\mathfrak{gl}_N)\otimes \mathrm{Weyl}^{\otimes NK}_{\hbar}\right)^{\GL_{N-L}}/I_{L,N}$. Restricting to the subalgebra $\mathbb C_{\hbar}[\cM(N,K)]=\left(U_{\hbar}(\mathfrak{gl}_N)\otimes \mathrm{Weyl}^{\otimes NK}_{\hbar}\right)^{\GL_{N}}$, we get a map between graded $\bC[\hbar]$ algebras
\begin{align}\label{eq:quantized truncation map}
\begin{split}
&\tau^N_L:\mathbb C_{\hbar}[\cM(N,K)]\to \mathbb C_{\hbar}[\cM(L,K)],\\
\tau^N_L&(T(u))=T(u),\quad \tau^N_L(\Tr(B^n))=\Tr(B^n).
\end{split}
\end{align}
which quantizes the embedding $\cM(L,K)\hookrightarrow \cM(N,K)$. We note that $\tau^N_L$ is surjective.

\begin{definition}
Let $Y_{\hbar}(\mathfrak {gl}_K)$ be the Yangian of $\fgl_K$ with generators $\{\mathfrak T_{ab}(u):=\delta_{ab}+\sum_{n=0}^\infty \mathfrak T^{(n)}_{ab}u^{-n-1}:1\le a,b\le K\}$ subject to the RTT relations
$$[\mathfrak T_{ab}(u),\mathfrak T_{cd}(v)]=\frac{\hbar}{u-v}\left(\mathfrak T_{cb}(u)\mathfrak T_{ad}(v)-\mathfrak T_{cb}(v)\mathfrak T_{ad}(u)\right).$$
Define $\Lambda$ as the polynomial ring of countably many generators $\bC[p_i:i\in \bZ_{\ge 1}]$, which is naturally identified with the ring of symmetric polynomials with $p_i$ being the $i$-th power sum function.
\end{definition}

\begin{proposition}\label{prop:the surjective map between Yangian and deformed ring of functions}
For every $N\in \bZ_{>0}$ there is a surjective map of $\bC[\hbar]$ algebras
\begin{align}
\begin{split}
&\rho_N:Y_\hbar(\fgl_K)\otimes \Lambda\twoheadrightarrow \mathbb C_{\hbar}[\cM(N,K)],\\
&\rho_N(\mathfrak{T}(u))=T(u),\quad \rho_N(p_n)=\Tr(B^n).
\end{split}
\end{align}
These maps are compatible in the sense that for all $N>L$ the diagram
\begin{center}
    \begin{tikzcd}
    Y_\hbar(\fgl_K)\otimes \Lambda \arrow[r,"\rho_N"] \arrow[dr,"\rho_L" '] & \mathbb C_{\hbar}[\cM(N,K)] \arrow[d,"\tau^N_L"]\\
    & \mathbb C_{\hbar}[\cM(L,K)]
    \end{tikzcd}
\end{center}
is commutative. Moreover, $\bigcap_{N=1}^{\infty}\ker\left(\rho_N\right)=0$.
\end{proposition}

\begin{proof}
According to Proposition \ref{prop: commutator}, $\rho_N$ is a $\bC[\hbar]$ algebra homomorphism. $\rho_N$ is surjective since it is surjective modulo $\hbar$. The compatibility with $\tau^N_L$ is clear from construction. The intersection of kernels is zero because $\mathbb C_{\hbar}[\cM(N,K)]$ is flat over $\mathbb C[\hbar]$ and the intersection of kernels modulo $\hbar$ is zero.
\end{proof}

\subsection{Another Map from \texorpdfstring{$Y_{\hbar}(\mathfrak {gl}_K)\otimes\Lambda$}{Yangian} to \texorpdfstring{$\mathbb C_{\hbar}[\cM(N,K)]$}{Chbar[M(N,K)]}}\label{sec:another map between Yangian and deformed ring of functions}

Recall that $\cM(N,K)$ is isomorphic to the Hamiltonian reduction description $\widetilde{\mathrm{Rep}}(N,K)/\!/\!/ \GL_N$, see \S\ref{subsec:Hamiltonian reduction}. This admits quantization. Namely, consider the algebra $U_\hbar(\mathfrak{gl}_N)\otimes U_{\hbar}(\mathfrak{gl}_N)\otimes \mathrm{Weyl}^{\otimes NK}_{\hbar}$ whose generators will be denoted as $(B,\widetilde{B}, J,I)$ with commutation relations \eqref{commutation relation} and
\begin{align}
    [\widetilde B_{ij},\widetilde B_{kl}]=\hbar(\delta_{il}\widetilde B_{kj}-\delta_{kj}\widetilde B_{il}).
\end{align}
\begin{definition}
The quantum moment map $\mu:\mathfrak{gl}_N\to U_{\hbar}(\mathfrak{gl}_N)\otimes U_{\hbar}(\mathfrak{gl}_N)\otimes \mathrm{Weyl}^{\otimes NK}_{\hbar}$ is 
\begin{align}
    \mu(E_{ij})=B_{ij}+\widetilde B_{ij}-\sum_{a=1}^KJ_{ia}I_{aj}+\hbar N\delta_{ij}.
\end{align}
And the quantum Hamiltonian reduction $(U_{\hbar}(\mathfrak{gl}_N)\otimes U_{\hbar}(\mathfrak{gl}_N)\otimes \mathrm{Weyl}^{\otimes NK}_{\hbar})/\!/\!/ \GL_N$ is defined as the $\GL_N$ invariant of $U_{\hbar}(\mathfrak{gl}_N)\otimes U_{\hbar}(\mathfrak{gl}_N)\otimes \mathrm{Weyl}^{\otimes NK}_{\hbar}$ quotient by the left ideal generated by $\mu(\mathfrak{gl}_N)$. Denote the quantum Hamiltonian reduction by $\cA_{N,K}$.
\end{definition}

Obviously there are two isomorphisms between $\mathbb C_{\hbar}[\cM(N,K)]$ and $\cA_{N,K}$, corresponding to two set of generators which are packaged in the generating functions 
\begin{align}\label{eq: two sets of gen}
\begin{split}
T_{ab}(u)=\delta_{ab}+I_a\frac{1}{u-B}J_b,&\;Z(u)=1-\mathrm{Tr}\left(\frac{\hbar}{u-B}\right),\\
\widetilde{T}_{ab}(u)=\delta_{ab}+I_a\frac{1}{u+\widetilde B}J_b,&\;\widetilde{Z}(u)=1-\mathrm{Tr}\left(\frac{\hbar}{u-\widetilde B}\right).
\end{split}
\end{align}
Then analogous to Proposition \ref{prop:the surjective map between Yangian and deformed ring of functions}, for every $N\in \bZ_{>0}$ there is a surjective map of $\bC[\hbar]$ algebras
\begin{align}
    \widetilde\rho_N:Y_\hbar(\fgl_K)\otimes \Lambda\twoheadrightarrow \mathbb C_{\hbar}[\cM(N,K)].
\end{align}
These maps are compatible in the sense that for $N>L$ the diagram
\begin{center}
    \begin{tikzcd}
    Y_\hbar(\fgl_K)\otimes \Lambda\arrow[r,"\widetilde\rho_N"] \arrow[dr,"\widetilde\rho_L" '] & \mathbb C_{\hbar}[\cM(N,K)] \arrow[d,"\widetilde\tau^N_L"]\\
    & \mathbb C_{\hbar}[\cM(L,K)]
    \end{tikzcd}
\end{center}
is commutative, where $\widetilde\tau^N_L: \mathbb C_{\hbar}[\cM(N,K)]\to \mathbb C_{\hbar}[\cM(L,K)]$ is defined similar to that of $\tau^N_L$ in \eqref{eq:quantized truncation map}, with $B$ replaced by $\widetilde B$. Moreover, $\bigcap_{N=1}^{\infty}\ker\left(\widetilde\rho_N\right)=0$.

\smallskip There are nontrivial relations between generators in \eqref{eq: two sets of gen}, which are originally proven in \cite[Appendix A]{DedushenkoGaiotto202009b}. We collect these results and reproduce the proof in the notation of this paper for the convenience of readers.

\begin{lemma}\label{lemma: identities}
We have the following identities:
\begin{align}
    T_{ab}(u)\widetilde{T}_{bc}(-u)&=\delta_{ac},\\
    T_{ab}(u)\widetilde{T}_{ba}(-u+K\hbar)&=KZ(u)\widetilde{Z}(-u+K\hbar)\label{eq:id 2}
\end{align}
\end{lemma}

\begin{proof}
First of all, we compute 
\begin{align*}
    &T_{ab}(u)\widetilde{T}_{bc}(w)=\delta_{ac}+I_a\frac{1}{u-B}J_c+I_a\frac{1}{w-\widetilde B}J_c+I_a\frac{1}{u-B}J_bI_b\frac{1}{w-\widetilde B}J_c\\
    &~=\delta_{ac}+I_a\frac{1}{u-B}J_c+I_a\frac{1}{w-\widetilde B}J_c+I_a\frac{1}{u-B}(B+\widetilde B)\frac{1}{w-\widetilde B}J_c\\
    &~=\delta_{ac}+(u+w)I_a\frac{1}{u-B}\frac{1}{w-\widetilde B}J_c.
\end{align*}
Taking $w=-u$, we get $T_{ab}(u)\widetilde{T}_{bc}(-u)=\delta_{ac}$.
Contracting with $\delta_{ac}$, we get
\begin{align*}
    &T_{ab}(u)\widetilde{T}_{ba}(w)=K+(u+w)\mathrm{Tr}\left(\frac{1}{w-\widetilde B}JI\frac{1}{u-B}\right)-K\hbar (u+w)\mathrm{Tr}\left(\frac{1}{w-\widetilde B}\frac{1}{u-B}\right)\\
    &~=K+(u+w)\mathrm{Tr}\left(\frac{1}{w-\widetilde B}(B+\widetilde B)\frac{1}{u-B}\right)+\hbar(u+w)\mathrm{Tr}\left(\frac{1}{w-\widetilde B}\right)\mathrm{Tr}\left(\frac{1}{u-B}\right)\\
    &\qquad -K\hbar (u+w)\mathrm{Tr}\left(\frac{1}{w-\widetilde B}\frac{1}{u-B}\right).
\end{align*}
Here the second equality follows from moment map condition. Taking $w=-u+K\hbar$, we get
\begin{align*}
    &T_{ab}(u)\widetilde{T}_{ba}(-u+K\hbar)=K\left(1-\mathrm{Tr}\left(\frac{\hbar}{u-B}\right)\right)\left(1-\mathrm{Tr}\left(\frac{\hbar}{-u+K\hbar-\widetilde B}\right)\right)=KZ(u)\widetilde{Z}(-u+K\hbar).
\end{align*}
\end{proof}

Recall that the quantum determinant of $T(u)$ is defined as \cite[\S 2.6]{molev2003yangians}
\begin{align}
    \mathrm{qdet}\:T(u)=\sum_{\sigma\in \mathfrak{S}_K}\mathrm{sgn}(\sigma) T_{\sigma(1),1}(u)\cdots T_{\sigma(K),K}(u-(K-1)\hbar).
\end{align}
It is proposed in \cite{DedushenkoGaiotto202009b} that quantum determinant of $T(u)$ should be related to Capelli's determinant of $B$ and $\widetilde{B}$. We prove their proposal in the next proposition.
\begin{proposition}
Let $C(u)$ be the Capelli's determinant of $B$
\begin{align}\label{eq:the definition of Capelli determinant}
    C(u)=\sum_{\sigma\in \mathfrak S_N}\mathrm{sgn}(\sigma)(u-(N-1)\hbar-B)_{\sigma(1),1}\cdots (u-B)_{\sigma(N),N},
\end{align}
and similarly let $\widetilde{C}(u)$ be the the Capelli's determinant of $\widetilde B$. Then
\begin{align}\label{eqn: qdet and Capelli's det}
    \mathrm{qdet}\:T(u)=\frac{\widetilde{C}(-u+(K-1)\hbar)}{(-1)^N\cdot C(u)}.
\end{align}
\end{proposition}

\begin{proof}
According to the quantum Newton's formula \cite[Theorem 4.1]{molev2003yangians}, we have
\begin{align}\label{eq:q Newton}
    Z(u)=\frac{C(u-\hbar)}{C(u)},\quad \widetilde Z(u)=\frac{\widetilde C(u-\hbar)}{\widetilde C(u)}
\end{align}
Define a rational function $$f(u):=\mathrm{qdet}\:T(u)\frac{C(u)}{\widetilde{C}(-u+(K-1)\hbar)} .$$ Compare the quantum Liouville formula \cite[Theorem 2.28]{molev2003yangians}:
\begin{align}
    \frac{1}{K}T_{ab}(u)\widetilde{T}_{ba}(-u+K\hbar)&=\frac{\mathrm{qdet}\:T(u-\hbar)}{\mathrm{qdet}\:T(u)},
\end{align}
with \eqref{eq:id 2} and \eqref{eq:q Newton}, we get $f(u)/f(u-\hbar)=1$. It follows that $f(u)$ does not depend on $u$. In particular, $f(u)=\lim_{u\to \infty}f(u)=(-1)^N$, i.e. $\mathrm{qdet}\:T(u)=(-1)^N\widetilde{C}(-u+(K-1)\hbar)/C(u)$. 
\end{proof}

\begin{remark}
Now we have RTT generator $T(u)$ and its inverse $\widetilde{T}(-u)$, then the $J$-generators of the Yangian for $\mathfrak{sl}_K$ can be obtained from them, in fact one define $B_{\mathrm{avr}}=\frac{1}{2}(B-\widetilde B)$, and
\begin{align}
    J^{(n)}_{ab}=I_aB_{\mathrm{avr}}^nJ_b,
\end{align}
then $J^{(0)}_{ab}$ are generators of $\mathfrak{gl}_K$ and they act on $J^{(1)}_{ab}$ as adjoint representation, and 
\begin{align}
    [J^{(1)}_{ab},J^{(1)}_{cd}]=\hbar(\delta_{bc}J^{(2)}_{ad}-\delta_{ad}J^{(2)}_{cb})+\frac{\hbar}{4}(J^{(0)}_{ed}J^{(0)}_{ae}J^{(0)}_{cb}-J^{(0)}_{eb}J^{(0)}_{ce}J^{(0)}_{ad}).
\end{align}
The above commutation relation shows that $\tilde J^{(0)}_{ab}=J^{(0)}_{ab}-\frac{1}{K}\delta_{ab}\sum_{c=1}^K J^{(0)}_{cc}$ and $\tilde J^{(1)}_{ab}=J^{(1)}_{ab}-\frac{1}{K}\delta_{ab}\sum_{c=1}^K J^{(1)}_{cc}$ generate the subalgebra $Y_{\hbar}(\mathfrak{sl}_K)\subset Y_\hbar(\fgl_K)\otimes \Lambda$.
\end{remark}

\subsection{Kernel of \texorpdfstring{$\rho_N$}{rhoN}}
In this subsection we characterize the kernel of the quotient map $\rho_N:Y_\hbar(\fgl_K)\otimes \Lambda\twoheadrightarrow \mathbb C_{\hbar}[\cM(N,K)]$. 
\begin{definition}\label{def:universal C(u)}
Fix $N$, define a power series $\mathfrak C(u)=u^N+\sum_{n>0}\mathfrak C_nu^{N-n}$ with coefficients $\mathfrak C_n\in Y_{\hbar}(\mathfrak {gl}_1)$ which is uniquely determined by 
\begin{align}\label{eqn: relation between power sum and Capelli}
    1-\frac{N\hbar}{u}-\hbar\sum_{n>0}\frac{p_n}{u^{n+1}}=\frac{\mathfrak C(u-\hbar)}{\mathfrak C(u)}.
\end{align}
Here $p_n$ are the power sum generators of $Y_{\hbar}(\mathfrak {gl}_1)$. Let RTT generator of $Y_{\hbar}(\mathfrak {gl}_K)$ be $$\mathfrak T(u)=1+\sum_{n\ge 0}\mathfrak T^{(n)}u^{-n-1},$$ then we define a power series $\widetilde{\mathfrak{C}}(u)=u^N+\sum_{n>0}\widetilde{\mathfrak{C}}_nu^{N-n}$ by
\begin{align}\label{eq:tilde C}
    \widetilde{\mathfrak{C}}(-u+(K-1)\hbar):=(-1)^N\mathfrak{C}(u)\mathrm{qdet}\: \mathfrak{T}(u).
\end{align}
Write the quantum minor of $\mathfrak T(u)$ for row indices $\underline{a}=(a_1<\cdots<a_i)$ and column indices $\underline{b}=(b_1<\cdots<b_i)$ as
\begin{align}
    \mathfrak T_{\underline{a},\underline{b}}(u)=\sum_{\sigma\in \mathfrak S_{i}}\mathrm{sgn}(\sigma)\mathfrak T_{\sigma(a_1),b_1}(u)\cdots \mathfrak T_{\sigma(a_i),b_i}(u-(i-1)\hbar).
\end{align}
\end{definition}

\begin{remark}
Let $C(u)$ be the Capelli's determinant of $B$, then by the quantum Newton's formula \cite{molev2003yangians}, we have
\begin{align}\label{eqn: quantum Newton}
    1-\mathrm{Tr}\left(\frac{\hbar}{u-B}\right)=\frac{C(u-\hbar)}{C(u)},
\end{align}
therefore the image of $\mathfrak C(u)$ in $\mathbb C_{\hbar}[\cM(N,K)]$ is $C(u)$. In the classical limit $\hbar\to 0$, $C(u)$ equals to $\det(u-B)$, and $C_n\equiv (-1)^n m_n\mod{\hbar}$, where $m_n$ are the generators of $\mathbb C[L^-\GL_1]$ that take the value of $a_n$ in the power series $1+\sum_{n\ge 1}a_nz^{-n}\in L^-\GL_1$.
\end{remark}

\begin{theorem}\label{theorem: quantum ideal}
The kernel of $\rho_N:Y_\hbar(\fgl_K)\otimes \Lambda\twoheadrightarrow \mathbb C_{\hbar}[\cM(N,K)]$ is generated by all coefficients for negative powers in $u$ in the power series
\begin{align}\label{eqn: ideal generators}
    \mathfrak C(u),\quad \widetilde{\mathfrak C}(u),\quad \mathfrak C(u)\mathfrak T_{\underline{a},\underline{b}}(u),
\end{align}
for all $\underline{a}=(a_1<\cdots<a_i),\underline{b}=(b_1<\cdots<b_i)$ and all $1\le i\le K-1$.
\end{theorem}

\begin{proof}
First of all, we show that \eqref{eqn: ideal generators} are mapped to polynomials. For $\mathfrak C(u)$, its image is the Capelli's determinant $C(u)$ of $B$, which is a polynomial. Note that $C(u)$ is known to be noncommutative version of characteristic polynomial in the sense that $C(B)=0$ \cite{molev2003yangians}, thus we have recursion relations: $T_{ab}^{(m)}+\sum_{n=1}^NC_nT_{ab}^{(m-n)}=0$ for all $m\ge N$, which is equivalent to that $C(u)T_{ab}(u)$ is a polynomial. It follows from \eqref{eqn: qdet and Capelli's det} that $C(u)\mathrm{qdet}\:T(u)$ is a polynomial. Next we consider the embedding $\mathbb C_{\hbar}[\cM(N,i)]\hookrightarrow \mathbb C_{\hbar}[\cM(N,K)]$ by $B\mapsto B$ and $J_{is}\mapsto J_{ia_s}$ and $I_{si}\mapsto I_{a_s i}$. This implies that $C(u)T_{\underline{a},\underline{a}}(u)$ are polynomials for all $\underline{a}=(a_1<\cdots<a_i)$ and all $1\le i\le K$. After taking commutators with $T^{(0)}_{rs}$ for various indices $r$ and $s$, we see that all coefficients for negative powers in $u$ in the power series $C(u)T_{\underline{a},\underline{b}}(u)$ are zero. In particular $\widetilde{\mathfrak{C}}(u)$ \eqref{eq:tilde C} is mapped to a polynomial. Thus we see that \eqref{eqn: ideal generators} are mapped to polynomials.

Next we show that the kernel of $Y_\hbar(\fgl_K)\otimes \Lambda\twoheadrightarrow \mathbb C_{\hbar}[\cM(N,K)]$ is generated by all coefficients for negative powers in $u$ in the power series \eqref{eqn: ideal generators}. By the flatness over $\mathbb C[\hbar]$ (Lemma \ref{lemma: flatness}), it suffices to show that they generate the ideal modulo $\hbar$. Let us define the closed subscheme $X$ in $L^-\GL_K\times L^-\GL_1$ by vanishing of all coefficients for negative powers in $u$ in the power series \eqref{eqn: ideal generators} modulo $\hbar$. Then $\cM(N,K)$ is a closed subscheme of $X$. In view of Proposition \ref{prop:phase space is a normal affine variety}, it is then enough to show that $X$ is reduced and irreducible of dimension $2NK$.

Consider the Drinfeld's generators $\{H_i(u),E_i(u),F_i(u):1\le i\le K-1\}$ of the subalgebra $Y_{\hbar}(\fsl_K)\subset Y_\hbar(\fgl_K)\otimes \Lambda$, they are related to generators $\{\mathfrak{C}(u),\mathfrak{T}(u)\}$ by the following \cite[Theorem 12.1.4]{ChariPressley199510}:
\begin{align}\label{eq:Drinfeld gen}
\begin{split}
    &H_i(u)=\mathfrak C(u)^{\delta_{i,1}}\frac{A_{i-1}(u+\frac{\hbar}{2})A_{i+1}(u+\frac{\hbar}{2})}{A_{i}(u)A_i(u+\hbar)},\\
    &E_i(u)=\mathfrak C(u+\frac{i-1}{2}\hbar)\mathfrak T_{\underline{i},\underline{i}^+}(u+\frac{i-1}{2}\hbar)A_i(u)^{-1},\\
    &F_i(u)=\mathfrak C(u+\frac{i-1}{2}\hbar)A_i(u)^{-1}\mathfrak T_{\underline{i}^+,\underline{i}}(u+\frac{i-1}{2}\hbar),
\end{split}
\end{align}
where $1\le i\le K-1$, and $\underline{i}=(1<\cdots<i)$, and $\underline{i}^+=(1<\cdots<i-1<i+1)$, and $A_0(u)=1$, and
\begin{align*}
    A_{i}(u)=u^N+\sum_{p=1}^{\infty}A_i^{(p)}u^{N-p}:=\mathfrak C(u+\frac{i-1}{2}\hbar)\mathfrak T_{\underline{i},\underline{i}}(u+\frac{i-1}{2}\hbar),
\end{align*}
By \cite[Theorem 2.20]{molev2003yangians}, we have $Y_\hbar(\fgl_K)\cong Y_{\hbar}(\fsl_K)\otimes Z(Y_{\hbar}(\mathfrak{gl}_K))$, where $Z(Y_{\hbar}(\mathfrak{gl}_K))$ is a commutative $\bC[\hbar]$ algebra freely generated by coefficients of power series $\mathrm{qdet}\:\mathfrak{T}(u)$. According to \eqref{eq:tilde C}, $Z(Y_{\hbar}(\mathfrak{gl}_K))\otimes \Lambda$ is a commutative $\bC[\hbar]$ algebra freely generated by $\{\mathfrak{C}_n,\widetilde{\mathfrak{C}}_n:n\in \bZ_{>0}\}$. Consider the algebra $$Y_\hbar(\fsl_K)[z_{1,1},\cdots,z_{1,N},z_{K-1,1},\cdots,z_{K-1,N}]:=Y_\hbar(\fsl_K)\otimes\bC[z_{1,1},\cdots,z_{1,N},z_{K-1,1},\cdots,z_{K-1,N}],$$ and we construct an algebra homomorphism 
\begin{align}\label{eq:alg iso}
    Y_\hbar(\fsl_K)[z_{1,1},\cdots,z_{1,N},z_{K-1,1},\cdots,z_{K-1,N}]\longrightarrow Y_\hbar(\fgl_K)\otimes \Lambda/(\mathfrak{C}_n,\widetilde{\mathfrak{C}}_n:n>N),
\end{align}
which is defined by the identity map on $Y_\hbar(\fsl_K)$ and 
\begin{align}
    \prod_{i=1}^N(u-z_{1,i})=\mathfrak C\left(u-\frac{1}{2}\hbar\right),\; \prod_{i=1}^N(u-z_{K-1,i})=(-1)^N\widetilde{\mathfrak C}\left(-u+\frac{K}{2}\hbar\right).
\end{align}
By the discussions above, \eqref{eq:alg iso} is an isomorphism. Therefore we have an isomorphism between quotient algebras
\begin{align}\label{eq:alg iso 2}
\frac{Y_\hbar(\fgl_K)\otimes \Lambda}{(\text{\small coeff. of negative powers in \eqref{eqn: ideal generators}})}
    \cong \frac{Y_\hbar(\fsl_K)[z_{1,1},\cdots,z_{1,N},z_{K-1,1},\cdots,z_{K-1,N}]}{(A^{(p)}_i:1\le i\le K-1,p>N)}.
\end{align}
According to \cite[\S 1.2.1]{Kamnitzer_2017} and \cite[Theorem 4.10]{kamnitzer2014yangians}, the RHS of \eqref{eq:alg iso 2} is isomorphic to the quantized Coulomb branch algebra $\mathcal A_\hbar$ associated to the the quiver in Figure \ref{fig:quiver for the Coulomb branch description} with gauge symmetry.
\begin{figure}[H]
    \centering
    \begin{tikzpicture}[scale=.5]
    
        \draw[line width=1pt] (-6,0) circle (.75cm);
        \draw[line width=1pt] (-5.25,0) -- (-4.25,0);
        \draw[line width=1pt] (-3.5,0) circle (.75cm);
        \draw[line width=1pt] (-2.75,0) -- (-1.75,0);
        \draw[line width=1pt,loosely dotted] (-1.75,0) -- (1.75,0);
        \draw[line width=1pt] (-6,-.75) -- (-6,-2.05);
        \draw[line width=1pt] (-6.65,-2.05) rectangle (-5.35,-3.35);
        
        \node at (-6, 0) {$N$};
        \node at (-6,-2.7) {$N$};
        \node at (-3.5,0) {$N$};

        \begin{scope}[xscale=-1]
        
        \draw[line width=1pt] (-6,0) circle (.75cm);
        \draw[line width=1pt] (-5.25,0) -- (-4.25,0);
        \draw[line width=1pt] (-3.5,0) circle (.75cm);
        \draw[line width=1pt] (-2.75,0) -- (-1.75,0);
        \draw[line width=1pt] (-6,-.75) -- (-6,-2.05);
        \draw[line width=1pt] (-6.65,-2.05) rectangle (-5.35,-3.35);

        \draw [line width=1pt, decorate,decoration={brace,amplitude=10pt,mirror,raise=.75cm}]
        (-6.8,0) -- (6.8,0) node[midway,yshift=1.4cm]{$K-1$ gauge nodes};

        \node at (-6, 0) {$N$};
        \node at (-6,-2.7) {$N$};
        \node at (-3.5,0) {$N$};
        \end{scope}

    \end{tikzpicture}
    \caption{The quiver diagram for the Coulomb branch description.}
    \label{fig:quiver for the Coulomb branch description}
\end{figure}
This implies that the scheme $X$ whose ring of functions is the LHS of \eqref{eq:alg iso 2} modulo $\hbar$ is isomorphic to the Coulomb branch $\cM_C$ associated to the quiver in Figure \ref{fig:quiver for the Coulomb branch description} with flavor symmetry. Then $X$ is reduced and irreducible of dimension $2NK$ by \cite[Theorem 3.20]{braverman2016coulomb}. This concludes the proof.
\end{proof}

\begin{corollary}\label{cor:M(N,K) as Coulomb branch}
$\cM(N,K)$ is isomorphic to the Coulomb branch associated to the quiver in Figure \ref{fig:quiver for the Coulomb branch description} with flavor symmetry. Therefore according to \cite[Theorem 3.20]{braverman2016coulomb}, $\cM(N,K)$ is isomorphic to the Beilinson-Drinfeld slice $\overline{\mathcal W}^{\underline{\lambda^*}}_0$ for $\GL_K$, where $\lambda=N\omega_1+N\omega_{K-1}$, $\omega_i$ is the $i$-th fundamental coweight for $\GL_K$. $\bC_\hbar[\cM(N,K)]$ is isomorphic to the quantized Coulomb branch algebra associated to the quiver in Figure \ref{fig:quiver for the Coulomb branch description} with flavor symmetry, which is the truncated Yangian $\mathbf Y^{\lambda}_0$.
\end{corollary}

\subsection{Quantized Coproduct}\label{sec:quantized coproduct}

In this subsection, we shall use the notation $\mathfrak{C}_N(u)$ to denote the power series define in Definition \ref{def:universal C(u)}.

\smallskip Let us define a coproduct 
\begin{align}\label{def: coproduct}
\begin{split}
\underset{n_1,n_2}{\mathbf \Delta}&:Y_\hbar(\fgl_K)\otimes \Lambda\to \left(Y_\hbar(\fgl_K)\otimes \Lambda\right)\otimes \left(Y_\hbar(\fgl_K)\otimes \Lambda\right),\\
\underset{n_1,n_2}{\mathbf \Delta}(\mathfrak T_{ab}(u))&=\sum_{c=1}^K\mathfrak T_{ac}(u)\otimes \mathfrak T_{cb}(u),\quad \underset{n_1,n_2}{\mathbf \Delta}(\mathfrak C_{n_1+n_2}(u))=\mathfrak C_{n_1}(u)\otimes \mathfrak C_{n_2}(u).
\end{split}
\end{align}
We note that the restriction of $\underset{n_1,n_2}{\mathbf \Delta}$ to the subalgebra $Y_\hbar(\fgl_K)$ agrees with the ordinary coproduct on $Y_\hbar(\fgl_K)$ \cite[\S 2.5]{molev2003yangians}. The coproduct \eqref{def: coproduct} is coassociative in the sense that 
\begin{align}
    \left(\underset{n_1,n_2}{\mathbf \Delta}\otimes\Id\right)\circ\underset{n_1+n_2,n_3}{\mathbf \Delta}= \left(\Id\otimes\underset{n_2,n_3}{\mathbf \Delta}\right)\circ\underset{n_1,n_2+n_3}{\mathbf \Delta}
\end{align}

\begin{proposition}\label{prop:existence of coproduct}
There exists a unique $\bC[\hbar]$ algebra homomorphism $$\underset{n_1,n_2}{\Delta}:\bC_\hbar[\cM(n_1+n_2,K)]\to \bC_\hbar[\cM(n_1,K)]\otimes \bC_\hbar[\cM(n_2,K)]$$ which makes the following diagram
\begin{equation}\label{eq:commutativity of coproduct and the map rho}
\begin{tikzcd}
Y_\hbar(\fgl_K)\otimes \Lambda \ar[r,"\underset{n_1,n_2}{\mathbf \Delta}"] \ar[d,"\rho_{n_1+n_2}" '] & \left(Y_\hbar(\fgl_K)\otimes \Lambda\right)\otimes \left(Y_\hbar(\fgl_K)\otimes \Lambda\right) \ar[d,"\rho_{n_1}\otimes\rho_{n_2}"] \\
\bC_\hbar[\cM(n_1+n_2,K)] \ar[r,"\underset{n_1,n_2}{\Delta}"] & \bC_\hbar[\cM(n_1,K)]\otimes \bC_\hbar[\cM(n_2,K)]
\end{tikzcd}
\end{equation}
commutative. Here the vertical maps are constructed in Proposition \eqref{prop:the surjective map between Yangian and deformed ring of functions}, and they are explicitly given by $\rho_N(\mathfrak{T}(u))=T(u)$ and $\rho_N(\mathfrak{C}_N(u))=C_N(u)$, where $C_N(u)$ is the Capelli determinant of $B$ \eqref{eq:the definition of Capelli determinant}.
\end{proposition}

\begin{proof}
We note that $\rho_{n_1+n_1}$ is surjective by Proposition \ref{prop:the surjective map between Yangian and deformed ring of functions}, therefore it is enough to show that $$\underset{n_1,n_2}{\mathbf \Delta}(\ker(\rho_{n_1+n_2}))\subseteq \ker(\rho_{n_1})\otimes\ker(\rho_{n_1}).$$
By Theorem \ref{theorem: quantum ideal}, we are left with showing that 
\begin{align*}
    (\rho_{n_1}\otimes\rho_{n_1})\circ\underset{n_1,n_2}{\mathbf \Delta}(\mathfrak{C}_{n_1+n_2}(u))\;\text{ and }\;(\rho_{n_1}\otimes\rho_{n_1})\circ\underset{n_1,n_2}{\mathbf \Delta}(\mathfrak{C}_{n_1+n_2}(u)\mathfrak{T}_{\underline{a},\underline{b}}(u))
\end{align*}
are polynomials in $u$, for all $\underline{a}=(a_1<\cdots<a_i),\underline{b}=(b_1<\cdots<b_i)$ and all $1\le i\le K$. By the definition of $\underset{n_1,n_2}{\mathbf \Delta}$, we have
\begin{align*}
    (\rho_{n_1}\otimes\rho_{n_1})\circ\underset{n_1,n_2}{\mathbf \Delta}(\mathfrak{C}_{n_1+n_2}(u))=\rho_{n_1}(\mathfrak{C}_{n_1}(u))\otimes \rho_{n_2}(\mathfrak{C}_{n_2}(u))=C_{n_1}(u)\otimes C_{n_2}(u),
\end{align*}
which is a polynomial in $u$. By \cite[Proposition 2.14]{molev2003yangians}, we have 
\begin{align*}
    \underset{n_1,n_2}{\mathbf \Delta}(\mathfrak{T}_{\underline{a},\underline{b}}(u))=\sum_{\underline{c}}\mathfrak{T}_{\underline{a},\underline{c}}(u)\otimes \mathfrak{T}_{\underline{c},\underline{b}}(u),
\end{align*}
summed over all subsets $\underline{c}=(c_1<\cdots<c_i)$ of $\{1,\cdots,K\}$. Thus we have
\begin{align*}
    (\rho_{n_1}\otimes\rho_{n_1})\circ\underset{n_1,n_2}{\mathbf \Delta}(\mathfrak{C}_{n_1+n_2}(u)\mathfrak{T}_{\underline{a},\underline{b}}(u))&=\sum_{\underline{c}}\rho_{n_1}(\mathfrak{C}_{n_1}(u)\mathfrak{T}_{\underline{a},\underline{c}}(u))\otimes \rho_{n_2}(\mathfrak{C}_{n_2}(u)\mathfrak{T}_{\underline{c},\underline{b}}(u))\\
    &=\sum_{\underline{c}} {C}_{n_1}(u){T}_{\underline{a},\underline{c}}(u)\otimes {C}_{n_2}(u){T}_{\underline{c},\underline{b}}(u),
\end{align*}
which is a polynomial in $u$ because each summand is a polynomial.
\end{proof}

The coproduct $\underset{n_1,n_2}{\Delta}$ is explicitly given by 
\begin{align}\label{eqn: quantum coproduct}
    \underset{n_1,n_2}{\Delta}(T_{ab}(u))=T_{ac}(u)\otimes T_{cb}(u),\quad \underset{n_1,n_2}{\Delta}(C_{n_1+n_2}(u))=C_{n_1}(u)\otimes C_{n_2}(u).
\end{align}
Equivalently, one can write down $\underset{n_1,n_2}{\Delta}$ using the second line of \eqref{eq: two sets of gen}:
\begin{align}
    \underset{n_1,n_2}{\Delta}(\widetilde{T}_{ab}(u))=\widetilde{T}_{cb}(u)\otimes \widetilde{T}_{ac}(u),\quad \underset{n_1,n_2}{\Delta}(\widetilde{C}_{n_1+n_2}(u))=\widetilde{C}_{n_1}(u)\otimes \widetilde{C}_{n_2}(u).
\end{align}

We can group the family of coproducts $\underset{n_1,n_2}{\mathbf \Delta}$ into a single one, at the cost of introducing a variable which plays the role of treating $n_1,n_2$ as undetermined variables.

\begin{definition}
We define the $\mathbb C[\hbar]$-Hopf algebra $Y_{\hbar}(\mathfrak{gl}_K)\otimes\Lambda_\delta$ as the algebra $Y_\hbar(\fgl_K)\otimes \Lambda[\delta]$, equipped with the coproduct
\begin{align}\label{eqn: large N quantum coproduct}
\begin{split}
\mathbf\Delta(\mathfrak T_{ab}(u))&=\mathfrak T_{ac}(u)\otimes \mathfrak T_{cb}(u),\;\mathbf\Delta(\delta)=\delta\otimes 1+1\otimes\delta,\\
\mathbf\Delta(p_n)&=p_n\otimes 1+1\otimes p_n-\hbar\sum_{i=0}^{n-1}p_{i}\otimes p_{n-1-i},
\end{split}
\end{align}
where we set $p_0:=\delta$. The counit $\epsilon:Y_{\hbar}(\mathfrak{gl}_K)\otimes\Lambda_\delta\to \bC[\hbar]$ is defined by 
\begin{align}
    \epsilon(\mathfrak T_{ab}(u))=\delta_{ab},\; \epsilon(p_n)=\epsilon(\delta)=0.
\end{align}
To write down the antipode map $S$, we define power series $\mathfrak A(u)=1+\sum_{n>0}\mathfrak A_nu^{-n}, \mathfrak A_n\in Y_{\hbar,\delta}(\mathfrak{gl}_1)$ by 
\begin{align}
    \left(1-\frac{\hbar}{u}\right)^{-\delta}\left(1-\hbar\sum_{n\ge 0}\frac{p_n}{u^{n+1}}\right)=\frac{\mathfrak A(u-\hbar)}{\mathfrak A(u)},
\end{align}
then the antipode $S$ is given by
\begin{align}
    S(\mathfrak T(u))=\mathfrak T(u)^{-1},\;S(\delta)=-\delta,\; S(\mathfrak A(u))=\mathfrak A(u)^{-1}.
\end{align}
\end{definition}

The quotient map $\rho_N$ in Proposition \ref{prop:the surjective map between Yangian and deformed ring of functions} naturally extends to $\rho'_N:Y_{\hbar}(\mathfrak{gl}_K)\otimes\Lambda_\delta\twoheadrightarrow \mathbb C_{\hbar}[\cM(N,K)]$ which is given by $$\rho'_N(\mathfrak T(u))= T(u),\; \rho'_N(\mathfrak{A}(u))=u^{-N}C_N(u),\; \rho'_N(\delta)= N.$$ Moreover the following diagram
\begin{equation}
\begin{tikzcd}
Y_{\hbar}(\mathfrak{gl}_K)\otimes\Lambda_\delta \ar[r,"\mathbf \Delta"] \ar[d,"\rho'_{n_1+n_2}" '] & \left(Y_{\hbar}(\mathfrak{gl}_K)\otimes\Lambda_\delta\right)\otimes \left(Y_{\hbar}(\mathfrak{gl}_K)\otimes\Lambda_\delta\right)\ar[d,"\rho'_{n_1}\otimes\rho'_{n_2}"] \\
\bC_\hbar[\cM(n_1+n_2,K)] \ar[r,"\underset{n_1,n_2}{\Delta}"] & \bC_\hbar[\cM(n_1,K)]\otimes \bC_\hbar[\cM(n_2,K)]
\end{tikzcd}
\end{equation}
is commutative.

\begin{remark}
The second line of \eqref{eqn: large N quantum coproduct} can be written in the following compact form
\begin{align}
    \mathbf\Delta(\mathfrak A(u))=\mathfrak A(u)\otimes \mathfrak A(u).
\end{align}
\end{remark}

\subsection{Quantized Phase Space and Coulomb Branch Algebra}
In this subsection we give a conceptual understanding of the identification between the quantized phase space $\mathbb C_{\hbar}[\cM(N,K)]$ and Coulomb branch algebra associated to the quiver in Figure \ref{fig:quiver for the Coulomb branch description}.

\smallskip Given a quiver $Q$, we denote by $\cA_\hbar(Q)$ the quantum Coulomb branch algebra associated to the quiver $Q$ with flavor symmetry, i.e. $$\cA_\hbar(Q):=\mathrm{H}^{(\GL(V)_{\mathcal O}\times\GL(W)_{\mathcal O})\rtimes \mathbb C^{\times}}_*(\mathcal R),$$ where $\mathcal R$ is the BFN's space of triples associated to the quiver $Q$, see \cite[Appendix A(ii)]{braverman2016coulomb} for details.

\begin{example}\label{ex:T[SU(N)]}
The $3d$ $\mscr{N}=4$ gauge theory associated to the following quiver
\begin{figure}[H]
    \centering
    \begin{tikzpicture}[scale=.5]
        
        \draw[line width=1pt] (-6,0) circle (.75cm);
        \draw[line width=1pt] (-5.25,0) -- (-4.25,0);
        \draw[line width=1pt] (-3.5,0) circle (.75cm);
        \draw[line width=1pt] (-2.75,0) -- (-1.75,0);
        \draw[line width=1pt,loosely dotted] (-1.75,0) -- (-.25,0);
        \draw[line width=1pt] (-.25,0) -- (.75,0);
        \draw[line width=1pt] (1.5,0) circle (.75cm);   
        \draw[line width=1pt] (2.25,0) -- (3.25,0);
        \draw[line width=1pt] (4,0) circle (.75cm);   
        
        \draw[line width=1pt] (4,-.7) -- (4,-2);
        
        \draw[line width=1pt] (3.35,-2) rectangle (4.65,-3.3);
        
        \node at (-6,0) {$1$};
        \node at (-3.5,0) {$2$};
        \node at (1.5,0) {$N$-$2$};
        \node at (4,0) {$N$-$1$};
        \node at (4,-2.65) {$N$};

    \end{tikzpicture}
    \caption{The quiver diagram of $\mathcal{T}[\text{SU}(N)]$ theory.}
    \label{fig:quiver for T[SU(N)] theory}
\end{figure}
\noindent is known as $\mathcal T[\mathrm{SU}(N)]$, its Coulomb branch algebra with flavor symmetry is isomorphic to $U_{\hbar}(\mathfrak{gl}_N)$. An explicit way to see this isomorphism is by looking at the evaluation representation of $Y_{\hbar}(\mathfrak{gl}_N):\mathfrak T_{ij}(u)\mapsto \delta_{ij}+\frac{E_{ij}}{u}$, where $E_{ij}$ are the generators of $U_{\hbar}(\mathfrak{gl}_N)$ satisfying relations $[E_{ij},E_{kl}]=\hbar(\delta_{jk}E_{il}-\delta_{il}E_{kj})$. Define $$A_n(u)=u^{[n]}\mathfrak T_{\underline{n},\underline{n}}(u+\frac{n-1}{2}\hbar),\; u^{[n]}:=(u+\frac{n-1}{2}\hbar)(u+\frac{n-3}{2}\hbar)\cdots (u-\frac{n-1}{2}\hbar),$$ where $\mathfrak T_{\underline{n},\underline{n}}(u)$ is the quantum determinant of the submatrix of $\mathfrak T(u)$ consisting of first $n$ rows and first $n$ columns. Write $A_n(u)=u^n+\sum_{i>0}A^{(i)}_n u^{n-i}$, then the kernel of $Y_{\hbar}(\mathfrak{gl}_N)\twoheadrightarrow U_{\hbar}(\mathfrak{gl}_N)$ contains $A^{(p)}_n$ for all $p>n$ and for all $1\le n\le N$. In the Drinfeld generators, we have
\begin{align*}
    H_n(u)=\frac{A_{n-1}(u+\frac{\hbar}{2})A_{n+1}(u+\frac{\hbar}{2})}{A_{n}(u)A_{n}(u+{\hbar})}, \:1\le n\le N-1,\: A_0(u)=1.
\end{align*}
Compare with \cite[Corollary B.17]{braverman2016coulomb} we conclude that the quotient of $Y_{\hbar}(\mathfrak {gl}_N)$ by the ideal generated by $A^{(p)}_n$ for all $p>n$ and for all $1\le n\le N$ is the truncated Yangian $\mathbf Y_0^{N\omega_{N-1}}$ for $\mathfrak{sl}_N$, where $\omega_{N-1}$ is the $(N-1)$-th fundamental coweight of $\mathfrak{sl}_N$, and the generating function of mass parameters $Z_{N-1}(u)=\prod_{i=1}^N(u-z_{N-1,i})$ is identified as $Z_{N-1}(u)=A_N(u+{\hbar})$. The classical limit of $\mathbf Y_0^{N\omega_{N-1}}$ is the function ring of $\overline{\mathcal W}^{\underline{N\omega_{N-1}}}_{0,\SL_N}$ (see \cite[\S 3(v)]{braverman2016coulomb}). It is known that  $\overline{\mathcal W}^{\underline{N\omega_{N-1}}}_{0,\SL_N}$ is reduced and irreducible of dimension $N^2$. On the other hand, the classical limit of $U_{\hbar}(\mathfrak{gl}_N)$ is the function ring of $\mathfrak{gl}_N^*$, which has dimension $N^2$ and embeds into $\overline{\mathcal W}^{\underline{N\omega_{N-1}}}_{0,\SL_N}$ as a closed subscheme, then the same argument as the proof of Theorem \ref{theorem: quantum ideal} shows that
\begin{align*}
    U_{\hbar}(\mathfrak{gl}_N)\cong \mathcal A_\hbar\text{ associated to the quiver in the Figure \ref{fig:quiver for T[SU(N)] theory} with flavor symmetry}.
\end{align*}
\end{example}

\bigskip Recall that balanced subquiver $Q^{\mathrm{bal}}\subseteq Q$ is formed by those edge-loop-free nodes $i\in Q_0$ such that $2\dim V_i=\dim W_i+\sum_{j}a_{ij}\dim V_j$ where $a_{ij}$ is the number of edges between $i$ and $j$. It is well-known that $Q^{\mathrm{bal}}$ is a union of finite ADE quivers, unless $Q^{\mathrm{bal}}$ is a union of connected components of $Q$ of affine type with zero framing on them \cite[Appendix A]{braverman2017ring}. It is shown in \cite[Appendix A]{braverman2017ring} that if it is not the latter case then the corresponding ADE group, denoted by $L^{\mathrm{bal}}$, acts on the Coulomb branch algebra $\cA_\hbar(Q)$, such that the infinitesimal action is generated by $\frac{1}{\hbar}[H^{(1)}_i,\bullet],\frac{1}{\hbar}[E^{(1)}_i,\bullet],\frac{1}{\hbar}[F^{(1)}_i,\bullet]$ for those $i\in Q^{\mathrm{bal}}_0$. 

\bigskip\begin{example}\label{ex:ADE}
In the case that $Q$ is of ADE type with gauge dimension vector $\mathbf v$ and flavour dimension vector $\mathbf{w}$, the classical Coulomb branch $\mathcal M_C(Q)$ is the Poisson variety $\overline{\mathcal{W}}^{\lambda^*}_{\mu^*}$, where $\lambda=\sum_{i\in Q_0}\mathbf{w}_i\lambda_i, \mu=\lambda-\sum_{i\in Q_0}\mathbf{v}_i\alpha_i$, $\lambda^*=-w_0(\lambda)$, $\lambda_i$ are fundamental coweights and $\alpha_i$ are fundamental coroots and $w_0$ is the longest element in the Weyl group of $G$. It is shown in \cite[Example A.5]{braverman2017ring} that $L^{\mathrm{bal}}$ action can be identified with the natural action of $\mathrm{Stab}_G(\mu^*)$ on $\overline{\mathcal{W}}^{\lambda^*}_{\mu^*}$ when $\mu$ is dominant. This holds for general $\mu$. In fact we can take a dominant $\nu$ such that $\langle \nu,\check{\alpha}_i\rangle=0, \forall i\in Q_0^{\mathrm{bal}}$ and $\mu+\nu$ is dominant, then the shift map $i_{0,\nu^*}:\mathbb C[\overline{\mathcal W}^{\lambda^*+\nu^*}_{\mu^*+\nu^*}]\to \mathbb C[\overline{\mathcal W}^{\lambda^*}_{\mu^*}]$ commutes with the action of $\mathrm{Stab}_G(\mu^*)\subseteq\mathrm{Stab}_G(\mu^*+\nu^*)$. Since $i_{0,\nu^*}$ is Poisson and preserves $E^{(1)}_i, F^{(1)}_i,H^{(1)}_i$ for $i\in Q^{\mathrm{bal}}_0$, it follows that the action of $L^{\mathrm{bal}}$ constructed in \cite[Proposition A.3]{braverman2017ring} commutes with the shift map. Since the action of $L^{\mathrm{bal}}$ agrees with the natural one for $\mathrm{Stab}_G(\mu^*)$ on $\overline{\mathcal W}^{\lambda^*+\nu^*}_{\mu^*+\nu^*}$, and the shift map is birational and equivariant for both of actions, these two actions agree on $\overline{\mathcal{W}}^{\lambda^*}_{\mu^*}$ as well.
\end{example}

\bigskip\begin{remark}
Suppose that there is another action of $L^{\mathrm{bal}}$ on $\cA_\hbar(Q)$ which acts trivially on $\hbar$ (not necessarily the one constructed in \cite[Appendix A]{braverman2017ring}), such that these two actions agree after modulo $\hbar$ and mass parameters (generators of $\mathrm H^*_{\GL(W)}(\mathrm{pt})$), then these two actions must agree on $\cA_\hbar(Q)$. In fact $\cA_\hbar(Q)$ is a flat deformation of $\cA_\hbar(Q)/(\hbar,\mathrm{mass})$ and the deformation spaces of modules for reductive group are trivial.
\end{remark}

\bigskip\begin{example}\label{ex:Weyl}
The Coulomb branch associated to the quiver of Figure \ref{fig:the quiver for the Weyl algebra}
\begin{figure}[H]
    \centering
    \hspace*{-2.cm}
    \begin{tikzpicture}[scale=.5]
        \draw[line width=1pt] (-6,0) circle (.75cm);
        \draw[line width=1pt] (-5.25,0) -- (-4.25,0);
        \draw[line width=1pt,loosely dotted] (-4.25,0) -- (-2.75,0);
        \draw[line width=1pt] (-2.75,0) -- (-1.75,0);
        \draw[line width=1pt] (-1,0) circle (.75cm);
        \draw[line width=1pt] (-.25,0) -- (+.75,0);
        \draw[line width=1pt] (1.5,0) circle (.75cm);
        
        \node at (-6, 0) {$1$};
        \node at (-1, 0) {$N$-$1$};

        \draw[line width=1pt] (2.25,0) -- (+3.25,0);
        \draw[line width=1pt,loosely dotted] (3.25,0) -- (+5.25,0);
        \draw[line width=1pt] (5.25,0) -- (+6.25,0);
        \draw[line width=1pt] (7,0) circle (.75cm);
        \draw[line width=1pt] (1.5,-.75) -- (1.5,-2.05);
        \draw[line width=1pt] (.85,-2.05) rectangle (2.15,-3.35);
        
        \node at (1.5, 0) {$N$};
        \node at (7,0) {$N$};
        \node at (1.5,-2.7) {$1$};
        
        \draw [line width=1pt, decorate,decoration={brace,amplitude=10pt,raise=.65cm}]
        (.75,0) -- (7.75,0) node[midway,yshift=1.3cm]{$K+1$ gauge nodes};

        \draw[line width=1pt] (7.75,0) -- (+8.75,0);
        \draw[line width=1pt] (9.5,0) circle (.75cm);
        \draw[line width=1pt] (10.25,0) -- (11.25,0);
        \draw[line width=1pt,loosely dotted] (11.25,0) -- (+12.75,0);
        \draw[line width=1pt] (12.75,0) -- (13.75,0);
        \draw[line width=1pt] (14.5,0) circle (.75cm);

        \node at (9.5,0) {$N$-$1$};
        \node at (14.5,0) {$1$};
        
    \end{tikzpicture}
    \caption{The quiver for the Weyl algebra $\mathrm{Weyl}^{\otimes N(K+N)}_{\hbar}$.}
    \label{fig:the quiver for the Weyl algebra}
\end{figure}
is isomorphic to generalized transverse slice $\overline{\mathcal W}^{\lambda_N^*}_{w_0(\lambda_N^*)}$ \cite[Theorem 3.10]{braverman2016coulomb}, where $\lambda_N$ is the $N$-th fundamental coweight of $\GL_{N(K+N)}$ and $w_0$ is the longest element of the Weyl group of $\GL_{N(K+N)}$ and $\lambda_N^*=-w_0(\lambda_N)$. The projection $\GL_{N(K+N)}(\!(z)\!)\to \mathrm{Gr}_{\GL_{N(K+N)}}$ identifies $\overline{\mathcal W}^{\lambda_N^*}_{w_0(\lambda_N^*)}$ with the cotangent bundle of the orbit $U_{\lambda_N}\cdot z^{-\lambda_N}$ as a Poisson variety, where $U_{\lambda_N}$ is the unipotent group whose Lie algebra is the eigenspace of $\lambda_N$ with eigenvalue $-1$, see \cite[Theorem 4.8]{Krylov_2021}. $U_{\lambda_N}$ can be naturally identified with the lower-triangular block in the following block-decomposition:
\begin{align*}
\begin{pNiceArray}{c:c}
  \SL_{N} & 0 \\
\hdottedline
 U_{\lambda_N} & \SL_{N+K}
\end{pNiceArray}.
\end{align*}
$L^{\mathrm{bal}}=\mathrm{Stab}_{\SL_{N(K+N)}}(w_0(\lambda_N^*))=\SL_N\times \SL_{N+K}$ acts on $U_{\lambda_N}$ naturally, and the action of $L^{\mathrm{bal}}$ on the Coulomb branch $\overline{\mathcal W}^{\lambda_N^*}_{w_0(\lambda_N^*)}\cong T^*U_{\lambda_N}$ is induced from its action on $U_{\lambda_N}$ by extending to the cotangent bundle. The half of cohomological grading $\deg_h$ is such that degree of linear coordinates on $U_{\lambda_N}$ is zero, and degree of linear coordinates on cotangent fiber is $1$. The quantized Coulomb branch algebra $\mathcal A_\hbar$ is an $L^{\mathrm{bal}}$-equivariant and $\deg_h$-graded quantization of the Poisson variety $T^*U_{\lambda_N}$. There exists a unique such quantization, which is the ring of $\hbar$-differential operators $\mathrm{Diff}_{\hbar}(U_{\lambda_N})$. $\mathrm{Diff}_{\hbar}(U_{\lambda_N})$ is isomorphic to the Weyl algebra $\mathrm{Weyl}^{\otimes N(K+N)}_{\hbar}$, therefore we conclude that
\begin{align*}
    \mathrm{Weyl}^{\otimes N(K+N)}_{\hbar}\cong \mathcal A_\hbar\text{ associated to the quiver in the Figure \ref{fig:the quiver for the Weyl algebra}}.
\end{align*}
These quantized Coulomb branches also appears in the context of quantized phase spaces of minuscule 't Hooft operators in $4d$ Chern-Simons theories, see \cite{CostelloGaiottoYagi202103}.
\end{example}

\bigskip Consider a quiver $Q$ containing following part
\begin{figure}[H]
    \centering
    \begin{tikzpicture}[scale=.5]
        
        \draw[line width=1pt] (-6,0) circle (.75cm);
        \draw[line width=1pt] (-5.25,0) -- (-4.25,0);
        \draw[line width=1pt] (-3.5,0) circle (.75cm);
        \draw[line width=1pt] (-2.75,0) -- (-1.75,0);
        \draw[line width=1pt,loosely dotted] (-1.75,0) -- (-.25,0);
        \draw[line width=1pt] (-.25,0) -- (.75,0);
        \draw[line width=1pt] (1.5,0) circle (.75cm);   
        \draw[line width=1pt] (2.25,0) -- (3.25,0);
        \draw[line width=1pt] (4,0) circle (.75cm);   
        
        \draw[line width=1pt] (4.4,.65) -- (5.75,2);
        \draw[line width=1pt] (4.75,0) -- (6.25,0);
        \draw[line width=1pt] (4.4,-.65) -- (5.75,-2);

        \node at (-6,0) {$1$};
        \node at (-3.5,0) {$2$};
        \node at (1.5,0) {$N$-$1$};
        \node at (4,0) {$N$};

    \end{tikzpicture}
    \caption{The quiver $Q$.}
    \label{fig:quiver Q}
\end{figure}
\noindent Then $\cA_{\hbar}(Q)$ admits an action of $\SL_N$, and also a grading ($\mathbb C^{\times}$ action) coming from $\pi_0(\mathrm{Gr}_{\GL_N})$ which commutes with the $\SL_N$ action, thus $\cA_{\hbar}(Q)$ admits an action of $\GL_N$. Denote the following quiver by $Q'$
\begin{figure}[H]
    \centering
    \begin{tikzpicture}[scale=.5]
        
        \draw[line width=1pt] (-.65,.65) rectangle (.65,-.65);
        
        \draw[line width=1pt] (.2,.65) -- (1.65,2);
        \draw[line width=1pt] (.65,0) -- (2.15,0);
        \draw[line width=1pt] (.2,-.65) -- (1.65,-2);

        \node at (0,0) {$N$};

    \end{tikzpicture}
    \caption{The quiver $Q'$.}
    \label{fig:quiver Q'}
\end{figure}
\noindent then we have
\begin{lemma}\label{lemma: cutting T[GLn]}
$\cA_{\hbar}(Q')\cong\cA_{\hbar}(Q)^{\GL_N}$.
\end{lemma}
\begin{proof}
Consider the affine Grassmannian $\mathrm{Gr}_{\GL_N}$ and denote by $\mathcal A_Q$ (resp. $\mathcal A_{Q'}$) the ring object in $D_{\GL_n(\mathscr O)\rtimes\mathbb C^{\times}}(\mathrm{Gr}_{\GL_N})$ coming from pushing forward of the dualing complex on the BFN space of triples corresponding to quiver gauge theory $Q$ (resp. $Q'$), see \cite{braverman2017ring}. Then we have $\mathcal A_{Q}\cong \mathcal A_{R}\overset{!}{\otimes}\mathcal A_{Q'}$ \cite{braverman2017ring}, where $\mathcal A_{R}$ is the regular ring object with a natural $\GL_N$ action (which is called the right action in \cite{braverman2017ring}). Therefore we have
\begin{align*}
    \cA_{\hbar}(Q)^{\GL_N}&=\mathrm H^*_{\GL_n(\mathscr O)\rtimes\mathbb C^{\times}}(\mathrm{Gr}_{\GL_N},\mathcal A_{R}\overset{!}{\otimes}\mathcal A_{Q'})^{\GL_N}=\mathrm H^*_{\GL_n(\mathscr O)\rtimes\mathbb C^{\times}}(\mathrm{Gr}_{\GL_N},\mathrm{IC}_0\overset{!}{\otimes}\mathcal A_{Q'})\\
    &=\mathrm{Ext}^*_{\GL_n(\mathscr O)\rtimes\mathbb C^{\times}}(\mathrm{IC}_0,\mathcal A_{Q'})=\cA_{\hbar}(Q').
\end{align*}
\end{proof}

\begin{remark}\label{rmk:another proof}
Apply Lemma \ref{lemma: cutting T[GLn]} to the quiver in the Figure \ref{fig:the quiver for the Weyl algebra} with $K=0$, and we see that $U_{\hbar}(\mathfrak{gl}_N)\cong \left(\mathrm{Weyl}^{\otimes N^2}_\hbar\right)^{\GL_N}$. It follows that $\mathbb C_{\hbar}[\cM(N,K)]\cong \left(\mathrm{Weyl}^{\otimes N(K+N)}_{\hbar}\right)^{\GL_N\times \GL_N}$ where the action comes from the restriction of $\GL_N\times \GL_{N+K}$ to $\GL_N\times \GL_N$. Apply Lemma \ref{lemma: cutting T[GLn]} to the quiver in the Figure \ref{fig:the quiver for the Weyl algebra} with nonzero $K$, followed by removing the edge between flavor vector spaces as it has no effect on Coulomb branch, we see that 
\begin{align*}
    \mathbb C_{\hbar}[\cM(N,K)]\cong \mathcal A_\hbar\text{ associated to the quiver in the Figure \ref{fig:quiver for the Coulomb branch description} with flavor symmetry}.
\end{align*}
\end{remark}

\bigskip\begin{remark}
Apply Lemma \ref{lemma: cutting T[GLn]} to the quiver in the Figure \ref{fig:quiver for T[SU(N)] theory} with $N$ replaced by $N+K$, and we see that $$U_{\hbar}(\mathfrak{gl}_{N+K})^{\GL_N}\cong \mathbf Y^{\lambda}_0,$$ where the right-hand-side is a truncated Yangian for $\mathfrak{sl}_{K}$ and $\lambda=N\omega_1+(N+K)\omega_{K-1}$, where $\omega_i$ is the $i$-th fundamental coweight of $\fsl_K$, see \cite[Appendix B(viii)]{braverman2016coulomb}. This is known as the centralizer construction of Yangian in the literature \cite[\S 2.13]{molev2003yangians}.
\end{remark}

~\\[.5cm]
\paragraph{Acknowledgement.} The authors would like to thank Kevin Costello, Davide Gaiotto, Ji Hoon Lee, and Nafiz Ishtiaque for discussions and comments on earlier version of this paper. The work of SFM
is funded by the Natural Sciences and Engineering Research Council of Canada (NSERC). The work of YZ is supported by Perimeter Institute for Theoretical Physics. Research at Perimeter Institute is supported by the Government of Canada through Industry Canada and by the Province of Ontario through the Ministry of Research and Innovation.

\appendix

\section{Relationship Between Two Poisson Structures}
\label{sec:Relationship between two Poisson structures}
In this appendix,\footnote{We would like to thank the referee of the journal {\it Advances in Theoretical and Mathematical Physics} who motivated the discussion of this appendix.} we explain the equivalence between the
natural Poisson structure on $\mcal{M}(N,K)$, given in \eqref{Poisson before quotient}, and the one inspired by the twisted holography computations in \cite{IshtiaqueMoosavianZhou201809}. We first derive the latter, and then show that (a bosonic version of) it coincides with \eqref{Poisson before quotient}.

\subsection{The Poisson Structure Inspired by Twisted Holography}\label{sec:a Poisson structure inspired by twisted holography}

In start by the construction of a Poisson structure on $\mcal{M}(N,K)$, which is motivated by the construction in \cite{IshtiaqueMoosavianZhou201809}, and consider 2d BF theory coupled to a 1d {\it fermionic} quantum-mechanical system. To cope with the notations of loc. cit. and to better compare the results, we switch to a new notation, similar to \cite{IshtiaqueMoosavianZhou201809}, as follows 
\begin{equation}\label{eq:redefinition of notations}
    I\to \ovl{\psi}, \qquad J\to\psi, \qquad B\to\B, \qquad A\to\A
\end{equation}
Note that $(I,J)$ are bosonic variables while $(\bpsi,\psi)$ have taken to be fermionic in \cite{IshtiaqueMoosavianZhou201809}. We will explain the difference in \S\ref{sec:fermionic vs bosonic quantum mechanics}. Furthermore, we switch the indices 
\begin{equation}
    a\leftrightarrow i, \qquad i=1,\ldots,K, \qquad a=1,\ldots,N.
\end{equation}
This temporary change of notation only applies to this part and not the rest of the paper.

\smallskip 2d BF theory coupled to a 1d quantum mechanics is described by the following action
\begin{equation}\label{eq:action of BF theory coupled to 1d fermionic quantum mechanics}
    S=\frac{1}{2\pi}\bigintsss_{\R^2_{x,w}}\tr(\B\msf{F}_\A)+\bigintsss_{\R_x\times\{w=0\}}\bpsi(\msf{d}+\A)\psi,
\end{equation}
where in the second term, all the indices are suppressed. The equations of motion of the theory are given by
\begin{equation}\label{eq:equations of motion of BF theory coupled to 1d quantum mechanics}
    D\B=\bpsi\psi\delta_{\R_x}, \qquad \msf{F}_\A=0, \qquad D\psi=D\bpsi=0,
\end{equation}
where $D\B:=\msf{d}\B+[\B,\A]$, $D\psi:=\msf{d}\psi+\A\psi$, and a similar expression for $\bpsi$, and $\delta_{\R_x}$ denotes a delta function supported on $\R_x$. The setup is depicted in Fig. \ref{fig:setup of BF theory coupled to a line defect}.

\smallskip Despite the above change of notation, we still denote the phase space of this system as $\mcal{M}(N,K)$ and the ring of functions on $\mcal{M}(N,K)$ as $\C[\mcal{M}(N,K)]$, although a graded one in the following sense. $\C[\mcal{M}(N,K)]$ is equipped with a $\mbb{Z}_2$-grading due to anti-commutativity relations of fermions
\begin{equation}\label{eq:anticommutativity of fermions}
    \bpsii{i}_a\,\bpsii{j}_b=-\bpsii{j}_b\bpsii{i}_a, \qquad \bpsii{i}_a\psi^b_j=-\psi^b_j\bpsii{i}_a, \qquad \psi^a_i\psi^b_j=-\psi^b_j\psi^a_i.
\end{equation}
Hence, 
\begin{equation}\label{eq:grading of the ring of functions on phase space}
    \C[\mcal{M}(N,K)]=\C[\mcal{M}(N,K)]_0\oplus\C[\mcal{M}(N,K)]_1,
\end{equation}
where $\C[\mcal{M}(N,K)]_0$ and $\C[\mcal{M}(N,K)]_1$ is the even and odd subspaces, respectively. The basic fields $\A,\B,\psi,$ and $\bpsi$ have gradings $0,0,+1,$ and $+1$, respectively. The algebra $\mcal{A}_{\text{ob}}$ of observables of the theory is given by
\begin{equation}
    \mcal{A}_{\text{ob}}:=\C[\mcal{M}(N,K)]^{\text{GL}_N},
\end{equation}
where $\C[\mcal{M}(N,K)]^{\text{GL}_N}$ denotes the space of all  $\text{GL}_N$-gauge-invariant elements in $\C[\mcal{M}(N,K)]$. We denote the quantum-mechanical algebra of observables of 2d BF theory couled to 1d quantum mechanics as $\wh{\mcal{A}}_{\text{ob}}$

\paragraph{A Graded Poisson Structure on $\mcal{A}_{\text{ob}}$ from the Classical Limit.} It is shown in \cite{IshtiaqueMoosavianZhou201809} that in the BF theory coupled to a 1d quantum mechanics, i.e. the system described by the Lagrangian \eqref{eq:action}, one can define two sets of gauge-invariant operators: (1) the operators\footnote{To avoid cluttering, we do not use $\wh{\mcal{O}}$ for the quantum-mechanical operators. It should be clear from the context whether we consider classical functions or quantum-mechanical operators.}
\begin{equation}\label{eq:operators O^i_j[n] in BF theory}
    \mcal{O}^i_j[n]:=\frac{1}{\hbar}\bpsi^{\,i}\B^n\psi_j, \qquad i,j=1,\ldots,K,\qquad n=0,1,\ldots,
\end{equation}
which from their very definitions are restricted to the line on which the 1d quantum mechanics lives; and (2) the set of operators
\begin{equation}\label{eq:operators B^n in BF theory}
    \mcal{O}[n]:=\frac{1}{\hbar}\tr_{\C^N}\B^n,\qquad n=0,1,\ldots,
\end{equation}
which are the only local gauge-invariant operators in the BF theory.\footnote{As we are working in two-dimensional flat space, it is always possible to choose the gauge $\A=0$. If we instead decide to work with the pure gauge $\A=g^{-1}\msf{d} g$, for some $\text{GL}_N$-valued function $g$ on $\R^2_{x,w}$, we still cannot get more gauge-invariant operators. The reason is as follows. The other local gauge-invariant operators are traces of polynomials of $\msf{F}_\A=\msf{d}\A+\A\wedge\A$, the field strength of the gauge field $\A$. However, it is evident from \eqref{eq:action of BF theory coupled to 1d fermionic quantum mechanics} that the equations of motion of $\B$ sets $\msf{F}_\A=0$. Therefore, all such operators vanish, which in turn renders \eqref{eq:operators B^n in BF theory} as the only gauge-invariant operators in BF theory (in the absence of the line defect).} These two sets of operators form the basis of the space of gauge-invariant operators, i.e.
\begin{equation}\label{eq:algebra of observables of 2d BF theory coupled to 1d quantum mechanics}
    \wh{\mcal{A}}_{\text{ob}}=\Big\langle (\mcal{O}^i_j[n],\mcal{O}[n]);\forall n,i,j\Big\rangle_{\C},
\end{equation}
where the angle bracket is used to denote the generating set of the algebra. Note that $\mcal{O}[n]$ commute with all the other operators as they live anywhere on the plane of BF theory, so one can pass them around. However, the operators $\mcal{O}^i_j[n]$, being restricted to the line can brought together in two different ways, from the left and also from the right, and one can then define their commutator. These have been illustrated in Fig. \ref{fig:setup of BF theory coupled to a line defect}.
\begin{figure}[H]
    \centering
    \TheSetup
    \caption{The setup of BF theory coupled to a line defect on which a fermionic quantum-mechanical system lives. While the operators $\mcal{O}[n]$ can be inserted anywhere in the bulk, the operators $\mcal{O}^i_j[n]$ live only on the line defect.}
    \label{fig:setup of BF theory coupled to a line defect}
\end{figure}
The commutator is then defined as 
\begin{equation}\label{eq:definition of commutator with the star product}
    [\mcal{O}^i_j[m],\mcal{O}^k_l[n]]_{\star}(x):=\lim_{\epsilon\to 0}\left(\mcal{O}^i_j[m](x+\epsilon)\star\mcal{O}^k_l[m](x-\epsilon)-\mcal{O}^k_l[m](x+\epsilon)\star\mcal{O}^i_j[m](x-\epsilon)\right),
\end{equation}
where as in \eqref{eq:action of BF theory coupled to 1d fermionic quantum mechanics}, $x$ denotes the coordinate on $\R_x$ along which the line defect is extended. As we are dealing with a coupled system, one needs to be careful about the quantum correction coming from the interaction of fields on the line with those of the bulk. It is then shown that \cite[eq. (92)]{IshtiaqueMoosavianZhou201809}
\begin{equation}\label{eq:quantum-corrected bracket [O^i_j[1],O^k_l[1]]_star}
    [\mcal{O}^i_j[1],\mcal{O}^k_l[1]]_\star=\delta^i_l\mcal{O}^k_j[2]-\delta^k_j\mcal{O}^i_l[2]+\frac{\hbar^2}{4}\left(\mcal{O}^m_j[0]\mcal{O}^k_m[0]\mcal{O}^i_l[0]-\mcal{O}^m_l[0]\mcal{O}^i_m[0]\mcal{O}^k_j[0]\right).
\end{equation}
Then, by a certain redefinition of operators \cite[eq. (98)]{IshtiaqueMoosavianZhou201809} and taking the $N\to\infty$ limit, which kills the trace relations among generators, \eqref{eq:quantum-corrected bracket [O^i_j[1],O^k_l[1]]_star} becomes one of the standard forms of $Y(\mfk{gl}_K)$, the Yangian associated with the Lie algebra $\mfk{gl}_K$, as presented, for example, in \cite[eq. (4), Theorem 12.1.1]{ChariPressley199510}. This relation together with the quantum-mechanically undeformed relations\footnote{Physically, it is clear why there is no corrections to these relations. The correction happens due to interactions in the plane of BF theory. The interaction vertex, which, by taking into account \eqref{eq:redefinition of notations}, can be read off from \eqref{eq:action of BF theory coupled to 1d fermionic quantum mechanics}, is $\B\wedge\A\wedge\A$. In the case of $[\mcal{O}^i_j[0],\mcal{O}^k_l[0]]_\star$, there is no $\B$ to be contracted in the bulk, while in the case of $[\mcal{O}^i_j[0],\mcal{O}^k_l[1]]_\star$, the $\B$ field in $\mcal{O}^k_l[1]$ can be contracted with an $\A$ in the bulk while the remaining the part of the vertex, $\B\wedge A$, form a loop. Any such diagrams vanish, as has been deduced in \cite[Footnote 26]{IshtiaqueMoosavianZhou201809}, hence no quantum correction in this case either.}
\begin{equation}\label{eq:undeformed brackets}
    \begin{aligned}
        [\mcal{O}^i_j[0],\mcal{O}^k_l[0]]_\star&=\delta^i_l\mcal{O}^i_j[0]-\delta^k_l\mcal{O}^i_l[0],
        \\
        [\mcal{O}^i_j[0],\mcal{O}^k_l[1]]_\star&=\delta^i_l\mcal{O}^i_j[1]-\delta^k_l\mcal{O}^i_l[1],
    \end{aligned}
\end{equation}
can then be used to reconstruct other Yangian relations \cite[Theorem 12.1.1]{ChariPressley199510}.

\smallskip By the basic philosophy of elementary quantum mechanics, or its of sophisticated version, deformation quantization  \cite{Berezin197410,Berezin197504,Berezin197506,BayenFlatoFronsdalLichnerowiczSternheimer197803I,BayenFlatoFronsdalLichnerowiczSternheimer197803II}, quantum mechanics amounts to deformation of the commutating algebra of observables of the classical theory $\mcal{A}_{\text{ob}}$ to the non-commutative algebra of quantum-mechanical observables, i.e. the following should hold
\begin{equation}
    \lim_{\hbar\to 0}\wh{\mcal{A}}_{\text{ob}}=\mcal{A}_{\text{ob}},
\end{equation}
to have a good notion of deformation quantization of the classical observables (i.e. gauge-invariant functions on phase space) of our theory. Let us emphasize that although $\mcal{A}_{\text{op}}$ is a commutative algebra, it is equipped with a {\it non-trivial} Poisson structure, which we denote as $\pois{\cdot}{\cdot}$. Its relation to the commutator in $\wh{\mcal{A}}_{\text{ob}}$ is 
\begin{equation}\label{eq:from commutators to Poisson bracket}
    [\cdot,\cdot]_\star=\hbar\pois{\cdot}{\cdot}+\mcal{O}(\hbar^2).
\end{equation}
There are many flavors of deformation quantization, such as Fedosov's \cite{Fedosov199402,Fedosov1995}, and Kontsevich's \cite{Kontsevich199709} deformation quantizations,\footnote{Fedosov's quantization applies to symplectic manifolds and can be generalized to regular Poisson manifolds (see \cite{DolgushevIsaevLyakhovichSharapov200206} for a discussion concerning non-regular Poisson manifolds). On the other hand, Kontsevich's deformation quantization works for general Poisson manifolds, even those having non-regular symplectic foliation.} which involve the notion of star product \cite{Moyal194901}, hence the star in the notation \eqref{eq:definition of commutator with the star product}. In our case, \eqref{eq:quantum-corrected bracket [O^i_j[1],O^k_l[1]]_star}  and \eqref{eq:undeformed brackets} are the defining relations of the quantum-mechanical algebra of observables. We thus would like to recover the commutative algebra of observable by taking $\hbar\to 0$ limit. However, it is clear that a naive $\hbar\to 0$ limit of \eqref{eq:quantum-corrected bracket [O^i_j[1],O^k_l[1]]_star}  and \eqref{eq:undeformed brackets} does not land us on a Poisson structure according to the general rule \eqref{eq:from commutators to Poisson bracket}. A proper way to take the limit is to redefine the operators \eqref{eq:operators O^i_j[n] in BF theory} as
\begin{equation}
    \wt{\mcal{O}}^i_j[n]:=\hbar\mcal{O}^i_j[n].
\end{equation}
In terms of these generators, \eqref{eq:undeformed brackets} becomes
\begin{equation}\label{eq:undeformed brackets in terms of tildeO generators}
    \begin{aligned}
        [\wt{\mcal{O}}^i_j[0],\wt{\mcal{O}}^k_l[0]]_\star&=\hbar\left(\delta^i_l\wt{\mcal{O}}^k_j[0]-\delta^k_j\wt{\mcal{O}}^i_l[0]\right),
        \\
        [\wt{\mcal{O}}^i_j[0],\wt{\mcal{O}}^k_l[1]]_\star&=\hbar\left(\delta^i_l\wt{\mcal{O}}^k_j[1]-\delta^k_j\wt{\mcal{O}}^i_l[1]\right),
    \end{aligned}
\end{equation}
while \eqref{eq:quantum-corrected bracket [O^i_j[1],O^k_l[1]]_star} becomes
\begin{equation}\label{eq:quantum-corrected bracket [tildeO^i_j[1],tildeO^k_l[1]]_star}
    [\wt{\mcal{O}}^i_j[1],\wt{\mcal{O}}^k_l[1]]_\star=\hbar\left\{\delta^i_l\wt{\mcal{O}}^k_j[2]-\delta^k_j\wt{\mcal{O}}^i_l[2]+\frac{1}{4}\left(\wt{\mcal{O}}^m_j[0]\wt{\mcal{O}}^k_m[0]\wt{\mcal{O}}^i_l[0]-\wt{\mcal{O}}^m_l[0]\wt{\mcal{O}}^i_m[0]\wt{\mcal{O}}^k_j[0]\right)\right\}.
\end{equation}
Comparing \eqref{eq:undeformed brackets in terms of tildeO generators} and \eqref{eq:quantum-corrected bracket [tildeO^i_j[1],tildeO^k_l[1]]_star}, it is now evident that by taking $\hbar\to 0$ limit, we end up with a Poisson structure $\pois{\cdot}{\cdot}$ on the classical algebra of observables
\begin{equation}
    \lim_{\hbar\to 0}[\wh{\mcal{F}},\wh{\mcal{G}}]_\star:=\hbar\pois{\mcal{F}}{\mcal{G}}, \qquad \forall\mcal{F},\mcal{G}\in\mcal{A}_{\text{ob}},
\end{equation}
where $\wh{\mcal{F}},\wh{\mcal{G}}\in\wh{\mcal{A}}$ are the quantum-mechanical operators corresponding to $\mcal{F}$ and $\mcal{G}$, respectively.\footnote{There exists the subtle question of possible normal ordering of quantum-mechanical operators, which we skip in this treatment, as it would not affect any conclusions.} A remarkable feature of $\wh{\mcal{A}}_{\text{op}}$ is that their is no higher-order $\mcal{O}(\hbar^2)$ corrections to the algebra. Hence, the bracket $[\cdot,\cdot]_\star$ provides a deformation of the Poisson algebra of observables of BF theory coupled to a fermionic quantum mechanics living on a line equipped with the Poisson structure $\{\cdot,\cdot\}_{\star}$, in the precise sense of deformation quantization \eqref{eq:from commutators to Poisson bracket}. Thinking of $\wt{\mcal{O}}^i_j[n]$ as classical observables, in particular, we have
\begin{equation}
    \begin{aligned}\label{eq:induced Poisson structure from classical limit}
    \pois{\wt{\mcal{O}}^i_j[0]}{\mcal{O}^k_l[0]}&=\delta^i_l\wt{\mcal{O}}^k_j[0]-\delta^k_j\wt{\mcal{O}}^i_l[0],
    \\
    \pois{\wt{\mcal{O}}^i_j[0]}{\mcal{O}^k_l[1]}&=\delta^i_l\wt{\mcal{O}}^k_j[1]-\delta^k_j\wt{\mcal{O}}^i_l[1],
    \\
    \pois{\wt{\mcal{O}}^i_j[1]}{\mcal{O}^k_l[1]}&=\delta^i_l\wt{\mcal{O}}^k_j[2]-\delta^k_j\wt{\mcal{O}}^i_l[2]+\frac{1}{4}\left(\wt{\mcal{O}}^m_j[0]\wt{\mcal{O}}^k_m[0]\wt{\mcal{O}}^i_l[0]-\wt{\mcal{O}}^m_l[0]\wt{\mcal{O}}^i_m[0]\wt{\mcal{O}}^k_j[0]\right)
    \end{aligned}
\end{equation}

\smallskip Finally, the grading \eqref{eq:grading of the ring of functions on phase space} will induce a grading on the Poisson structure $\pois{\cdot}{\cdot}$ such that 
\begin{equation}
    \pois{\mcal{F}_i}{\mcal{F}_j}\in\C[\mcal{M}(N,K)]_{i+j}, \qquad \forall\mcal{F}_i\in\C[\mcal{M}(N,K)]_i, 
\end{equation}
where the $\mbb{Z}_2$ relations $0+0=0,1+0=1,$ and $1+1=0$ are understood. This shows that while $\C[\mcal{M}(N,K)]_0$ is a subring of $\C[\mcal{M}(N,K)]$, $\C[\mcal{M}(N,K)]_1$ is merely a subspace thereof. Furthermore, 
\begin{equation}
    \pois{\mcal{F}_i}{\mcal{F}_j}=(-1)^{ij+1}\pois{\mcal{F}_j}{\mcal{F}_i},  \qquad  \forall\mcal{F}_i\in\C[\mcal{M}(N,K)]_i,
\end{equation}
which implies that we have indeed a graded Poisson structure. 

\paragraph{Poisson Brackets of Elementary Fields.} As we have explained above, the Poisson structure $\pois{\cdot}{\cdot}$ equips the commutative $\mcal{A}_{\text{ob}}$ with a graded Poisson structure. Furthermore, $\pois{\cdot}{\cdot}$ induces a Poisson structure on $\C[\mcal{M}(N,K)]$. Therefore, we would like to know the Poisson brackets of elementary fields $\B,\psi,$ and $\bpsi$ with respect to $\pois{\cdot}{\cdot}$.  
\begin{remark}[Gauge Fields] As before, we choose the gauge $\A=0$. Therefore, $\A$ (or more precisely, its components $\A^a_\mu$ where $\mu=1,2$ denotes the spacetime indices) does not appear as an element of $\C[\mcal{M}(N,K)]$. Therefore, there is no Poisson bracket between $\A$ and the rest of the fields.  \qed 
\end{remark}
These Poisson brackets can be constructed systematically by starting from the first bracket \eqref{eq:induced Poisson structure from classical limit} and deducing the Poisson brackets $\pois{\bpsi}{\bpsi},\pois{\bpsi}{\psi},$ and $\pois{\psi}{\psi}$. Then, we take the second bracket of \eqref{eq:induced Poisson structure from classical limit} and read off $\pois{\B}{\psi}$ and $\pois{\B}{\bpsi}$. Finally, by considering the third bracket of \eqref{eq:induced Poisson structure from classical limit}, $\pois{\B}{\B}$ can be proposed. 

\smallskip By the basic properties of Poisson structure, the left-hand side of the first relation in \eqref{eq:induced Poisson structure from classical limit} can be written as
\begin{equation}\label{eq:detailed computation of {O^i_j[0],O^k_l[0]}}
   \begin{aligned}
        \pois{\wt{\mcal{O}}^i_j[0]}{\wt{\mcal{O}}^k_l[0]}&=\pois{\bpsii{i}}{\bpsii{k}}\psi_l\psi_j-\bpsii{k}\pois{\bpsii{i}}{\psi_l}\psi_j
    \\
    &+\bpsii{i}\pois{\psi_j}{\bpsii{k}}\psi_l-\bpsii{i}\bpsii{k}\pois{\psi_j}{\psi_l},
   \end{aligned}
\end{equation}
where $-1$ in the second and last term comes from anti-commutativity of fermions \eqref{eq:anticommutativity of fermions}, and we have suppressed all $\mfk{gl}_N$ indices. For this to vanish, we define the Poisson brackets
\begin{equation}\label{eq:Poisson bracket of psi and barpsi}
    \pois{\psi^a_i}{\psi^b_j}:=0,\qquad \pois{\psi^a_i}{\bpsii{j}_b}:=+\delta_i^j\delta^a_b, \qquad  \pois{\bpsii{i}_a}{\bpsii{j}_b}:=0. 
\end{equation}

\smallskip Next, consider the second relation in \eqref{eq:induced Poisson structure from classical limit}
\begin{equation}\label{eq:detailed computation of {O^i_j[0],O^k_l[1]}}
    \begin{aligned}
        \pois{\wt{\mcal{O}}^i_j[0]}{\wt{\mcal{O}}^k_l[1]}&=\pois{\bpsii{i}}{\bpsii{k}}\B\psi_l\psi_j-\bpsii{k}\pois{\bpsii{i}}{\B}\psi_l\psi_j-\bpsii{k}\B\pois{\bpsii{i}}{\psi_l}\psi_j    
        \\
        &+\bpsii{i}\pois{\psi_j}{\bpsii{k}}\B\psi_l-\bpsii{i}\bpsii{k}\pois{\psi_j}{\B}\psi_l-\bpsii{i}\bpsii{k}\B\pois{\psi_j}{\psi_l},
    \end{aligned}
\end{equation}
which, upon using \eqref{eq:Poisson bracket of psi and barpsi}, leaves us with the Poisson brackets of $\psi$ and $\bpsi$ with $\B$, which are the second and fourth terms. For their sum to match the second term of \eqref{eq:induced Poisson structure from classical limit}, we propose the following Poisson brackets\footnote{There is a freedom of choosing $\pm$ sign in either of the brackets. Here, we have defined $\pois{\psi^a_i}{\B^b_c}$ with $+$ and $\pois{\bpsii{i}_a}{\B^b_c}$ with $-$. As we will see below, $\alpha$ itself can have only two different values (see \eqref{eq:values of alpha}) for the Jacobi identities to be satisfied.}
\begin{equation}\label{eq:Poisson brackets of psi and barpsi with the B field}
    \pois{\psi^a_i}{\B^b_c}:=+\alpha\delta^a_c\psi^b_i, \qquad \pois{\bpsii{i}_a}{\B^b_c}:=-\alpha\delta^b_a\bpsii{i}_c,
\end{equation}
where $\alpha$ is a coefficient which will be fixed when we prove the Jacobi identity for the Poisson bracket below. We then have
\begin{equation}
    \begin{aligned}
        \pois{\wt{\mcal{O}}^i_j[0]}{\wt{\mcal{O}}^k_l[1]}&=-\bpsii{k}\B\pois{\bpsii{i}}{\psi_l}\psi_j+\bpsii{i}\pois{\psi_j}{\bpsii{k}}\B\psi_l+\alpha\bpsii{i}_a\psi^a_l\bpsii{k}_b\psi^b_j-\alpha\bpsii{i}_a\psi^a_l\bpsii{k}_b\psi^b_j
        \\
        &=\delta^i_l\wt{\mcal{O}}^k_j[1]-\delta^k_j\wt{\mcal{O}}^i_l[1],
    \end{aligned}
\end{equation}
as claimed. 

\smallskip Finally, the third relation in \eqref{eq:induced Poisson structure from classical limit} gives
\begin{equation}\label{eq:detailed computation of {O^i_j[1],O^k_l[1]}}
    \begin{aligned}
        \big\{\wt{\mcal{O}}^i_j[1],\wt{\mcal{O}}^k_l[1]\big\}_{\star}&=
        \pois{\bpsii{i}}{\bpsii{k}}\B\psi_l\B\psi_j-\bpsii{k}\pois{\bpsii{i}}{\B}\psi_l\B\psi_j-\bpsii{k}\B\pois{\bpsii{i}}{\psi_l}\B\psi_j
        \\
        &+\bpsii{i}\pois{\B}{\bpsii{k}}\B\psi_l\psi_j+\bpsii{i}\bpsii{k}\pois{\B}{\B}\psi_l\psi_j+\bpsii{i}\bpsii{k}\B\pois{\B}{\psi_l}\psi_j
        \\
        &+\bpsii{i}\B\pois{\psi_j}{\bpsii{k}}\B\psi_l-\bpsii{i}\B\bpsii{k}\B\pois{\psi_j}{\psi_l}-\bpsii{i}\B\bpsii{k}\pois{\psi_j}{\B}\psi_l.
    \end{aligned}
\end{equation}
The four terms involving the Poisson bracket of $\B$ with $\psi$ or $\bpsi$ give
\begin{eqaligned}\label{eq:four terms involving Poisson bracket of B with fermionic fields}
    -\bpsii{k}\pois{\bpsii{i}}{\B}\psi_l\B\psi_j&=+\alpha\bpsii{k}_c\delta_a^c\bpsii{i}_d\psi_l^d\B^a_b\psi_j^b=+\alpha\wt{\mcal{O}}^i_l[0]\wt{\mcal{O}}^k_j[1],
    \\
    +\bpsii{i}\pois{\B}{\bpsii{k}}\B\psi_l\psi_j&=+\alpha\delta^a_c\bpsii{i}_a\bpsii{k}_b\B^c_d\psi_l^d\psi_j^b=-\alpha\wt{\mcal{O}}^k_j[0]\wt{\mcal{O}}^i_l[1],
    \\
    +\bpsii{i}\bpsii{k}\B\pois{\B}{\psi_l}\psi_j&=-\alpha\delta^d_b\bpsii{i}_a\bpsii{k}_c\B^c_d\psi^a_l\psi^b_j=+\alpha\wt{\mcal{O}}^i_l[0]\wt{\mcal{O}}^k_j[0],
    \\
    -\bpsii{i}\B\bpsii{k}\pois{\psi_j}{\B}\psi_l&=-\alpha\delta^b_d\bpsii{i}_a\B^a_b\bpsii{k}_c\psi^c_j\psi_l^d=-\alpha\wt{\mcal{O}}^k_j[0]\wt{\mcal{O}}^i_l[1],
\end{eqaligned}
which in total contribute $2\alpha(\wt{\mcal{O}}^i_l[0]\wt{\mcal{O}}^k_j[1]-\wt{\mcal{O}}^k_j[0]\wt{\mcal{O}}^i_l[1])$. Using this, \eqref{eq:Poisson bracket of psi and barpsi}, and \eqref{eq:Poisson brackets of psi and barpsi with the B field}, we propose the following Poisson bracket of $\B$ with itself
\begin{equation}\label{eq:Poisson bracket of B and B}
    \pois{\B^a_b}{\B^c_d}:=2\alpha\left(\delta^a_d\B^c_b-\delta^c_b\B^a_d\right)+\frac{1}{4}\left(\delta^a_d\bpsii{m}_b\psi^c_m-\delta^c_b\bpsii{m}_d\psi^a_m\right).
\end{equation}
The contribution of the first piece of this to $\bpsii{i}\bpsii{k}\pois{\B}{\B}\psi_l\psi_j$ is
\begin{equation}
    2\alpha\bpsii{i}_a\bpsii{k}_c\left(\delta^a_d\B^c_b-\delta^c_b\B^a_d\right)\psi_l^d\psi_j^b=2\alpha\left(-\wt{\mcal{O}}^i_l[0]\wt{\mcal{O}}^k_j[1]+\wt{\mcal{O}}^k_j[0]\wt{\mcal{O}}^i_l[1]\right),
\end{equation}
which precisely cancels the sum of contributions \eqref{eq:four terms involving Poisson bracket of B with fermionic fields}. Adding the contribution of the second term of \eqref{eq:Poisson bracket of B and B} to $\pois{\B}{\B}$, we have
\begin{equation}
    \begin{aligned}
       \big\{\wt{\mcal{O}}^i_j[1],\wt{\mcal{O}}^k_l[1]\big\}_{\star}&=\delta^i_l\wt{\mcal{O}}^k_j[2]-\delta^k_j\wt{\mcal{O}}^i_l[2]+\frac{1}{4}\left(\bpsii{m}_a\psi^a_j\bpsii{k}_b\psi^b_m\bpsii{i}_c\psi^c_l-\bpsii{m}_a\psi^a_l\bpsii{i}_b\psi^b_m\bpsii{k}_c\psi^c_j\right)
       \\
       &=\delta^i_l\wt{\mcal{O}}^k_j[2]-\delta^k_j\wt{\mcal{O}}^i_l[2]+\frac{1}{4}\left(\wt{\mcal{O}}^m_j[0]\wt{\mcal{O}}^k_m[0]\wt{\mcal{O}}^i_l[0]-\wt{\mcal{O}}^m_l[0]\wt{\mcal{O}}^i_m[0]\wt{\mcal{O}}^k_j[0]\right),
    \end{aligned}
\end{equation}
which is the last term of \eqref{eq:induced Poisson structure from classical limit}. 

\smallskip Finally, we need to show that \eqref{eq:Poisson bracket of psi and barpsi}, \eqref{eq:Poisson brackets of psi and barpsi with the B field}, and \eqref{eq:Poisson bracket of B and B} satisfy the graded Jacobi identity. Denoting the grading of a field $\Phi$ by $\deg \Phi$, the graded Jacobi identity is
\begin{equation}
    \begin{aligned}
        \pois{\Phi}{\pois{\Phi'}{\Phi''}}&+(-1)^{\deg\Phi(\deg\Phi'+\deg\Phi'')}\pois{\Phi'}{\pois{\Phi''}{\Phi}}
        \\
        &+(-1)^{(\deg\Phi+\deg\Phi')\deg\Phi''}\pois{\Phi''}{\pois{\Phi}{\Phi'}}=0.
    \end{aligned}
\end{equation}
Due to \eqref{eq:Poisson bracket of psi and barpsi}, for $(\Phi,\Phi',\Phi'')$ involving fermionic fields or two fermionic fields of the same type and $\B$ such as $(\psi,\psi,\B)$, there is nothing to show as the Jacobi identity is trivial. The first non-trivial case is $(\psi,\bpsi,\B)$, where we have
\begin{eqaligned}
    \pois{\psi^a_i}{\pois{\bpsii{j}_b}{\B^c_d}}&=-\alpha\delta^c_b\delta^a_d\delta^j_i,
    \\
    \pois{\bpsii{j}_b}{\pois{\B^c_d}{\psi^a_i}}&=-\alpha\delta^a_d\delta^c_b\delta^j_i,
    \\
    \pois{\B^c_d}{\pois{\psi^a_j}{\bpsii{j}_b}}&=0,
\end{eqaligned}
and hence
\begin{equation}
    \pois{\psi^a_i}{\pois{\bpsii{j}_b}{\B^c_d}}-\pois{\bpsii{j}_b}{\pois{\B^c_d}{\psi^a_i}}+\pois{\B^c_d}{\pois{\psi^a_j}{\bpsii{j}_b}}=0.
\end{equation}

For $(\psi,\B,\B)$, we use \eqref{eq:Poisson brackets of psi and barpsi with the B field} to compute the followings
\begin{equation}
    \begin{aligned}
        \pois{\psi^a_i}{\pois{\B^b_c}{\B^d_e}}&=\left(2\alpha^2+\frac{1}{4}\right)\left(\delta^b_e\delta^a_c\psi^d_i-\delta^d_c\delta^a_e\psi^b_i\right),
        \\
        \pois{\B^b_c}{\pois{\B^d_e}{\psi^a_i}}&=+\alpha^2\delta^a_e\delta^d_c\psi^b_i
        \\
        \pois{\B^d_e}{\pois{\psi^a_i}{\B^b_c}}&=-\alpha^2\delta^a_c\delta^b_e\psi^d_i,
    \end{aligned}
\end{equation}
and hence
\begin{equation}
    \pois{\psi^a_i}{\pois{\B^b_c}{\B^d_e}}+\pois{\B^b_c}{\pois{\B^d_e}{\psi^a_i}}+\pois{\B^d_e}{\pois{\psi^a_i}{\B^b_c}}=\left(\alpha^2+\frac{1}{4}\right)\left(\delta^b_e\delta^a_c\psi^d_i-\delta^d_c\delta^a_e\psi^b_i\right).
\end{equation}
For this to vanish, we need to fix
\begin{equation}\label{eq:values of alpha}
    \alpha=\pm\frac{\mfk{i}}{2}.
\end{equation}
As we are working with a complex Lie algebra ($\mfk{gl}_K$), it is acceptable to have complex structure constants in \eqref{eq:Poisson brackets of psi and barpsi with the B field} and \eqref{eq:Poisson bracket of B and B}. A similar computation, using \eqref{eq:Poisson brackets of psi and barpsi with the B field} and \eqref{eq:values of alpha}, gives 
\begin{equation}
    \pois{\bpsi^a_i}{\pois{\B^b_c}{\B^d_e}}+\pois{\B^b_c}{\pois{\B^d_e}{\bpsi^a_i}}+\pois{\B^d_e}{\pois{\bpsi^a_i}{\B^b_c}}=0.
\end{equation}
Finally, we consider the Poisson bracket of three $\B$ fields
\begin{eqaligned}
    \pois{\B^a_b}{\pois{\B^c_d}{\B^e_f}}&=4\alpha\delta^c_f\delta^a_d\left(\alpha\B^e_b+\frac{1}{16}\bpsii{m}_b\psi^e_m\right)-4\alpha\delta^c_f\delta^e_b\left(\alpha\B^a_d+\frac{1}{16}\bpsii{m}_d\psi^a_m\right)
    \\
    &+4\alpha\delta^e_d\delta^c_b\left(\alpha\B^a_f+\frac{1}{16}\bpsii{m}_f\psi^a_m\right)-4\alpha\delta^e_d\delta^a_f\left(\alpha\B^c_b+\frac{1}{16}\bpsii{m}_b\psi^c_m\right),
    \\
    \pois{\B^c_d}{\pois{\B^e_f}{\B^a_b}}&=4\alpha\delta^e_b\delta^c_f\left(\alpha\B^a_d+\frac{1}{16}\bpsii{m}_d\psi^a_m\right)-4\alpha\delta^e_b\delta^a_d\left(\alpha\B^c_f+\frac{1}{16}\bpsii{m}_f\psi^c_m\right)
    \\
    &+4\alpha\delta^a_f\delta^e_d\left(\alpha\B^c_b+\frac{1}{16}\bpsii{m}_b\psi^c_m\right)-4\alpha\delta^a_f\delta^c_b\left(\alpha\B^e_d+\frac{1}{16}\bpsii{m}_d\psi^b_m\right),
    \\
    \pois{\B^e_f}{\pois{\B^a_b}{\B^c_d}}&=4\alpha\delta^a_d\delta^e_b\left(\alpha\B^c_f+\frac{1}{16}\bpsii{m}_f\psi^c_m\right)-4\alpha\delta^a_d\delta^c_f\left(\alpha\B^e_b+\frac{1}{16}\bpsii{m}_b\psi^e_m\right)
    \\
    &+4\alpha\delta^c_b\delta^a_f\left(\alpha\B^e_d+\frac{1}{16}\bpsii{m}_d\psi^e_m\right)-4\alpha\delta^c_b\delta^e_d\left(\alpha\B^a_f+\frac{1}{16}\bpsii{m}_f\psi^a_m\right),
\end{eqaligned}
and as such
\begin{equation}
  \pois{\B^a_b}{\pois{\B^c_d}{\B^e_f}}+\pois{\B^c_d}{\pois{\B^e_f}{\B^a_b}}+\pois{\B^e_f}{\pois{\B^a_b}{\B^c_d}}=0,
\end{equation}
irrespective of the value of $\alpha$. 

\smallskip Therefore, we have successfully constructed a Poisson structure $\pois{\cdot}{\cdot}$ on $\C[\mcal{M}(N,K)]$ from the classical (i.e. $\hbar\to 0$) limit of commutators of generators of $Y(\mfk{gl}_K)$, which has been realized as the full algebra of observables of 2d BF theory coupled to a 1d quantum mechanics in \cite{IshtiaqueMoosavianZhou201809}. We conclude with the summary of these results as follows

\begin{theorem}\label{thr:Poisson structure inspired by twisted holography}
    The commutator in $\wh{\mcal{A}}_{\normalfont\text{ob}}$, the algebra of observables of 2d BF theory coupled to a 1d quantum mechanics, induces, in the classical {\normalfont(}i.e. $\hbar\to 0${\normalfont)} limit, a graded Poisson structure on $\C[\mcal{M}(N,K)]$ {\normalfont(}and hence on  $\mcal{A}_{\normalfont\text{ob}}${\normalfont)}.
    Furthermore, $\pois{\cdot}{\cdot}$ equips $\C[\mcal{M}(N,K)]$ with a graded Poisson structure where the Poisson brackets between the elementary fields are given by
    \begin{equation}\label{eq:Poisson structure of elementary fields in the fermionic case}
        \begin{gathered}
            \pois{\psi^a_i}{\psi^b_j}:=0,\qquad \pois{\psi^a_i}{\bpsii{j}_b}:=+\delta_i^j\delta^a_b, \qquad  \pois{\bpsii{i}_a}{\bpsii{j}_b}:=0,
            \\
            \pois{\psi^a_i}{\B^b_c}:=+\alpha\delta^a_c\psi^b_i, \qquad \pois{\bpsii{i}_a}{\B^b_c}:=-\alpha\delta^b_a\bpsii{i}_c,
            \\
            \pois{\B^a_b}{\B^c_d}:=2\alpha\left(\delta^a_d\B^c_b-\delta^c_b\B^a_d\right)+\frac{1}{4}\left(\delta^a_d\bpsii{m}_b\psi^c_m-\delta^c_b\bpsii{m}_d\psi^a_m\right),
        \end{gathered}
    \end{equation}
    where $\alpha$ is a parameter, which is fixed by the Jacobi identities to be
    \begin{equation}
        \alpha=\pm\frac{\mfk{i}}{2}. 
    \end{equation}
\end{theorem}

\subsection{Fermionic vs Bosonic Quantum Mechanics}\label{sec:fermionic vs bosonic quantum mechanics}

 In \cite{IshtiaqueMoosavianZhou201809}, the quantum-mechanical system that couples to the 2d BF theory has fermionic degrees of freedom. As such, as we realized above, we ended up with a graded Poisson structure in the classical limit. However, nothing prevent us from coupling our 2d BF theory to 1d quantum mechanics with bosonic degrees of freedom. From the standard perspective of holography, it is indeed more natural to consider fermionic degrees of freedom on the line since integrating them out would naturally leads to Wilson lines. From the perspective of holography studied in \cite{IshtiaqueMoosavianZhou201809}, different descriptions are related to the different bulk descriptions: fermionic degrees of freedoms are related to configurations of D3--D5 branes in the bulk, while bosonic degrees of freedom, related to the fermionic ones by bosonization, are captured by an alternative but equivalent description in terms of D3--D3 brane system in the bulk, as realized in \cite{GomisPasserini200604,GomisPasserini200612}. 

\smallskip If we switch to bosonic degrees of freedom as
\begin{equation}
    \bpsii{i}_a\mapsto \ovl{\phi}^i_a, \qquad \psi^a_i\mapsto\phi^a_i,
\end{equation}
the computations in \cite{IshtiaqueMoosavianZhou201809} will not be affected, as they should not. This is due to the fact that although we have now commuting degrees of freedom, the propagators $\langle\ovl{\phi}^i_a\phi^b_j\rangle$ which enters the computations involves an extra minus sign, as explained in \cite[Remark 3, pg. 14]{IshtiaqueMoosavianZhou201809}. Therefore, if we go through the same exercise as the fermionic case, we end up with a  Poisson structure on $\mcal{A}_{\text{ob}}$, which is now generated by the gauge-invariant functions $\mcal{O}^i_j[n]:=\ovl{\phi}^i\B^n\phi_j$ and $\mcal{O}[n]:=\tr_{\C^N}\B^n$. However, this Poisson structure is not graded anymore. The only difference in the computations similar to \eqref{eq:detailed computation of {O^i_j[0],O^k_l[0]}}, \eqref{eq:detailed computation of {O^i_j[0],O^k_l[1]}}, and \eqref{eq:detailed computation of {O^i_j[1],O^k_l[1]}} is that all signs are now $+$. However, the Poisson structure, which we denote as $\poisb{\cdot}{\cdot}$, will be similar to $\pois{\cdot}{\cdot}$. They can be read-off from \eqref{eq:Poisson structure of elementary fields in the fermionic case}
\begin{equation}\label{eq:Poisson structure of elementary fields from bosonic quantum mechanics}
        \begin{gathered}
            \poisb{\phi^a_i}{\phi^b_j}=0,\qquad \poisb{\ovl{\phi}^{i}_a}{\phi^b_j}=+\delta_i^j\delta^a_b, \qquad  \poisb{\ovl{\phi}^i_a}{\ovl{\phi}^j_b}=0,
            \\
            \poisb{\phi^a_i}{\B^b_c}=+\alpha\delta^a_c\phi^i_b, \qquad \poisb{\ovl{\phi}^{i}_a}{\B^b_c}=-\alpha\delta_a^b\ovl{\phi}^{\,i}_c
            \\
            \poisb{\B^a_b}{\B^c_d}:=2\alpha\left(\delta^a_d\B^c_b-\delta^c_b\B^a_d\right)+\frac{1}{4}\left(\delta^a_d\ovl{\phi}^{\,m}_b\phi^c_m-\delta^c_b\ovl{\phi}^{\,m}_d\phi^a_m\right),
        \end{gathered}
    \end{equation} 
with $\alpha$ as in \eqref{eq:values of alpha}. Furthermore, none of the Jacobi identity computations is dependent on the fermionic nature of the degrees of freedom (only all $-\mapsto +$). Hence, $\poisb{\cdot}{\cdot}$ is indeed a Poisson structure on $\C[\mcal{M}(N,K)]$, which is now not graded as well. Furthermore, as in the case of fermionic degrees of freedom, this Poisson structure equips $\mcal{A}_{\text{ob}}$ with the structure of a commutative Poisson algebra. We thus see that the only effect of switching from fermionic to bosonic degrees of freedom on the line is that we end up with an ordinary Poisson structure, rather than a graded one, on the algebra of classical observables. 

\subsection{Equivalence of Poisson Structures \texorpdfstring{\eqref{Poisson before quotient}}{} and \texorpdfstring{\eqref{eq:Poisson structure of elementary fields from bosonic quantum mechanics}}{}} 
A natural question is to compare the Poisson structures we introduced above, that is, \eqref{Poisson before quotient} and \eqref{eq:Poisson structure of elementary fields from bosonic quantum mechanics}. To relate the two, we introduce a new $\B$ field, which we denote as $\wt{\B}$
\begin{equation}\label{eq:generic form of tildeB}
    \wt{\B}^a_b:=p\B^a_b+q\ovl{\phi}^{\,m}_b\phi^a_m,
\end{equation}
for some $p$ and $q$ to be determined below. Then, 
\begin{eqaligned}
    \poisb{\phi^a_i}{\wt{\B}^b_c}&=\poisb{\phi^a_i}{p\B^b_c+q\ovl{\phi}^{\,m}_c\phi^b_m}
    \\
    &=p\poisb{\phi^a_i}{\B^b_c}+q\poisb{\phi^a_i}{\ovl{\phi}^{\,m}_c\phi^b_m}
    \\
    &=\alpha p\delta_c^a\phi^b_i+q\delta^m_i\delta^a_c\phi^b_m
    \\
    &=\delta^a_c(\alpha p+ q)\phi^b_i.
\end{eqaligned}
To be able to compare this with \eqref{Poisson before quotient}, we set this to zero which gives the constraint
\begin{equation}\label{eq:relation between p and q}
    q=-\alpha p.
\end{equation}
Next, consider 
\begin{eqaligned}
    \poisb{\wt{\B}^a_b}{\wt{\B}^c_d}&=p^2\poisb{\B^a_b}{\B^c_d}+pq\poisb{\B^a_b}{\ovl{\phi}^{\,m}_d\phi^c_m}+pq\poisb{\ovl{\phi}^{\,m}_b\phi^a_m}{\B^c_d}+q^2\poisb{\ovl{\phi}^{\,m}_b\phi^a_m}{\ovl{\phi}^{\,n}_d\phi^c_n}
    \\
    &=2\alpha p^2(\delta^a_d\B^c_b-\delta^c_b\B^a_d)+\frac{p^2}{4}\left(\delta^a_d\ovl{\phi}^{\,m}_b\phi^c_m-\delta^c_b\ovl{\phi}^{\,m}_d\phi^a_m\right)+\alpha pq(\delta^a_d\ovl{\phi}^{\,m}_b\phi^c_m-\delta^c_b\ovl{\phi}^{\,m}_d\phi^a_m)
    \\
    &+\alpha pq(\delta^a_d\ovl{\phi}^{\,m}_b\phi^c_m-\delta^c_b\ovl{\psi}^{\,m}_d\phi^a_m)+q^2(\delta^a_d\ovl{\phi}^{\,m}_b\phi^c_m-\delta^c_b\ovl{\phi}^{\,m}_d\phi^a_m)
    \\
    &=2\alpha p(\delta^a_d(p\B^c_b+q\ovl{\phi}^{\,m}_b\phi^c_m)-\delta^c_b(p\B^a_d+q\ovl{\phi}^{\,m}_d\phi^a_m))+\left(\frac{p^2}{4}+q^2\right)(\delta^a_d\ovl{\phi}^{\,m}_b\phi^c_m-\delta^c_b\ovl{\phi}^{\,m}_d\phi^a_m).
    \\
    &=2\alpha p(\delta^a_d\wt{\B}^c_b-\delta^c_b\wt{\B}^a_d)+\left(\frac{p^2}{4}+q^2\right)(\delta^a_d\ovl{\phi}^{\,m}_b\phi^c_m-\delta^c_b\ovl{\phi}^{\,m}_d\phi^a_m).
\end{eqaligned}
First of all, notice that
\begin{equation}
    \frac{p^2}{4}+q^2=\frac{q^2}{4\alpha^2}+q^2=-q^2+q^2=0,
\end{equation}
where we have used \eqref{eq:values of alpha}. Then, comparing with \eqref{Poisson before quotient}, shows that if we take
\begin{equation}
    p=\frac{1}{2\alpha},
\end{equation}
then, we have
\begin{equation}
    \poisb{\wt{\B}^a_b}{\wt{\B}^c_d}=\delta^a_d\wt{\B}^c_b-\delta^c_b\wt{\B}^a_d.
\end{equation}
Therefore, we end-up with the following relation between symplectic structures in \eqref{Poisson before quotient} and \eqref{eq:Poisson structure of elementary fields from bosonic quantum mechanics} 
\begin{equation}\label{eq:relation between the natural Poisson brackets vs the one inspired by the twisted holography}
    \text{$(B,I,J,\{\cdot,\cdot\})$ in \eqref{Poisson before quotient}}, \qquad \Longleftrightarrow \qquad \text{$(\wt{\B}_\pm,\ovl{\phi},\phi;\poisb{\cdot}{\cdot})$ in \eqref{eq:Poisson structure of elementary fields from bosonic quantum mechanics}},
\end{equation}
where from \eqref{eq:generic form of tildeB}, \eqref{eq:relation between p and q}, and \eqref{eq:values of alpha}, we have
\begin{equation}\label{eq:redefinition of B field}
    (\wt{\B}_\pm)^a_b=\pm\mfk{i}\B^a_b-\frac{1}{2}\ovl{\phi}^{\,m}_b\phi^a_m,
\end{equation}
and either sign can be used but has to be fixed once and for all. 

\smallskip Finally, let us point out that since the Poisson bracket $\pois{\cdot}{\cdot}$ in \eqref{eq:Poisson structure of elementary fields in the fermionic case} is graded, it cannot be directly compared with \eqref{Poisson before quotient}. That was the reason we compared the latter with \eqref{eq:Poisson structure of elementary fields from bosonic quantum mechanics}, which is not graded.

\section{Hall--Littlewood Polynomials}
\label{sec: Hall-Littlewood Polynomials}

In this appendix we review some background on symmetric functions, following \cite[\S 3]{haiman2003combinatorics}.
\begin{definition}
For a tuple of integers $\lambda=(\lambda_1\ge\cdots\ge \lambda_n)\in \bN^n$, we denote its associated partition by $(1^{\alpha_1},2^{\alpha_2},\cdots)$, then the Hall--Littlewood polynomial $P_{\lambda}(x;q)$ in the variables $x=(x_1,\cdots,x_n)$ and $q$ is defined by the formula
\begin{align}
    P_{\lambda}(x;q):=\frac{1}{\prod_{i\ge 0}[\alpha_i]_q!}\sum_{w\in \mathfrak S_n}w\left(x^{\lambda}\prod_{i<j}\frac{1-qx_j/x_i}{1-x_j/x_i}\right).
\end{align}
Here $\alpha_0=n-\sum_{i\ge 1}\alpha_i$, and $x^{\lambda}=x_1^{\lambda_1}\cdots x_n^{\lambda_n}$, and we use the following $q$-number notation
\begin{align*}
    [n]_q=\frac{1-q^n}{1-q},\;  [n]_q!=[n]_q[n-1]_q\cdots [1]_q,\; \left[\begin{matrix} n \\ k \end{matrix}\right]_q=\frac{[n]_q!}{[k]_q![n-k]_q!}.
\end{align*}
\end{definition}
The Hall--Littlewood polynomial $P_{\lambda}(x;q)$ is an interpolation between Schur symmetric functions $s_{\lambda}(x)$ and monomial symmetric functions $m_{\lambda}(x)$, in fact we have 
\begin{align}\label{P(0) and P(1)}
    P_{\lambda}(x;0)=s_{\lambda}(x),\; P_{\lambda}(x;1)=m_{\lambda}(x).
\end{align}

\begin{definition}
The Kostka-Foulkes functions are coefficients of the expansion
\begin{align}
    s_{\lambda}(x)=\sum_{\mu}K_{\lambda \mu}(q)P_{\mu}(x;q).
\end{align}
In particular, by \eqref{P(0) and P(1)} we have
\begin{align*}
    K_{\lambda \mu}(0)=\delta_{\lambda \mu}.
\end{align*}
\end{definition}

\paragraph{Jing Operators and Transformed Hall--Littlewood Polynomials.} N. Jing found a definition of Hall--Littlewood polynomials using vertex algebra \cite{jing1991vertex}. Before giving his definition, we recall some plethystic notations.

\smallskip The ring of symmetric functions, denoted by $\Lambda$, is freely generated by power sum functions $p_k$, that is 
\begin{align*}
    \Lambda=\bC[p_k:k\in \mathbb Z_{\ge 1}].
\end{align*}
Let $A$ be a formal Laurent series in indeterminates $a_1, a_2,\cdots$, we define $p_k[A]$ to be the result of replacing each indeterminate $a_i$ in A by $a^k_i$. Then for any $f\in \Lambda$, the plethystic substitution of $A$ into $f$, denoted $f[A]$, is the image of $f$ under the homomorphism sending $p_k$ to $p_k[A]$.

\begin{example}
We list some special cases here.

\begin{itemize}
    \item Let $A=a_1+\cdots +a_n$, then $p_k[A]=a_1^k+\cdots +a_n^k=p_k(a_1,\cdots,a_n)$, and thus for any $f\in \Lambda$, we have $f[A]=f(a_1,\cdots,a_n)$.
    \item Let $A,B$ be formal Laurent series, then $p_k[A\pm B]=p_k[A]\pm p_k[B]$.
    \item Let $\mathrm{PE}=\exp{\left(\sum_{k=1}^{\infty}p_k/k\right)}$, then we have
    \begin{align*}
        \mathrm{PE}[A+B]=\mathrm{PE}[A]\mathrm{PE}[B],\; \mathrm{PE}[A-B]=\mathrm{PE}[A]/\mathrm{PE}[B].
    \end{align*}
    For a single variable $x$, we have $\mathrm{PE}(x)=\frac{1}{1-x}$, thus for a summation $X=x_1+x_2+\cdots$,
    \begin{align*}
        \mathrm{PE}(X)=\prod_{i\ge 1}\frac{1}{1-x_i},\; \mathrm{PE}(-X)=\prod_{i\ge 1}(1-x_i).
    \end{align*}
\end{itemize}
\end{example}

\noindent For the rest of this section, we fix the notation $X=x_1+x_2+\cdots$.
\begin{definition}
The Jing operators are the coefficients $S^q(u)=\sum_{m\in \bZ}S^q_m u^m$ of the operator generating function $S^q(u)$ defined by 
\begin{align}
    S^q(u)f=f[X+(q-1)u^{-1}]\mathrm{PE}[uX].
\end{align}
\end{definition}

\begin{proposition}[{\cite[Proposition 2.12]{jing1991vertex}}]
Jing operators $S^q_m$ satisfy relations \footnote{We note that our $q$ is denoted by $t$ there and our $S^q_m$ is denoted by $H_{-m}$ there.}$:$
\begin{align}\label{eqn: relation of Jing operators}
    S^q_{n}S^q_{m+1}-qS^q_{m+1}S^q_{n}=qS^q_{n+1}S^q_m-S^q_mS^q_{n+1}.
\end{align}
\end{proposition}

\begin{definition}
Let $\mu=(\mu_1\ge\cdots \ge \mu_l)\in \mathbb Z_{\ge 0}^l$ be a tuple of non-increasing integers, define the \textit{transformed Hall--Littlewood polynomial} by
\begin{align}\label{definition: transformed HL}
    H_{\mu}(x;q)=S^q_{\mu_1}S^q_{\mu_2}\cdots S^q_{\mu_l}(1).
\end{align}
For a general array $\mu=(\mu_1,\cdots,\mu_l)\in \mathbb Z_{\ge 0}^l$, we define the \textit{generalized} transformed Hall--Littlewood polynomial by the same formula above.
\end{definition}
Using relations \eqref{eqn: relation of Jing operators} recursively, we can bring a product of operators $S^q_{\mu_1}\cdots S^q_{\mu_l}$ for an array $\mu=(\mu_1,\cdots,\mu_l)\in \mathbb Z_{\ge 0}^l$ into a linear combination of operators $S^q_{\mu_1'}\cdots S^q_{\mu_l'}$ such that $\mu_1'\ge \cdots\ge \mu_l'$, in other words, a generalized transformed Hall--Littlewood polynomial can be written as linear combination of usual transformed Hall--Littlewood polynomials.

The following proposition summarizes the fundamental properties of transformed Hall--Littlewood polynomials.
\begin{proposition}[{\cite[3.4.3]{haiman2003combinatorics}}]
The transformed Hall--Littlewood polynomials $H_{\mu}$ are related to the classical Hall--Littlewood polynomials $P_{\mu}$ by
\begin{align}
    H_{\mu}[(1-q)X;q]=(1-q)^{l(\mu)}\prod_{i=1}^{\mu_1}[\alpha_i(\mu)]_q!P_{\mu}(x;q).
\end{align}
They are uniquely characterized by the following properties.
\begin{itemize}
    \item[(i)] $H_{\mu}(x;q)\in s_{\mu}(x)+\mathbb Z[q]\cdot \{s_{\lambda}(x):\lambda>\mu\}$,
    \item[(ii)] $H_{\mu}[(1-q)x;q]\in \mathbb Z[q]\cdot \{s_{\lambda}(x):\lambda\le \mu\}$.
\end{itemize}
And $H_{\mu}$ is related to Schur functions by
\begin{align}
    H_{\mu}(x;q)=\sum_{\lambda}K_{\lambda \mu}(q)s_{\lambda}(x).
\end{align}
\end{proposition}

We can write down the action of Jing operators $S^q_m$ explicitly as in the following lemma.
\begin{lemma}
For an $n$-variable function $f\in \bC[p_1,\cdots,p_n](q)$, where $p_{k}(x)=x_1^k+\cdots +x_n^k$, Jing operator $S^q_m$ acts on it as
\begin{align}\label{Jing Operator}
    (S^q_mf)(x;q)=\sum_{i=1}^nf(x_1,\cdots,qx_i,\cdots,x_n; q)\frac{x_i^m}{\prod_{j\neq i}(1-x_j/x_i)}.
\end{align}
\end{lemma}

\begin{proof}
Notice that
\begin{align}
    \mathrm{PE}(uX)=\prod_{i= 1}^n\frac{1}{1-ux_i}=\sum_{i=1}^n\frac{1}{1-ux_i}\prod_{j\neq i}\frac{1}{1-x_j/x_i}.
\end{align}
Without loss of generality, we assume that $f=p_{k_1}\cdots p_{k_s}$, then by definition, $S^q_m$ is the coefficient of $u^m$ in the series expansion 
\begin{align*}
    (p_{k_1}+(q^{k_1}-1)u^{-k_1})\cdots (p_{k_s}+(q^{k_s}-1)u^{-k_s})\sum_{i=1}^n\frac{1}{1-ux_i}\prod_{j\neq i}\frac{1}{1-x_j/x_i}.
\end{align*}
For the $i$-th summand, its $u^m$ coefficient is
\begin{align*}
    &(p_{k_1}+(q^{k_1}-1)x_i^{k_1})\cdots (p_{k_s}+(q^{k_s}-1)x_i^{k_s})\frac{x_i^m}{\prod_{j\neq i}(1-x_j/x_i)}\\
    &~=(x_1^{k_1}+\cdots+q^{k_1}x_i^{k_1}+\cdots+x_n^{k_1})\cdots (x_s^{k_s}+\cdots+q^{k_s}x_i^{k_s}+\cdots+x_n^{k_s})\frac{x_i^m}{\prod_{j\neq i}(1-x_j/x_i)}\\
    &~=f(x_1,\cdots,qx_i,\cdots,x_n; q)\frac{x_i^m}{\prod_{j\neq i}(1-x_j/x_i)}.
\end{align*}
Summing over $i$ gives the desired formula \eqref{Jing Operator}.
\end{proof}

\section{Affine Grassmannians and Geometrization of Jing Operators}
\label{appsec:affine grassamannians and geometrization of Jing operators}
In this appendix, we give a geometric definition of Jing operators $S^q_m$. Recall that 
\begin{align}
    K_{\GL_n\times \mathbb C^{\times}}(\mathrm{pt})\otimes_\bZ \bC=\bC[x_1^{\pm},\cdots,x_n^{\pm},q^{\pm}]^{\mathfrak S_n}.
\end{align}
To simplify notation, we will abbreviate $K_{\GL_n\times \mathbb C^{\times}}(\mathrm{pt})\otimes_\bZ \bC$ to $K_{\GL_n\times \mathbb C^{\times}}(\mathrm{pt})$. Notice that there is a subalgebra $\bC[p_1,\cdots,p_n,q^{\pm}]\subset K_{\GL_n\times \mathbb C^{\times}}(\mathrm{pt})$, where $p_i=x_1^i+\cdots+x_n^i$ is the $i$-th power sum polynomial.

Consider the affine Grassmannian $\Gr_{\GL_n}=\GL_n(\mathscr K)/\GL_n(\mathscr O)$, and let $\omega_1=(1,0,\cdots,0)$ be the first fundamental coweight of $\GL_n$, then the $\GL_n(\mathscr O)$-orbit $\Gr^{\omega_1}$ is isomorphic to $\mathbb P^{n-1}$ and it is fixed by the $\mathbb C^{\times}$-rotation.

Consider the $\GL_n(\mathscr O)\rtimes \mathbb C^{\times}$-equivariant bounded derive category of coherent sheaves on $\Gr_{\GL_n}$, denoted by $D^b_{\GL_n(\mathscr O)\rtimes \mathbb C^{\times}}(\Gr_{\GL_n})$. Here coherent sheaves on ind-scheme like $\Gr_{\GL_n}$ are defined to have finite type support. In particular, for any $\mathcal F\in D^b_{\GL_n(\mathscr O)\rtimes \mathbb C^{\times}}(\Gr_{\GL_n})$, we have $\chi (\mathcal F)\in K_{\GL_n(\mathscr O)\rtimes \mathbb C^{\times}}(\mathrm{pt})=K_{\GL_n\times \mathbb C^{\times}}(\mathrm{pt})$.

There is a convolution product on affine Grassmannian, defined as:
\begin{align}\label{eqn: convolution}
    m:\Gr_{\GL_n}\widetilde{\times} \Gr_{\GL_n}=\GL_n(\mathscr K)\overset{\GL_n(\mathscr O)}{\times}\GL_n(\mathscr K)/\GL_n(\mathscr O)\to \GL_n(\mathscr K)/\GL_n(\mathscr O).
\end{align}
Here the map sends $(g_1,g_2)$ to $g_1g_2$. The convolution map of $\Gr_{\GL_n}$ induces a functor $\star : D^b_{\GL_n(\mathscr O)\rtimes \mathbb C^{\times}}(\Gr_{\GL_n})\times D^b_{\GL_n(\mathscr O)\rtimes \mathbb C^{\times}}(\Gr_{\GL_n})\to D^b_{\GL_n(\mathscr O)\rtimes \mathbb C^{\times}}(\Gr_{\GL_n})$ defined as 
\begin{align*}
    \mathcal F\star \mathcal G=\mathbf Rm_*(\mathcal F\widetilde{\boxtimes} \mathcal G).
\end{align*}
Passing to the $K$-theory, we obtain an map 
\begin{align}
    \star: K_{\GL_n(\mathscr O)\rtimes \mathbb C^{\times}}(\Gr_{\GL_n})\otimes K_{\GL_n(\mathscr O)\rtimes \mathbb C^{\times}}(\Gr_{\GL_n})\longrightarrow K_{\GL_n(\mathscr O)\rtimes \mathbb C^{\times}}(\Gr_{\GL_n}).
\end{align}
It is known that the $\star$-product on $K_{\GL_n(\mathscr O)\rtimes \mathbb C^{\times}}(\Gr_{\GL_n})$ is associative and its classical limit $K_{\GL_n(\mathscr O)}(\Gr_{\GL_n})$ is commutative, this is an example of K-theoretic Coulomb branch in the sense of \cite{braverman2016coulomb}.

Using the convolution algebra $K_{\GL_n(\mathscr O)\rtimes \mathbb C^{\times}}(\Gr_{\GL_n})$ we can realize Jing operators $S^q_m$ geometrically as follows. The determinant line bundle $\mathcal O(1)$ on $\Gr_{\GL_n}$ \cite[1.5]{zhu2016introduction} is $\GL_n(\mathscr O)\rtimes \mathbb C^{\times}$-equivariant construction. Let us use $\mathcal O_{\Gr^{\omega_1}}(m)$ to denote $i_*i^*\mathcal O(1)^{\otimes m}$ where $i:\Gr^{\omega_1}\hookrightarrow \Gr_{\GL_n}$ is the natural embedding. Since $i$ is $\GL_n(\mathscr O)\rtimes \mathbb C^{\times}$-equivariant, $\mathcal O_{\Gr^{\omega_1}}(m)$ is also $\GL_n(\mathscr O)\rtimes \mathbb C^{\times}$-equivariant.

\begin{proposition}
For $\mathcal F\in D^b_{\GL_n(\mathscr O)\rtimes \mathbb C^{\times}}(\Gr_{\GL_n})$, let $\chi=\chi(\mathcal F)\in \bC[x_1^{\pm},\cdots,x_n^{\pm},q^{\pm}]^{S_n}$ be the equivariant Euler characteristic of $\mathcal F$, similarly let $\widetilde{\chi}=\chi(\mathcal O_{\Gr^{\omega_1}}(m)\star \mathcal F)$. Then
\begin{align}\label{Convolution}
    \widetilde{\chi}(x;q)=\sum_{i=1}^n\chi(x_1,\cdots,qx_i,\cdots,x_n; q)\frac{x_i^m}{\prod_{j\neq i}(1-x_j/x_i)}.
\end{align}
\end{proposition}

\begin{proof}
Let $p:\Gr_{\GL_n}\widetilde{\times} \Gr_{\GL_n}\to \Gr_{\GL_n}$ be the projection to the first component map, i.e. $p(g_1,g_2)=g_1$, this is a fibration with fibers isomorphic to $\Gr_{\GL_n}$. Then by the projection formula we have
\begin{align}\label{Projection}
    \chi(\mathcal O_{\Gr^{\omega_1}}(m)\star \mathcal F)=\chi(\mathbb P^{n-1},\mathcal O(m)\otimes \mathbf L i^*\mathbf R p_*\widetilde{\mathcal F}).
\end{align}
Here $\widetilde{\mathcal F}=\mathcal O\widetilde \boxtimes \mathcal F$ is the twist of $\mathcal F$ on $\Gr_{\GL_n}\widetilde{\times} \Gr_{\GL_n}$. We use the localization on $\mathbb P^{n-1}$ to compute the right hand side of \eqref{Projection} as following. Let the maximal torus of $\GL_n$ be $T$, then $T$-fixed points of $\mathbb P^{n-1}$ are $[1,0,\cdots,0],\cdots, [0,\cdots,1,\cdots,0],\cdots,[0,\cdots,1]$ (in homogeneous coordinates of $\mathbb P^{n-1}$), label these points by $e_1,\cdots ,e_n$. Observe that
\begin{itemize}
    \item[(1)] The fiber of determinant line bundle $\mathcal O(1)$ at $e_i$ has $T$-weight $x_i$,
    \item[(2)] The tangent space at $e_i$ has $T$-weights $x_i/x_j,j\in \{1,\cdots,n\}\backslash\{i\}$,
    \item[(3)] The fiber of $\mathbf L i^*\mathbf R p_*\widetilde{\mathcal F}$ at $e_i$ has the same $T$-weights as $\chi (\mathcal F)$, but the $\mathbb C^{\times}$-action is different, because the fiber $p^{-1}(e_i)$ is is identified with $\Gr_{\GL_n}$ via a translation $g\mapsto z^{\omega_i-\omega_{i-1}}g$ and the new $\mathbb C^{\times}$ acts through the diagonal of $\mathbb C^{\times}_{\mathrm{rotation}}\times T_i$, where $T_i$ is the $i$'th $\mathbb C^{\times}$-component of $T$. In other word, the fiber of $\mathbf L i^*\mathbf R p_*\widetilde{\mathcal F}$ at $e_i$ has the $T\times \mathbb C^{\times}$-weights $$\chi(\mathcal F)(x_1,\cdots,qx_i,\cdots,x_n; q).$$
\end{itemize}
Then \eqref{Convolution} follows from applying localization to $\mathcal O(m)\otimes \mathbf L i^*\mathbf R p_*\widetilde{\mathcal F}$ using three observations made above.
\end{proof}

Comparing \eqref{Convolution} and \eqref{Jing Operator}, we have the following
\begin{corollary}
If $\chi(\mathcal F)\in \bC[p_1,\cdots,p_n,q^{\pm}]\subset K_{\GL_n\times \mathbb C^{\times}}(\mathrm{pt})$, then 
\begin{align}\label{geometric Jing operator}
    \chi(\mathcal O_{\Gr^{\omega_1}}(m)\star \mathcal F)=S^q_m\chi(\mathcal F).
\end{align}
\end{corollary}
From this corollary we see that the operator $\mathcal O_{\Gr^{\omega_1}}(m)\star (-)$ is a geometrization of the Jing operator $S^q_m$. In fact, it extends the domain of $S^q_m$ to $K_{\GL_n\times \mathbb C^{\times}}(\mathrm{pt})$, and negative $m$ is also allowed.

\begin{corollary}\label{corollary: character of convolution}
Let $\mu=(\mu_1,\cdots ,\mu_l)$ be an array of nonnegative integers, then
\begin{align}
    H_{\mu}(x;q)=\chi(\Gr_{\GL_n},\mathcal O_{\Gr^{\omega_1}}(\mu_1)\star\cdots \star\mathcal O_{\Gr^{\omega_1}}(\mu_l)).
\end{align}
\end{corollary}

\begin{proof}
Combine \eqref{geometric Jing operator} with the definition of $H_{\mu}$ in terms of iterative action of $S^q_{\mu_i}$ \eqref{definition: transformed HL}.
\end{proof}

\begin{corollary}\label{corollary: character of O(k)}
Let $\overline{\Gr}^{N\omega_1}$ be the closure of the $\GL_n(\mathscr O)$-orbit through $z^{N\omega_1}$, then 
\begin{align}
    \chi(\overline{\Gr}^{N\omega_1},\mathcal O(k))=H_{(k^N)}(x;q).
\end{align}
Here $(k^N)$ is is the short-hand notation for $N$-tuples of $k$, i.e. $(k^N):=(k,k,\cdots,k)$.
\end{corollary}

\begin{proof}
Let $m:\Gr_{\GL_n}\widetilde{\times} \Gr_{\GL_n}\cdots\widetilde{\times} \Gr_{\GL_n}\to \Gr_{\GL_n}$ be the convolution map of $N$-copies of $\Gr_{\GL_n}$, it follows from the definition of determinate line bundle that there is a $\GL_n(\mathscr O)\rtimes \mathbb C^{\times}$-equivariant isomorphism 
\begin{align}\label{eq:pullback of O(1)}
    m^*\mathcal O(1)\cong\mathcal O(1)\widetilde\boxtimes\cdots \widetilde\boxtimes \mathcal O(1).
\end{align}
It is known that $m(\Gr^{\omega_1}\widetilde{\times} \cdots\widetilde{\times} \Gr^{\omega_1})=\overline{\Gr}^{N\omega_1}$, and it is birational, thus $m$ is a resolution of singularities. It is also known that $\overline{\Gr}^{N \omega_1}$ has rational singularities (this is true for all $G(\mathscr O)$-orbit closure on affine Grassmannian of any reductive group $G$, see \cite[Theorem 2.7]{kamnitzer2014yangians}), therefore $\mathbf Rm_*\mathcal{O}\cong \mathcal O$ and $\mathbf Rm_*m^*\mathcal{O}(k)\cong \mathcal O(k)$, thus
\begin{align*}
    \chi(\overline{\Gr}^{N\omega_1},\mathcal O(k))&=\chi(\Gr^{\omega_1}\widetilde{\times} \cdots\widetilde{\times} \Gr^{\omega_1},m^*\mathcal O(k))\\
    \text{\small by \eqref{eq:pullback of O(1)}}\quad &=\chi(\Gr_{\GL_n},\mathcal O_{\Gr^{\omega_1}}(k)\star\cdots \star\mathcal O_{\Gr^{\omega_1}}(k))\\
    \text{\small by \eqref{geometric Jing operator}}\quad &=H_{(k^N)}(x;q).
\end{align*}
\end{proof}

\section*{Declarations}

\paragraph{Authors’ contributions.} All authors designed the research, performed the research, and wrote the paper. All authors gave their final approval for publication.

\paragraph{Competing Interests.} The authors have no competing interests.

\paragraph{Data Availability.} Data sharing is not applicable to this article, as no datasets were generated.

\bibliographystyle{JHEP}
\bibliography{References}

@article{finkelberg2014quantization, title={{Quantization of Drinfeld Zastava in type $A$}}, volume={16}, ISSN={1435-9863}, url={http://dx.doi.org/10.4171/jems/432}, DOI={10.4171/jems/432}, number={2}, journal={Journal of the European Mathematical Society}, publisher={European Mathematical Society - EMS - Publishing House GmbH}, author={Finkelberg, Michael and Rybnikov, Leonid}, year={2014}, month=jan, pages={235–271} }

@article{jing1991vertex, title={{Vertex operators and Hall-Littlewood symmetric functions}}, volume={87}, ISSN={0001-8708}, url={http://dx.doi.org/10.1016/0001-8708(91)90072-f}, DOI={10.1016/0001-8708(91)90072-f}, number={2}, journal={Advances in Mathematics}, publisher={Elsevier BV}, author={Jing, Naihuan}, year={1991}, month=jun, pages={226–248} }

@article{haiman2003combinatorics, title={{Combinatorics, symmetric functions and Hilbert schemes}}, volume={2002}, ISSN={2164-4829}, url={http://dx.doi.org/10.4310/cdm.2002.v2002.n1.a2}, DOI={10.4310/cdm.2002.v2002.n1.a2}, number={1}, journal={Current Developments in Mathematics}, publisher={International Press of Boston}, author={Haiman, M.}, year={2002}, pages={39–111} }

@misc{zhu2016introduction, title={{An introduction to affine Grassmannians and the geometric Satake equivalaence}}, ISSN={2472-5064}, url={http://dx.doi.org/10.1090/pcms/024/02}, DOI={10.1090/pcms/024/02}, journal={Geometry of Moduli Spaces and Representation Theory}, publisher={American Mathematical Society}, author={Zhu, Xinwen}, year={2017}, month=dec, pages={59–154} }

@article{kamnitzer2014yangians, title={{Yangians and quantizations of slices in the affine Grassmannian}}, volume={8}, ISSN={1937-0652}, url={http://dx.doi.org/10.2140/ant.2014.8.857}, DOI={10.2140/ant.2014.8.857}, number={4}, journal={Algebra \& Number Theory}, publisher={Mathematical Sciences Publishers}, author={Kamnitzer, Joel and Webster, Ben and Weekes, Alex and Yacobi, Oded}, year={2014}, month=aug, pages={857–893} }

@misc{stacks-project,
    shorthand    = {Stacks},
    author       = {{Stacks Project Authors}},
    title        = {\textit{Stacks Project}},
    howpublished = {\url{https://stacks.math.columbia.edu}},
    year         = {2018},
}

@book{kollar1998birational, title={{Birational Geometry of Algebraic Varieties}}, ISBN={9780511662560}, url={http://dx.doi.org/10.1017/cbo9780511662560}, DOI={10.1017/cbo9780511662560}, publisher={Cambridge University Press}, author={Kollár, Janos and Mori, Shigefumi}, year={1998}, month=sep }

@article{braverman2016coulomb, title={{Coulomb branches of $3d$ $\mathcal{N}=4$ quiver gauge theories and slices in the affine Grassmannian}}, volume={23}, ISSN={1095-0753}, url={http://dx.doi.org/10.4310/atmp.2019.v23.n1.a3}, DOI={10.4310/atmp.2019.v23.n1.a3}, number={1}, journal={Advances in Theoretical and Mathematical Physics}, publisher={International Press of Boston}, author={Braverman, Alexander and Finkelberg, Michael and Nakajima, Hiraku}, year={2019}, pages={75–166} }

@article{tHooft199310,
    author = "'t Hooft, Gerard",
    title = "{Dimensional Reduction in Quantum Gravity}",
    eprint = "gr-qc/9310026",
    archivePrefix = "arXiv",
    reportNumber = "THU-93-26",
    journal = "Conf. Proc. C",
    volume = "930308",
    pages = "284--296",
    year = "1993"
}

@article{Susskind199409,
    author = "Susskind, Leonard",
    title = "{The World as a Hologram}",
    eprint = "hep-th/9409089",
    archivePrefix = "arXiv",
    reportNumber = "SU-ITP-94-33",
    doi = "10.1063/1.531249",
    journal = "J. Math. Phys.",
    volume = "36",
    pages = "6377--6396",
    year = "1995",
    month = "Nov"
}

@article{Maldacena199711,
    author = "Maldacena, Juan Martin",
    title = "{The Large-N Limit of Superconformal Field Theories and Supergravity}",
    eprint = "hep-th/9711200",
    archivePrefix = "arXiv",
    reportNumber = "HUTP-97-A097, HUTP-98-A097",
    doi = "10.4310/ATMP.1998.v2.n2.a1",
    journal = "Adv. Theor. Math. Phys.",
    volume = "2",
    pages = "231--252",
    year = "1998",
    month = "Jan"
}

@article{Costello201610,
    author = "Costello, Kevin",
    title = "{M-Theory in the Omega-Background and 5-Dimensional Non-Commutative Gauge Theory}",
    eprint = "1610.04144",
    archivePrefix = "arXiv",
    primaryClass = "hep-th",
    month = "10",
    year = "2016"
}

@article{Costello201705,
    author = "Costello, Kevin",
    title = "{Holography and Koszul Duality: The Example of the $M2$-Brane}",
    eprint = "1705.02500",
    archivePrefix = "arXiv",
    primaryClass = "hep-th",
    month = "5",
    year = "2017"
}

@article{IshtiaqueMoosavianZhou201809,
    author = "Ishtiaque, Nafiz and Moosavian, Seyed Faroogh and Zhou, Yehao",
    title = "{Topological Holography: The Example of the D2-D4 Brane System}",
    eprint = "1809.00372",
    archivePrefix = "arXiv",
    primaryClass = "hep-th",
    doi = "10.21468/SciPostPhys.9.2.017",
    journal = "SciPost Phys.",
    volume = "9",
    number = "2",
    pages = "017",
    year = "2020"
}

@article{CostelloGaiotto1821,
    author = "Costello, Kevin and Gaiotto, Davide",
    title = "{Twisted Holography}",
    eprint = "1812.09257",
    archivePrefix = "arXiv",
    primaryClass = "hep-th",
    doi = "10.1007/JHEP01(2025)087",
    journal = "J. High Energy Phys.",
    volume = "01",
    pages = "087",
    year = "2025",
    month = "Jan"
}

@article{GaiottoOh201907,
    author = "Gaiotto, Davide and Oh, Jihwan",
    title = "{Aspects of $\Omega$-Deformed M-Theory}",
    eprint = "1907.06495",
    archivePrefix = "arXiv",
    primaryClass = "hep-th",
    doi = "10.1007/JHEP12(2024)184",
    journal = "J. High Energy Phys.",
    volume = "12",
    pages = "184",
    year = "2024",
    month = "Dec"
}

@article{RaghavendranYoo201910,
    author = "Raghavendran, Surya and Yoo, Philsang",
    title = "{Twisted S-Duality}",
    eprint = "1910.13653",
    archivePrefix = "arXiv",
    primaryClass = "math-ph",
    month = "10",
    year = "2019"
}

@article{LiTroost201911,
    author = "Li, Songyuan and Troost, Jan",
    title = "{Pure and Twisted Holography}",
    eprint = "1911.06019",
    archivePrefix = "arXiv",
    primaryClass = "hep-th",
    doi = "10.1007/JHEP03(2020)144",
    journal = "J. High Energy Phys.",
    volume = "03",
    pages = "144",
    year = "2020",
    month = "Mar"
}

@article{CostelloPaquette202001,
    author = "Costello, Kevin and Paquette, Natalie M.",
    title = "{Twisted Supergravity and Koszul Duality: A Case Study in AdS$_3$}",
    eprint = "2001.02177",
    archivePrefix = "arXiv",
    primaryClass = "hep-th",
    doi = "10.1007/s00220-021-04065-3",
    journal = "Commun. Math. Phys.",
    volume = "384",
    number = "1",
    pages = "279--339",
    year = "2021",
    month = "Apr"
}

@article{OhZhou202002,
    author = "Oh, Jihwan and Zhou, Yehao",
    title = "{Feynman Diagrams and $\Omega$-Deformed M-Theory}",
    eprint = "2002.07343",
    archivePrefix = "arXiv",
    primaryClass = "hep-th",
    doi = "10.21468/SciPostPhys.10.2.029",
    journal = "SciPost Phys.",
    volume = "10",
    number = "2",
    pages = "029",
    year = "2021",
    month = "Feb"
}

@article{GaiottoAbajian202004,
    author = "Gaiotto, Davide and Abajian, Jacob",
    title = "{Twisted M2 Brane Holography and Sphere Correlation Functions}",
    eprint = "2004.13810",
    archivePrefix = "arXiv",
    primaryClass = "hep-th",
    doi = "10.1007/JHEP03(2025)195",
    journal = "J. High Energy Phys.",
    volume = "03",
    pages = "195",
    year = "2025",
    month = "Mar"
}

@article{LiTroost202005,
    author = "Li, Songyuan and Troost, Jan",
    title = "{Twisted String Theory in Anti-de Sitter Space}",
    eprint = "2005.13817",
    archivePrefix = "arXiv",
    primaryClass = "hep-th",
    doi = "10.1007/JHEP11(2020)047",
    journal = "J. High Energy Phys.",
    volume = "11",
    pages = "047",
    year = "2020",
    month = "Nov"
}

@article{OhZhou202103,
    author = "Oh, Jihwan and Zhou, Yehao",
    title = "{Twisted Holography of Defect Fusions}",
    eprint = "2103.00963",
    archivePrefix = "arXiv",
    primaryClass = "hep-th",
    doi = "10.21468/SciPostPhys.10.5.105",
    journal = "SciPost Phys.",
    volume = "10",
    number = "5",
    pages = "105",
    year = "2021",
    month = "May"
}

@article{OhZhou202105,
    author = "Oh, Jihwan and Zhou, Yehao",
    title = "{A Domain Wall in Twisted M-Theory}",
    eprint = "2105.09537",
    archivePrefix = "arXiv",
    primaryClass = "hep-th",
    doi = "10.21468/SciPostPhys.11.4.077",
    journal = "SciPost Phys.",
    volume = "11",
    pages = "077",
    year = "2021",
    month = "Oct"
}

@article{BudzikGaiotto202106,
    author = "Budzik, Kasia and Gaiotto, Davide",
    title = "{Giant Gravitons in Twisted Holography}",
    eprint = "2106.14859",
    archivePrefix = "arXiv",
    primaryClass = "hep-th",
    doi = "10.1007/JHEP10(2023)131",
    journal = "J. High Energy Phys.",
    volume = "10",
    pages = "131",
    year = "2023",
    month = "Oct"
}

@article{PaquetteWilliams202110,
    author = "Paquette, Natalie M. and Williams, Brian R.",
    title = "{Koszul Duality in Quantum Field Theory}",
    eprint = "2110.10257",
    archivePrefix = "arXiv",
    primaryClass = "hep-th",
    doi = "10.5802/cml.88",
    journal = "Conflu. Math.",
    volume = "14",
    number = "2",
    pages = "87--138",
    year = "2023"
}

@article{CostelloWilliams202110,
    author = "Costello, Kevin and Williams, Brian R.",
    title = "{Twisted Heterotic/Type I Duality}",
    eprint = "2110.14616",
    archivePrefix = "arXiv",
    primaryClass = "hep-th",
    month = "10",
    year = "2021"
}

@article{CostelloLi201905,
    author = "Costello, Kevin and Li, Si",
    title = "{Anomaly Cancellation in the Topological String}",
    eprint = "1905.09269",
    archivePrefix = "arXiv",
    primaryClass = "hep-th",
    doi = "10.4310/ATMP.2020.v24.n7.a2",
    journal = "Adv. Theor. Math. Phys.",
    volume = "24",
    number = "7",
    pages = "1723--1771",
    year = "2020",
    month = "Sep"
}

@article{CostelloLi201606,
    author = "Costello, Kevin and Li, Si",
    title = "{Twisted Supergravity and its Quantization}",
    eprint = "1606.00365",
    archivePrefix = "arXiv",
    primaryClass = "hep-th",
    month = "6",
    year = "2016"
}

@article{Krylov_2021, title={Almost dominant generalized slices and convolution diagrams over them}, volume={392}, ISSN={0001-8708}, url={http://dx.doi.org/10.1016/j.aim.2021.108034}, DOI={10.1016/j.aim.2021.108034}, journal={Advances in Mathematics}, publisher={Elsevier BV}, author={Krylov, Vasily and Perunov, Ivan}, year={2021}, month=dec, pages={108034} }

@article{Costello201303,
    author = "Costello, Kevin",
    title = "{Supersymmetric Gauge Theory and the Yangian}",
    eprint = "1303.2632",
    archivePrefix = "arXiv",
    primaryClass = "hep-th",
    month = "3",
    year = "2013"
}

@article{GaiottoLee202109,
    author = "Gaiotto, Davide and Lee, Ji Hoon",
    title = "{The Giant Graviton Expansion}",
    eprint = "2109.02545",
    archivePrefix = "arXiv",
    primaryClass = "hep-th",
    doi = "10.1007/JHEP08(2024)025",
    journal = "J. High Energy Phys.",
    volume = "08",
    pages = "025",
    year = "2024",
    month = "Aug"
}

@book{bruns2006determinantal, title={{Determinantal Rings}}, ISBN={9783540392743}, ISSN={1617-9692}, url={http://dx.doi.org/10.1007/bfb0080378}, DOI={10.1007/bfb0080378}, journal={Lecture Notes in Mathematics}, publisher={Springer Berlin Heidelberg}, author={Bruns, Winfried and Vetter, Udo}, year={1988} }

@article{king1994moduli, title={Moduli of representations of finite dimensional algebras}, volume={45}, ISSN={1464-3847}, url={http://dx.doi.org/10.1093/qmath/45.4.515}, DOI={10.1093/qmath/45.4.515}, number={4}, journal={The Quarterly Journal of Mathematics}, publisher={Oxford University Press (OUP)}, author={King, A. D.}, year={1994}, pages={515–530} }

@article{Halpern_Leistner_2014, title={{The derived category of a GIT quotient}}, volume={28}, ISSN={1088-6834}, url={http://dx.doi.org/10.1090/s0894-0347-2014-00815-8}, DOI={10.1090/s0894-0347-2014-00815-8}, number={3}, journal={Journal of the American Mathematical Society}, publisher={American Mathematical Society (AMS)}, author={Halpern-Leistner, Daniel}, year={2014}, month=oct, pages={871–912} }

@book{Slodowy_1980, title={{Simple Singularities and Simple Algebraic Groups}}, ISBN={9783540381914}, ISSN={1617-9692}, url={http://dx.doi.org/10.1007/bfb0090294}, DOI={10.1007/bfb0090294}, journal={Lecture Notes in Mathematics}, publisher={Springer Berlin Heidelberg}, author={Slodowy, Peter}, year={1980} }

@article{DedushenkoGaiotto202009b,
    author = "Dedushenko, Mykola and Gaiotto, Davide",
    title = "{Correlators on the Wall and $\mfk{sl}_n$ Spin Chain}",
    eprint = "2009.11198",
    archivePrefix = "arXiv",
    primaryClass = "hep-th",
    doi = "10.1063/5.0073021",
    journal = "J. Math. Phys.",
    volume = "63",
    number = "9",
    pages = "092301",
    year = "2022",
    month = "Sep"
}

@incollection{molev2003yangians, 
    title="\href{https://www.sciencedirect.com/science/article/pii/S1570795403800761}{Yangians and their Applications}",
    @doi={10.1016/s1570-7954(03)80076-1}, 
    booktitle={Handbook of Algebra}, 
    publisher={Elsevier}, 
    author={Molev, A.I.}, 
    year={2003}, 
    pages={907–959} 
}

@article{braverman2017ring, title={{Ring objects in the equivariant derived Satake category arising from Coulomb branches}}, volume={23}, ISSN={1095-0753}, url={http://dx.doi.org/10.4310/atmp.2019.v23.n2.a1}, DOI={10.4310/atmp.2019.v23.n2.a1}, number={2}, journal={Advances in Theoretical and Mathematical Physics}, publisher={International Press of Boston}, author={Braverman, Alexander and Finkelberg, Michael and Nakajima, Hiraku}, year={2019}, pages={253–344} }

@article{CostelloGaiottoYagi202103,
    author = "Costello, Kevin and Gaiotto, Davide and Yagi, Junya",
    title = "{Q-operators are \textquoteright{}t Hooft Lines}",
    eprint = "2103.01835",
    archivePrefix = "arXiv",
    primaryClass = "hep-th",
    doi = "10.1007/JHEP11(2024)003",
    journal = "J. High Energy Phys.",
    volume = "11",
    pages = "003",
    year = "2024",
    month = "Nov"
}

@article{Witten199802,
    author = "Witten, Edward",
    title = "{Anti-de Sitter Space and Holography}",
    eprint = "hep-th/9802150",
    archivePrefix = "arXiv",
    reportNumber = "IASSNS-HEP-98-15",
    doi = "10.4310/ATMP.1998.v2.n2.a2",
    journal = "Adv. Theor. Math. Phys.",
    volume = "2",
    pages = "253--291",
    year = "1998",
    month = "Jan"
}

@article{Kamnitzer_2017, title={{Reducedness of affine Grassmannian slices in type A}}, volume={146}, ISSN={1088-6826}, url={http://dx.doi.org/10.1090/proc/13850}, DOI={10.1090/proc/13850}, number={2}, journal={Proceedings of the American Mathematical Society}, publisher={American Mathematical Society (AMS)}, author={Kamnitzer, Joel and Muthiah, Dinakar and Weekes, Alex and Yacobi, Oded}, year={2017}, month=nov, pages={861–874} }

@book{ChariPressley199510,
  title="\href{https://www.cambridge.org/us/universitypress/subjects/mathematics/algebra/guide-quantum-groups}{A Guide to Quantum Groups}",
  author={Chari, Vyjayanthi and Pressley, Andrew N},
  year="1995",
    month = "Oct",
  publisher={Cambridge university press}
}

@article{Kontsevich199709,
    author = "Kontsevich, Maxim",
    title = "{Deformation Quantization of Poisson Manifolds. 1.}",
    eprint = "q-alg/9709040",
    archivePrefix = "arXiv",
    doi = "10.1023/B:MATH.0000027508.00421.bf",
    journal = "Lett. Math. Phys.",
    volume = "66",
    pages = "157--216",
    year = "2003",
    month = "Dec"
}

@article{Moyal194901,
    author = "Moyal, J. E.",
    title = "{Quantum Mechanics as a Statistical Theory}",
    doi = "10.1017/S0305004100000487",
    journal = "Proc. Cambridge Phil. Soc.",
    volume = "45",
    pages = "99--124",
    year = "1949",
    month = "Jan"
}

@article{Fedosov199402,
    author = "Fedosov, Boris V.",
    title = "{A Simple Geometrical Construction of Deformation Quantization}",
    journal = "J. Diff. Geom.",
    volume = "40",
    number = "2",
    pages = "213--238",
    year = "1994",
    month = "Feb"
}

@book{Fedosov1995,
    author = "Fedosov, B. V.",
    title = "\href{https://scispace.com/pdf/deformation-quantization-and-index-theory-1nynsaz02q.pdf}{Deformation Quantization and Index Theory}",
    year = "1995",
    publisher = "John Wiley \& Sons",
}

@article{Berezin197506,
    author = "Berezin, F. A.",
    title = "{General Concept of Quantization}",
    reportNumber = "ITF-74-20E-MC",
    doi = "10.1007/BF01609397",
    journal = "Commun. Math. Phys.",
    volume = "40",
    pages = "153--174",
    year = "1975",
    month = "Jun"
}

@article{Berezin197410, 
    title="{Quantization}", 
    volume={8}, 
    doi={10.1070/im1974v008n05abeh002140}, 
    number={5}, 
    journal={Math. USSR Izv.}, 
    publisher={Steklov Mathematical Institute}, 
    author={Berezin, F A}, 
    year={1974}, 
    month="Oct", 
    pages="1109--1165" 
}

@article{Berezin197504, 
    title="{Quantization in Complex Symmetric Spaces}",
    volume={9}, 
    doi={10.1070/im1975v009n02abeh001480}, 
    number={2}, 
    journal={Math. USSR Izv.}, 
    publisher={Steklov Mathematical Institute}, 
    author={Berezin, F. A.}, 
    year={1975}, 
    month="Apr", 
    pages="341--379" 
}

@article{BayenFlatoFronsdalLichnerowiczSternheimer197803I, 
    title="{Deformation Theory and Quantization. I. Deformations of Symplectic Structures}", 
    volume={111}, 
    doi={10.1016/0003-4916(78)90224-5}, 
    number={1}, 
    journal={Ann. Phys.}, 
    publisher={Elsevier BV}, 
    author={Bayen, F and Flato, M and Fronsdal, C and Lichnerowicz, A and Sternheimer, D}, 
    year={1978}, 
    month="Mar", 
    pages={61--110} 
}

@article{BayenFlatoFronsdalLichnerowiczSternheimer197803II, 
    title="{Deformation Theory and Quantization. II. Physical Applications}", 
    volume={111}, 
    doi={10.1016/0003-4916(78)90225-7}, 
    number={1}, 
    journal={Ann. Phys.}, 
    publisher={Elsevier BV}, 
    author={Bayen, F and Flato, M and Fronsdal, C and Lichnerowicz, A and Sternheimer, D}, 
    year={1978}, 
    month="Mar", 
    pages="111--151" 
}

@article{DolgushevIsaevLyakhovichSharapov200206,
    author = "Dolgushev, V. A. and Isaev, A. P. and Lyakhovich, S. L. and Sharapov, A. A.",
    title = "{On the Fedosov Deformation Quantization Beyond the Regular Poisson Manifolds}",
    eprint = "hep-th/0206039",
    archivePrefix = "arXiv",
    reportNumber = "ITEP-TH-30-02",
    doi = "10.1016/S0550-3213(02)00763-0",
    journal = "Nucl. Phys. B",
    volume = "645",
    pages = "457--476",
    year = "2002",
    month = "Dec"
}

@article{GomisPasserini200604,
    author = "Gomis, Jaume and Passerini, Filippo",
    title = "{Holographic Wilson Loops}",
    eprint = "hep-th/0604007",
    archivePrefix = "arXiv",
    reportNumber = "KUL-TF-06-11",
    doi = "10.1088/1126-6708/2006/08/074",
    journal = "J. High Energy Phys.",
    volume = "08",
    pages = "074",
    year = "2006",
    month = "Aug"
}

@article{GomisPasserini200612,
    author = "Gomis, Jaume and Passerini, Filippo",
    title = "{Wilson Loops as D3-Branes}",
    eprint = "hep-th/0612022",
    archivePrefix = "arXiv",
    doi = "10.1088/1126-6708/2007/01/097",
    journal = "J. High Energy Phys.",
    volume = "01",
    pages = "097",
    year = "2007",
    month = "January",
}

\end{document}